\documentclass[a4paper,cleveref,autoref, thm-restate, numberwithinsect]{lipics-v2021}
\newif\iflong
\newif\ifshort

\longtrue

\iflong
\else
\shorttrue
\fi

\usepackage{macros}
\usepackage{multicol}
\usepackage[disable]{todonotes}
\usepackage{comment}
\usepackage{url}            %
\usepackage{booktabs}       %
\usepackage{amsfonts}       %
\usepackage{hyperref}
\usepackage[linesnumbered,lined,boxed,noend]{algorithm2e}
\usepackage{mathtools}
\usepackage{pdfpages}

\usepackage{boxedminipage}
\usepackage{xspace}

\usepackage{amsmath, amssymb}
\usepackage{amsthm}

\newtheorem{longtheorem}{Theorem}
\newtheorem{longlemma}[longtheorem]{Lemma}
\newtheorem{longcorollary}[longtheorem]{Corollary}

\theoremstyle{remark}
\newtheorem{longclaim}[longtheorem]{Claim}

\hideLIPIcs
\nolinenumbers
\def\ie{\emph{i.e.}\xspace}
\def\eg{\emph{e.g.}\xspace}

\title{Space-Efficient Parameterized Algorithms on Graphs of Low Shrubdepth}
\titlerunning{Space-Efficient Parameterized Algorithms on Graphs of Low Shrubdepth}
 \author{Benjamin Bergougnoux}{Institute of Informatics, University of Warsaw, Poland}{benjamin.bergougnoux@mimuw.edu.pl}{}{}
 \author{Vera Chekan}{Humboldt-Universität zu Berlin, Germany}{vera.chekan@informatik.hu-berlin.de}{https://orcid.org/0000-0002-6165-1566}{Supported by the DFG Research Training Group 2434 “Facets of Complexity.”}
 \author{Robert Ganian}{Algorithms and Complexity Group, TU Wien, Vienna, Austria}{rganian@gmail.com}{https://orcid.org/0000-0002-7762-8045}{Project No. Y1329 of the Austrian Science Fund (FWF), WWTF Project ICT22-029.}
\author{Mamadou Moustapha {Kant\'e}}{Université Clermont Auvergne, Clermont Auvergne INP, LIMOS, CNRS, Clermont-Ferrand, France}{mamadou.kante@uca.fr}{https://orcid.org/0000-0003-1838-7744}{Supported by the French National Research Agency (ANR-18-CE40-0025-01 and ANR-20-CE48-0002).}
\author{Matthias Mnich}{Hamburg University of Technology, Institute for Algorithms and Complexity, Hamburg, Germany}{matthias.mnich@tuhh.de}{https://orcid.org/0000-0002-4721-5354}{}
 \author{Sang-il Oum}{Discrete Mathematics Group, Institute for Basic Science (IBS), Daejeon, Korea \and 
Department of Mathematical Sciences, KAIST, Daejeon, Korea}{sangil@ibs.re.kr}{https://orcid.org/0000-0002-6889-7286}{Supported by the Institute for Basic Science (IBS-R029-C1).}
\author{Micha{\l} Pilipczuk}{Institute of Informatics, University of Warsaw, Poland}{michal.pilipczuk@mimuw.edu.pl}{}{\flag[0.17\textwidth]{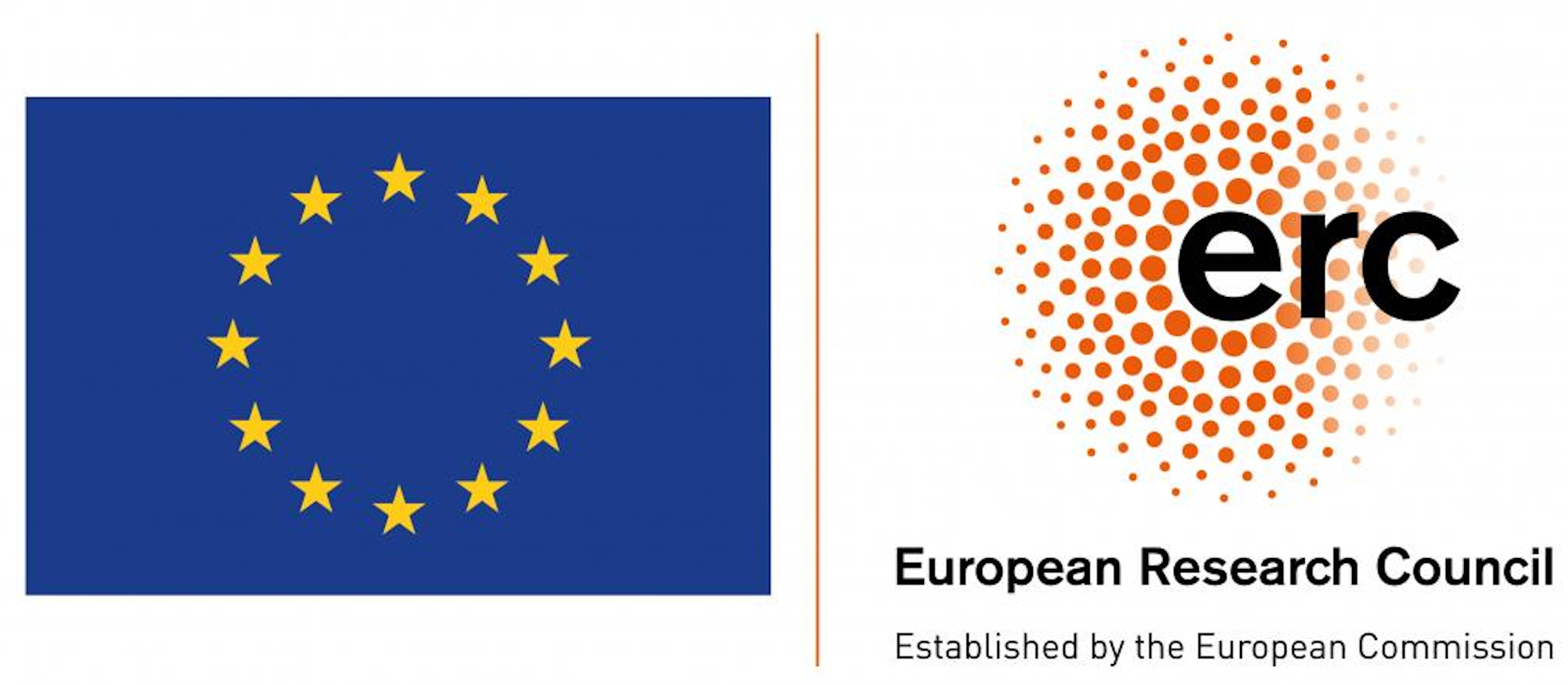}This work is a part of project BOBR that has received funding from the European Research Council (ERC) under the European Union’s Horizon 2020 research and innovation programme (grant agreement No. 948057).}
 \author{Erik Jan van Leeuwen}{Dept.\ Information and Computing Sciences, Utrecht University, The Netherlands}{e.j.vanleeuwen@uu.nl}{https://orcid.org/0000-0001-5240-7257}{}
\authorrunning{Bergougnoux, Chekan, Ganian, Kant\'e, Mnich, Oum, Pilipczuk, van Leeuwen} %

\Copyright{Benjamin Bergougnoux, Vera Chekan, Robert Ganian, Mamadou M. Kant\'e, Matthias Mnich, Sang-il Oum, Micha{\l} Pilipczuk, and Eric Jan van Leeuwen}

\acknowledgements{This work was initiated at the \emph{Graph Decompositions: Small Width, Big Challenges} workshop held at the Lorentz Center in Leiden, The Netherlands, in 2022.}


\begin{document}
	\maketitle

\begin{abstract}
    Dynamic programming on various graph decompositions is one of the most fundamental techniques used in parameterized complexity. Unfortunately, even if we consider concepts as simple as path or tree decompositions, such dynamic programming uses space that is exponential in the decomposition's width, and there are good reasons to believe that this is necessary. However, it has been shown that in graphs of low treedepth it is possible to design algorithms which achieve polynomial space complexity without requiring worse time complexity than their counterparts working on tree decompositions of bounded width. Here, {\em{treedepth}} is a graph parameter that, intuitively speaking, takes into account both the depth and the width of a tree decomposition of the graph, rather than the width alone.

    Motivated by the above, we consider graphs that admit clique expressions with bounded depth and label count, or equivalently, graphs of low shrubdepth. Here, shrubdepth is a bounded-depth analogue of cliquewidth, in the same way as treedepth is a bounded-depth analogue of treewidth. We show that also in this setting, bounding the depth of the decomposition is a deciding factor for improving the space complexity. More precisely, we prove that on $n$-vertex graphs equipped with a tree-model (a decomposition notion underlying shrubdepth) 
    of depth $d$ and using $k$ labels,
    \begin{itemize}
        \item {\sc{Independent Set}} can be solved in time $2^{\Oh(dk)}\cdot n^{\Oh(1)}$ using $\Oh(dk^2\log n)$ space;
        \item {\sc{Max Cut}} can be solved in time $n^{\Oh(dk)}$ using $\Oh(dk\log n)$ space; and
        \item {\sc{Dominating Set}} can be solved in time $2^{\Oh(dk)}\cdot n^{\Oh(1)}$ using $n^{\Oh(1)}$ space via a randomized algorithm.
    \end{itemize}
    We also establish a lower bound, conditional on a certain assumption about the complexity of {\sc{Longest Common Subsequence}}, which shows that at least in the case of {\sc{Independent Set}} the exponent of the parametric factor in the time complexity has to grow with $d$ if one wishes to keep the space complexity polynomial. 
\end{abstract}

\section{Introduction}

\noindent \textbf{Treewidth and Treedepth.} \quad Dynamic programming on graph decompositions is a fundamental method in the design of parameterized algorithms. Among various decomposition notions, {\em{tree decompositions}}, which underly the parameter {\em{treewidth}}, are perhaps the most widely used; see e.g.~\cite{cygan2015parameterized,DowneyF2013} for an introduction. A tree decomposition of a graph $G$ of width~$k$ provides a way to ``sweep'' $G$ while keeping track of at most $k+1$ ``interface vertices'' at a time. This can be used for dynamic programming: during the sweep, the algorithm maintains a set of representative partial solutions within the part already swept, one for each possible behavior of a partial solution on the interface vertices. Thus, the width of the decomposition is the key factor influencing the number of partial solutions that need to be stored.

In a vast majority of applications, this number of different partial solutions depends (at least) exponentially on the width $k$ of the decomposition, which often leads to time complexity of the form $f(k)\cdot n^{\Oh(1)}$ for an exponential function $f$. This should not be surprising, as most problems where this technique is used are ${\mathsf{NP}}$-hard. Unfortunately, the space complexity---which often appears to be the true bottleneck in practice ---is also exponential. There is a simple tradeoff trick, first observed by Lokshtanov et al.~\cite{LokshtanovMS11}, which can often be used to reduce the space complexity to polynomial at the cost of increasing the time complexity. For instance, {\sc{Independent Set}} can be solved in $2^k\cdot n^{\Oh(1)}$ time and using $2^k\cdot n^{\Oh(1)}$ space on an $n$-vertex graph equipped with a width-$k$ tree decomposition via dynamic programming~\cite{Furer17};
combining this algorithm with a simple recursive Divide\&Conquer scheme yields an algorithm with running time $2^{\Oh(k^2)}\cdot n^{\Oh(1)}$ and space complexity $n^{\Oh(1)}$.

Allender et al.~\cite{AllenderCLPT14} and then Pilipczuk and Wrochna~\cite{PilipczukW18} studied the question whether the loss on the time complexity is necessary if one wants to achieve polynomial space complexity in the context of dynamic programming on tree decompositions. While the formal formulation of their results is somewhat technical and complicated, the take-away message is the following: there are good complexity-theoretical reasons to believe that even in the simpler setting of path decompositions, one cannot achieve algorithms with polynomial space complexity whose running times asymptotically match the running times of their exponential-space counterparts. We refer to the works~\cite{AllenderCLPT14,PilipczukW18} for further details.

However, starting with the work of F\"urer and Yu~\cite{FurerY17}, a long line of advances~\cite{HegerfeldK20,NadaraPS22,NederlofPSW20,PilipczukW18} showed that bounding the {\em{depth}}, rather than the width, of a decomposition leads to the possibility of designing algorithms that are both time- and space-efficient. To this end, we consider the {\em{treedepth}} of a graph $G$, which is the least possible depth of an {\em{elimination forest}}: a forest $F$ on the vertex set of $G$ such that every two vertices adjacent in $G$ are in the ancestor/descendant relation in $F$.
An elimination forest of depth $d$ can be regarded as a tree decomposition of depth $d$, and thus treedepth is the bounded-depth analogue of treewidth. As shown in~\cite{FurerY17,HegerfeldK20,NederlofPSW20,PilipczukW18}, for many classic problems, including {\sc{3-Coloring}}, {\sc{Independent Set}}, {\sc{Dominating Set}},
and {\sc{Hamiltonicity}}, it is possible to design algorithms with running time $2^{\Oh(d)}\cdot n^{\Oh(1)}$ and polynomial space complexity, assuming the graph is supplied with an elimination forest of depth $d$. In certain cases, the space complexity can even be as low as $\Oh(d+\log n)$ or $\Oh(d\log n)$~\cite{PilipczukW18}. Typically, the main idea is to reformulate the classic bottom-up dynamic programming approach so that it can be replaced by a simple top-down recursion. This reformulation is by no means easy---it often involves a highly non-trivial use of algebraic transforms or other tools of algebraic flavor, such as inclusion-exclusion~branching.

\smallskip
\noindent \textbf{Cliquewidth and Shrubdepth.}\quad
In this work, we are interested in the parameter {\em{cliquewidth}} and its low-depth counterpart: {\em{shrubdepth}}. While treewidth applies only to sparse graphs, cliquewidth is a notion of tree-likeness suited for dense graphs as well. The decompositions underlying cliquewidth are called {\em{clique expressions}}~\cite{CourcelleMR00}. A clique expression is a term operating over {\em{$k$-labelled graphs}}---graphs where every vertex is assigned one of $k$ labels---and the allowed operations are: (i) apply any renaming function to the labels; (ii) make a complete bipartite graph between two given labels; and (iii) take the disjoint union of two $k$-labelled graphs. Then the cliquewidth of $G$ is the least number of labels using which (some labelling of) $G$ can be constructed. Similarly to treewidth, dynamic programming over clique expressions can be used to solve a wide range of problems, in particular all problems expressible in $\mathsf{MSO}_1$ logic, in $\mathsf{FPT}$ time when parameterized by cliquewidth. Furthermore, while several problems involving edge selection or edge counting, such as {\sc{Hamiltonicity}} or {\sc{Max Cut}}, remain $\mathsf{W}[1]$-hard under the cliquewidth parameterization~\cite{FominGLS10,FominGLS14}, standard dynamic programming still allows us to solve them in $\mathsf{XP}$ time. In this sense, clique-width can be seen as the ``least restrictive'' general-purpose graph parameter which allows for efficient dynamic programming algorithms where the decompositions can also be computed efficiently~\cite{FominK22}.
Nevertheless, since the cliquewidth of a graph is at least as large as its linear cliquewidth, which in turn is as large as its pathwidth, the lower bounds of Allender et al.~\cite{AllenderCLPT14} and of Pilipczuk and Wrochna~\cite{PilipczukW18} carry over to the cliquewidth setting. Hence, reducing the space complexity to polynomial requires a sacrifice in the time complexity.

Shrubdepth, introduced by Ganian et al.~\cite{GanianHNOM19}, is a variant of cliquewidth where we stipulate the decomposition to have bounded depth. This necessitates altering the set of operations used in clique expressions in order to allow taking disjoint unions of multiple graphs as a single operation. In this context, we call the decompositions used for shrubdepth {\em{$(d,k)$-tree-models}}, where $d$ stands for the depth and $k$ for the number of labels used; a formal definition is provided in Section~\ref{sec:prelims}. 
Shrubdepth appears to be a notion of depth that is sound from the model-theoretic perspective, is $\mathsf{FPT}$-time computable~\cite{GajarskyK20}, and has become an important concept in the logic-based theory of well-structured dense graphs~\cite{Dreier21,DreierGKPT22,GajarskyKNMPST20,GanianHNOM19,OhlmannPPT23,OSSONADEMENDEZ2022103660}.

Since shrubdepth is a bounded-depth analogue of cliquewidth in the same way as treedepth is a bounded-depth analogue of treewidth, it is natural to ask whether for graphs from classes of bounded shrubdepth, or more concretely, for graphs admitting $(d,k)$-tree-models where both $d$ and $k$ are considered parameters, one can design space-efficient $\mathsf{FPT}$ algorithms. Exploring this question is the topic of this work.

\smallskip
\noindent \textbf{Our contribution.}\quad
We consider three example problems: {\sc{Independent Set}}, {\sc{Max Cut}}, and {\sc{Dominating Set}}. For each of them we show that on graphs supplied with $(d,k)$-tree-models where $d=\Oh(1)$, one can design space-efficient fixed-parameter algorithms whose running times asymptotically match the running times of their exponential-space counterparts working on general clique expressions. While we focus on the three problems mentioned above for concreteness, we in fact provide a more general algebraic framework, inspired by the work on the treedepth parameterization~\cite{FurerY17,HegerfeldK20,NadaraPS22,NederlofPSW20,PilipczukW18}, that can be applied to a wider range of problems. Once the depth $d$ is not considered a constant, the running times of our algorithms increase with $d$. To mitigate this concern, we give a conditional lower bound showing that this is likely to be necessary if one wishes to keep the space complexity~polynomial. 

Recall that standard dynamic programming solves the {\sc{Independent Set}} problem in time $2^k\cdot n^{\Oh(1)}$ and space $2^k\cdot n^{\Oh(1)}$ on a graph constructed by a clique expression of width~$k$~\cite{Furer17}. Our first contribution is to show that on graphs with $(d,k)$-tree-models, the space complexity can be reduced to as low as $\Oh(dk^2\cdot \log n)$ at the cost of allowing time complexity $2^{\Oh(dk)}\cdot n^{\Oh(1)}$. In fact, we tackle the more general problem of computing the independent set polynomial.

\iflong
\begin{restatable}{theorem}{spaceIS}
\label{thm:IS}
There is an algorithm which takes as input an $n$-vertex graph $G$ along with a $(d,k)$-tree model of $G$, runs in time $2^{\Oh(k d)}\cdot n^{\Oh(1)}$ and uses at most $\Oh(dk^2 \log n)$ space, and computes the independent set polynomial of $G$.
\end{restatable}
\fi
\ifshort
\begin{restatable}{theorem}{spaceIS}
\label{thm:IS}
There is an algorithm which takes as input an $n$-vertex graph $G$ along with a $(d,k)$-tree model of $G$, runs in time $2^{\Oh(k d)}\cdot n^{\Oh(1)}$ and uses at most $\Oh(dk^2 \log n)$ space, and computes the independent set polynomial of $G$.
\end{restatable}
\fi

The idea of the proof of Theorem~\ref{thm:IS} is to reorganize the computation of the standard bottom-up dynamic programming by applying the zeta-transform to the computed tables. This allows a radical simplification of the way a dynamic programming table for a node is computed from the tables of its children, so that the whole dynamic programming can be replaced by top-down recursion. Applying just this yields an algorithm with space polynomial in $n$. We reduce space to $\bigoh(dk^2 \log n)$ by computing the result modulo several small primes, and using space-efficient Chinese remaindering. This is inspired by the algorithm for {\sc{Dominating Set}} on graphs of small treedepth of Pilipczuk and Wrochna~\cite{PilipczukW18}.

In fact, the technique used to prove Theorem~\ref{thm:IS} is much more general and can be used to tackle all coloring-like problems of local character. We formalize those under a single umbrella by solving the problem of counting List $H$-homomorphisms (for an arbitrary but fixed pattern graph $H$), for which we provide an algorithm with the same complexity guarantees as those of Theorem~\ref{thm:IS}. The concrete problems captured by this framework include, e.g., {\sc{Odd Cycle Transveral}} and {\sc{$q$-Coloring}} for a fixed constant $q$\ifshort~(details in the appendix).\fi\iflong; details are provided in Section~\ref{subsec:list-homomorphisms}.\fi

Next, we turn our attention to the {\sc{Max Cut}} problem. This problem is $\mathsf{W}[1]$-hard when parameterized by cliquewidth, but it admits a simple $n^{\Oh(k)}$-time algorithm on $n$-vertex graphs provided with clique expressions of width $k$~\cite{FominGLS14}. Our second contribution is a space-efficient counterpart of this result for graphs equipped with bounded-depth tree-models.

\iflong
\begin{restatable}{theorem}{spaceMC}
\label{thm:maxcut}
There is an algorithm which takes as input an $n$-vertex graph $G$ along with a $(d,k)$-tree model of $G$, runs in time $n^{\Oh(dk)}$ and uses at most $\Oh(dk \log n)$ space, and solves the {\sc{Max Cut}} problem in $G$.
\end{restatable}
\fi
\ifshort
\begin{restatable}{theorem}{spaceMC}
\label{thm:maxcut}
There is an algorithm which takes as input an $n$-vertex graph $G$ along with a $(d,k)$-tree model of $G$, runs in time $n^{\Oh(dk)}$ and uses at most $\Oh(dk \log n)$ space, and solves the {\sc{Max Cut}} problem on $G$.
\end{restatable}
\fi

Upon closer inspection, the standard dynamic programming for {\sc{Max Cut}} on clique expressions solves a {\sc{Subset Sum}}-like problem whenever aggregating the dynamic programming tables of children to compute the table of their parent. We apply the approach of Kane~\cite{Kane10} that was used to solve {\sc{Unary Subset Sum}} in logarithmic space: we encode the aforementioned {\sc{Subset Sum}}-like problem as computing the product of polynomials, and use Chinese remaindering to compute this product in a space-efficient way.

Finally, we consider the {\sc{Dominating Set}} problem, for which we prove the following.

\iflong
\begin{restatable}{theorem}{spaceDS}
\label{thm:domset}
There is a randomized algorithm which takes as input an $n$-vertex graph $G$ along with a $(d,k)$-tree model of $G$, runs in time $2^{\Oh(dk)} \cdot \polyn$ and uses at most $\Oh(dk^2 \log n+n\log n)$ space, and reports the minimum size of a dominating set in $G$ that is correct with probability at least $1/2$.
\end{restatable}
\fi
\ifshort
\begin{restatable}{theorem}{spaceDS}
\label{thm:domset}
There is a randomized algorithm which takes as input an $n$-vertex graph $G$ along with a $(d,k)$-tree model of $G$, runs in time $2^{\Oh(dk)} \cdot \polyn$ and uses at most $\Oh(dk^2 \log n+n\log n)$ space, and reports the minimum size of a dominating set in $G$ that is correct with probability at least $1/2$.
\end{restatable}
\fi

Note that the algorithm of Theorem~\ref{thm:domset} is randomized and uses much more space than our previous algorithms: more than $n \log n$. The reason for this is that we use the inclusion-exclusion approach proposed very recently by Hegerfeld and Kratsch~\cite{HegerfeldK23}, which is able to count dominating sets only modulo $2$. Consequently, while the parity of the number of dominating sets of certain size can be computed in space $\Oh(dk^2\log n)$, to determine the existence of such dominating sets we use the Isolation Lemma and count the parity of the number of dominating sets of all possible weights. This introduces randomization and necessitates sampling---and storing---a weight function. At this point we do not know how to remove neither the randomization nor the super-linear space complexity in Theorem~\ref{thm:domset}; we believe this is an excellent open~problem.

Note that in all the algorithms presented above, the running times contain a factor $d$ in the exponent compared to the standard (exponential-space) dynamic programming on clique expressions. The following conditional lower bound shows that some additional dependency on the depth is indeed necessary; the relevant precise definitions are provided in Section~\ref{sec:lb}.

\begin{restatable}{theorem}{lcs-lb}
\label{thm:lcs-lb}
Suppose {\sc{Longest Common Subsequence}} cannot be solved in time $M^{f(r)}$ and space $f(r)\cdot M^{\Oh(1)}$ for any computable function $f$, even if the length $t$ of the sought subsequence is bounded by $\delta(N)$ for any unbounded computable function $\delta$; here $r$ is the number of strings on input, $N$ is the common length of each string, and $M$ is the total bitsize of the instance. Then for every unbounded computable function $\delta$, there is no algorithm that solves the {\sc{Independent Set}} problem in graphs supplied with $(d,k)$-tree-models satisfying $d\leq \delta(k)$ that would run in time $2^{\Oh(k)}\cdot n^{\Oh(1)}$ and simultaneously use $n^{\Oh(1)}$ space.
\end{restatable}

The possibility of achieving time- and space-efficient algorithms for  {\sc{Longest Common Subsequence}} was also the base of conjectures formulated by Pilipczuk and Wrochna~\cite{PilipczukW18} for their lower bounds against time- and space-efficient algorithms on graphs of bounded pathwidth. The supposition made in Theorem~\ref{thm:lcs-lb} is a refined version of those conjectures that takes also the length of the sought subsequence into account. The reduction underlying Theorem~\ref{thm:lcs-lb} is loosely inspired by the constructions of~\cite{PilipczukW18}, but requires new ideas due to the different setting of tree-models of low depth.

Finally, given that the above results point to a fundamental role of shrubdepth in terms of space complexity, it is natural to ask whether shrubdepth can also be used to obtain meaningful tractability results with respect to the ``usual'' notion of fixed-parameter tractability. We conclude our exposition by highlighting two examples of problems which are \NP-hard on graphs of bounded cliquewidth (and even of bounded pathwidth)~\cite{ChlebikovaC17,LiP22}, 
and yet which admit fixed-parameter algorithms when parameterized by the shrubdepth.

\begin{restatable}{theorem}{FPTproblems}
\label{thm:FPTproblems}
{\sc Metric Dimension} and {\sc Firefighter} can be solved in fixed-parameter time on graphs supplied with $(d,k)$-tree-models, where $d$ and $k$ are considered the parameters.
\end{restatable}

\section{Preliminaries}\label{sec:prelims}
For a positive integer $k$, we denote by $[k]=\{1,\ldots,k\}$ and $[k]_0=[k]\cup \{0\}$.
For a function $f \colon A \to B$ and elements $a, b$ (not necessarily from $A \cup B$), the function $f[a \mapsto b] \colon A \cup \{a\} \to B \cup \{b\}$ is given by $f[a \mapsto b](x) = f(x)$ for $x \neq a$ and $f[a \mapsto b](a) = b$.
We use standard graph terminology~\cite{Diestel12}. 
\ifshort
The full proofs of our results also require the use of algebraic tools---notably the cover product and the fast subset convolution machinery of Björklund et al.~\cite{BjorklundHKK07}. 
\fi

We use the same computational model as Pilipczuk and Wrochna~\cite{PilipczukW18}, namely the RAM model where each operation takes time polynomially proportional to the number of bits of the input, and the space is measured in terms of bits. We say that an algorithm $A$ runs in time $t(n)$ and space $s(n)$ if, for every input of size $n$, the number of operations of $A$ is bounded by $t(n)$ and the auxiliary used space of $A$ has size bounded by $s(n)$ bits. 

\subparagraph{Shrubdepth.}
We first introduce the decomposition notion for shrubdepth: {\em{tree-models}}.

    \begin{definition}
    For $d,k\in \bN$, a $(d,k)$-tree-model $(T, \mathcal{M}, \cR, \lab)$ of a graph $G$ is 
    a rooted tree~$T$ of depth $d$ together with a family of symmetric Boolean $k \times k$-matrices $\mathcal{M} = \{M_a\}_{a \in V(T)}$, a labeling function $\lab \colon V(G) \to [k]$, and a family of renaming functions $\cR = \{\rename_{ab}\}_{ab \in E(T)}$ with $\rename_{ab} \colon [k] \to [k]$ for all $ab \in E(T)$ such that:
    \begin{itemize}  
    	\item The leaves of $T$ are identified with vertices of $G$. For each node $a$ of $T$, we denote by $V_a \subseteq V(G)$ the leaves of $T$ that are
          descendants of $a$, and with $G_a = G[V_a]$ we denote the subgraph induced by these vertices. 
    	
    	\item With each node $a$ of $T$ we associate a labeling function $\lab_a : V_a \to [k]$ defined as follows. 
            If $a$ is a leaf, then $\lab_a(a) = \lab(a)$.
            If $a$ is a non-leaf node, then for every child $b$ of $a$ and every vertex $v \in V_b$, we have $\lab_a(v) = \rename_{ab}(\lab_b(v))$.

    	\item 
            For every pair of vertices $(u,v)$ of $G$, let $a$ denote their least common ancestor in $T$. 
            Then we have $uv \in E(G)$ if and only if $M_a[\lab_a(u),\lab_a(v)]=1$.
    \end{itemize}
    \end{definition}

        We introduce some notation. If $(T, \mathcal{M}, \cR, \lab)$ is a $(d,k)$-tree model of a graph $G$, then for every node $a$ of $T$ and every $i\in [k]$, let $V_a(i) = \lab_a^{-1}(i)$ be the set
    of vertices labeled $i$ at $a$.
    Given a subset $X$ of $V_a$ and $i\in [k]$, let $X_a(i) = X \cap V_a(i)$ be the vertices of $X$ labeled $i$ at $a$.

\iflong
    A $(d,k)$-tree-model can be understood as a term of depth $d$ that constructs a $k$-labelled graph from single-vertex graphs by means of the following operations: renaming of the labels, and joining several labelled graphs while introducing edges between vertices originating from different parts based on their labels. This makes tree-models much closer to the {\em{NLC-decompositions}} which underly the parameter {\em{NLC-width}} than to clique expressions. NLC-width is a graph parameter introduced by Wanke~\cite{Wanke94} that can be seen as an alternative, functionally equivalent variant of cliquewidth.
    \fi

    We say that a class $\mathcal{C}$ of graphs has {\em{shrubdepth}} $d$ if there exists $k\in \N$ such that every graph in $\mathcal{C}$ admits a $(d,k)$-tree-model. Thus, shrubdepth is a parameter of a graph class, rather than of a single graph; though there are functionally equivalent notions, such as {\em{SC-depth}}~\cite{GanianHNOM19} or {\em{rank-depth}}~\cite{DeVosKO20}, that are suited for the treatment of single graphs. \ifshort
We remark that in the original definition proposed by Ganian et al.~\cite{GanianHNOM19}, relabeling is not allowed; however, using either definition yields the same notion of shrubdepth. 
Moreover, throughout this work we abstract away from the computation of the tree-models themselves and assume that a $(d,k)$-tree-model of the considered graph is provided on input. 

We note that a fixed-parameter algorithm for computing tree-models has been proposed by Gajarsk\'y and Kreutzer~\cite{GajarskyK20} (in the sense of Ganian et al.~\cite{GanianHNOM19}). The approach of Gajarsk\'y and Kreutzer is essentially kernelization: they iteratively ``peel off'' isomorphic parts of the graph until the problem is reduced to a kernel of size bounded only in terms of $d$ and $k$. This kernel is then treated by any brute-force method. Consequently, a straightforward inspection of their algorithm~\cite{GajarskyK20} shows that it can be implemented with polynomial space; but not space of the form $(d+k)^{\Oh(1)}\cdot \log n$, due to the necessity of storing all the intermediate graphs in the kernelization process. We leave as an open question the computation of a $(d,k)$-tree model, for a given graph $G$, running in time $f(d,k)\cdot n^{\Oh(1)}$ and using space $(d+k)^{\Oh(1)}\cdot \log n$.

    \fi

\iflong
    We remark that in the original definition proposed by Ganian et al.~\cite{GanianHNOM19}, there is no renaming of the labels: for every vertex $u\in V(G)$, $\lambda_a(u)$ is always the same label $\lambda(u)$ for all relevant nodes $a$. This boils down to all the renaming functions $\rename_{ab}$ equal to the identify function on $[k]$. Clearly, a $(d,k)$-tree-model in the sense of Ganian et al. is also a $(d,k)$-tree-model in our sense, while a $(d,k)$-tree-model in our sense can be easily turned into a $(d,k^{d+1})$-model in the sense of Ganian et al. by setting $\lambda(u)$ to be the $d+1$ tuple of consisting of labels $\lambda_a(u)$, for $a$ ranging over the ancestors of $u$ in $T$. Thus, using either definition yields the same notion of shrubdepth for graph classes. We choose to use the definition with renaming, as it provides more flexibility in the construction of tree-models that can result in a smaller number of labels and, consequently, better running times. It is also closer to the original definitions of clique expressions or NLC-decompositions.

    Within this work we will always assume that a $(d,k)$-tree-model of the considered graph is provided on input. Thus, we abstract away the complexity of computing tree-models, but let us briefly discuss this problem. Gajarsk\'y and Kreutzer~\cite{GajarskyK20} gave an algorithm that given a graph $G$ and parameters $d$ and $k$, computes a $(d,k)$-tree-model of $G$ (in the sense of Ganian et al.~\cite{GanianHNOM19}), if there exists one, in time $f(d,k)\cdot n^{\Oh(1)}$ for a computable function $f$. The approach of Gajarsk\'y and Kreutzer is essentially kernelization: they iteratively ``peel off'' isomorphic parts of the graph until the problem is reduced to a kernel of size bounded only in terms of $d$ and $k$. This kernel is then treated by any brute-force method. Consequently, a straightforward inspection of the algorithm of~\cite{GajarskyK20} shows that it can be implemented so that it uses polynomial space; but not space of the form $(d+k)^{\Oh(1)}\cdot \log n$, due to the necessity of storing all the intermediate graphs in the kernelization process.
    \fi

    \iflong
    \subparagraph{Cover products and transforms.} We now recall the algebraic tools we are going to use.
    Let $U$ be a finite set and  $R$ be a ring. Let  $g_1, \dots, g_t \colon 2^U \to R$ be set functions, for some~integer~$t$. For every
    $i\in [t]$, 
    the \emph{zeta-transform} $\xi g_i \colon 2^U \to R$ of $g_i$ is defined by 
    \ifshort
$(\xi g_i)(Y) = \sum\limits_{X \subseteq Y} g_i(X)$,
    \fi
    \iflong
    \[
        (\xi g_i)(Y) = \sum\limits_{X \subseteq Y} g_i(X),
    \]
    \fi
    and similarly, the \emph{Möbius-transform} $\mu g_i \colon 2^U \to R$ of $g_i$ is given by
    \ifshort
$(\mu g_i)(Y) = \sum\limits_{X \subseteq Y} (-1)^{|Y \setminus X|} g_i(X)$.
    \fi
    \iflong
    \[
        (\mu g_i)(Y) = \sum\limits_{X \subseteq Y} (-1)^{|Y \setminus X|} g_i(X).
    \]
    \fi
    The \emph{cover product} $g_1 \ast_c g_2 \ast_c \dots \ast_c g_t \colon 2^U \to R$ of $g_1, \dots, g_t$ is defined by
    \[
        (g_1 \ast_c g_2 \ast_c \dots \ast_c g_t)(Y) = \sum\limits_{\substack{X_1, \dots, X_t \subseteq 2^{[k]} \colon \\ X_1 \cup \dots \cup X_t = Y}} g_1(X_1) \cdot g_2(X_2) \cdot \dots \cdot g_t(X_t).
    \]
    We emphasize that unlike another well-known concept of subset convolution, here the sets $X_1, \dots, X_t$ are not required to be pairwise disjoint.
    The following result of Björklund et al.~\cite{BjorklundHKK07} will be relevant for us:
    \begin{lemma}[\cite{BjorklundHKK07}] \label{lem:cover-product}
        Let $U$ be a finite set,  $R$ be a ring, and  $g_1, \dots, g_t \colon 2^U \to R$ be set functions for a positive integer $t$.
        Then for every $X \in 2^U$, it holds that
        \[
            (\xi(g_1 \ast_c g_2 \ast_c \dots \ast_c g_t)) (X) = (\xi g_1)(X) \cdot (\xi g_2)(X) \cdot \dots \cdot (\xi g_t)(X).
        \]
        Also for every $i \in [t]$, we have $\mu (\xi (g_i)) = g_i$.
    \end{lemma}
\fi

\section{Space-Efficient Algorithms on Tree-Models}
\ifshort
\noindent \textbf{Independent Set.} \quad 
\fi
\iflong
 \subsection{Independent Set}\label{subsec:independent-set}
\fi
    %\todo{V: Some notation still needs to be unified with Mamadou's changes. I'll do this on Thursday morning.}
    In this section, we provide a fixed-parameter algorithm %solving the \textsc{Maximum Independent Set} problem parameterized by shrub-depth running 
    computing the independent set polynomial of a graph in time $2^{\mathcal{O}(dk)} \cdot \polyn$ and using  $\poly(d, k) \log n$ space, when given a $(d,k)$-tree model. In particular, given a $(d, k)$-tree model $(T, \mathcal{M}, \cR, \lab)$ of an $n$-vertex graph~$G$, our algorithm will allow to compute the number of independent sets of size $p$ for each $p\in [n]$. 
    For simplicity of representation, we start by describing an algorithm that uses $\poly(d, k, n)$ space and then show how a result by Pilipczuk and Wrochna~\cite{PilipczukW18} can be applied to decrease the space complexity to $ \poly(d, k) \log n$.
    %We first describe the main building blocks and then combine %them to obtain the algorithm. 
    % Let $(T, \mathcal{M}, \cR, \lab)$ be a $(k, d)$-tree model of an $n$-vertex graph $G$. Our algorithm will allow to compute, for each $p\in [n]$, the number of independent sets of size $p$. 
    % To solve the \textsc{Independent Set} problem %, which is enough since 
    % it suffices then to take the maximum integer $1\leq p\leq n$ such that the number of independent sets of size $p$ is non-zero. 

    %If $T$ is not part of the input, we can apply the algorithm of Gajarský and Kreutzer that given $G$, $k$, and $d$ as input either outputs an $(k, d)$-model of $G$ or it correctly determines that no such model exists. 
    
    %Writing about labels of vertices of a graph $G_a$ for some %node $a$ of $T$, we refer to the labels in $\lab_a$ if not %explicitly stated otherwise.

    In order to simplify forthcoming definitions/statements, let $a$ be an internal node of $T$ with $b_1,\ldots, b_t$ as children. %We first note that, 
    For $S\subseteq [k]$, we denote by $q(a,S,p)$ the number of independent sets $I$ of size $p$ of $G_a$ such that $S=\{i\in [k] \mid I_a(i)\ne \emptyset\}$. %, then the number of independent sets of size $p$ of $G_a$ is precisely    
    Let us define the polynomial
    \ifshort
$\IS(a,S) = \sum_{p\in \NN} q(a,S,p) \cdot x^p$.
    \fi
    \iflong
    \begin{align*}
        \IS(a,S) & = \sum\limits_{p\in \NN} q(a,S,p) \cdot x^p.
    \end{align*}
    \fi
    %we denote the polynomial whose coefficients
    For the root $r$ of $T$, the number of independent sets of~$G$ of size $p$ is then given by 
    \iflong
    \[ \sum\limits_{S\subseteq [k]} q(r,S,p).\]
    \fi
    \ifshort
$\sum_{S\subseteq [k]} q(r,S,p)$
    \fi
    and the independent set polynomial of~$G$ is
    \iflong
    \[
        \sum\limits_{S \subseteq [k]} \IS(r, S).
    \]
    \fi
    \ifshort
$\sum_{S \subseteq [k]} \IS(r, S).$
    \fi
    %\todo{V: I swapped $\alpha$ and $\beta$ in the Mamadou's part in order to make it consistent with the computations below.}
    Therefore, the problem boils down to the computation of $\IS(r, S)$ and its coefficients $q(r,S,p)$. A usual way to obtain a polynomial or logarithmic space algorithm is a top-down traversal of a rooted tree-like representation of the input---in our case, this will be the tree model. In this top-down traversal, the computation of coefficients $q(a,S,p)$ of $\IS(a, S)$ %$ 
    makes some requests to the coefficients $q(b_i,S_i,p_i)$ of $\IS(b_i, S_i)$
    %$q(b_i,S_i,p_i)$,
    %\todo{V: formally speaking, this is not correct. To achieve logarithmic space, we are not querying the coefficients itself but only the evaluation of the polynomial at a certain prime. Or do we want to still keep this sentence for intuition?}
    for each $i\in [t]$, for some integer $p_i$, and some set $S_i$ of labels of $G_{b_i}$ so that $\sum_{i\in [t]} p_i = p$ and $\bigcup_{i\in [t]} \rename_{ab_i}(S_i) = S$.
    Since there are exponentially many (in $t$) possible partitions of $p$ into $t$ integers and $t$ can be $\Theta(n)$, we must avoid running over all such integer partitions, and this will be done by the fast computation of a certain subset cover. %grouping the computations with respect to the number of children the requests are made for. We then use inclusion-exclusion branching to speed-up the computations while retaining the space complexity. %by $\poly(k,d)\cdot \log(n)$. 
    
    We will later show that if some independent set of $G_a$ contains vertices of labels $i$ and~$j$ with $M_a[i, j] = 1$, then all these vertices come from the same child of $a$. In particular, the vertices of label $i$ (rsp. $j$) cannot come from multiple children of $a$. %We will formalize this later. 
    To implement this observation, after fixing a set $S$ of labels, for each label class in $S$ we ``guess'' (i.e., branch on) whether it will come from a single child of $a$ or from many. 
    Such a guess is denoted by $\alpha \colon S\to \{1_{=}, 2_{\geq}\}$.
    So, the assignment $\alpha$ will allow us to control the absence of edges in the sought-after independent set.
    %The idea is that we ``guess'' (i.e., branch on), for each label class intersecting the desired independent sets, whether it will come from a single child or from many.
    %this will allow us to forbid combinations of partial solutions that induce edges.
    %We formalize these ideas below.
    %For every label in $S$, the assignment $\alpha$ is intended to reflect whether the vertices of an independent set $I$ come from a single child of $a$ or from at least two of them.
    %Later, in \cref{obs:alpha-2-geq}, we will argue that for a label $i_1$, an independent set $I$ satisfying $\alpha(i_1) = 2_\geq$ might exist only in the very special case where $M_a[i_1, i_2] = 0$ for all $i_2 \in S$: otherwise, there would necessarily be an edge in $I$ created at $a$.
    %So the assignment $\alpha$ allows us to control the absence of edges in the sought-after independent set.
    For a fixed $\alpha$, naively branching over all possibilities of assigning the labels of $S$ to the children of $a$ with respect to $\alpha$ would take time exponential in $t$, which could be as large as $\Theta(n)$.
    %To avoid this, we will rely on the fast computation of a certain subset cover by using algebraic techniques.
    We will use inclusion-exclusion branching to speed-up the computations while retaining the space complexity.
    In some sense, we will first allow less restricted assignments of labels to the children of $a$, and then filter out the ones that result in non-independent sets using the construction of a certain auxiliary graph.
    The former will be implemented by using ``less restricted'' guesses $\beta \colon S \to \{1_=, 1_\geq\}$ where $1_\geq$ reflects that vertices of the corresponding label come from at least one child of $a$.
    Note that if the vertices of some label $i$ come from exactly one child of $a$, then such an independent set satisfies both $\beta(i) = 1_=$ and $\beta(i) = 1_\geq$.
    Although it might seem counterintuitive, this type of guesses will enable a fast computation of a certain subset cover.
    After that, we will be able to compute the number of independent sets satisfying guesses of type $\alpha \colon S \to \{1_=, 2_\geq\}$ by observing that independent sets where some label $i$ occurs in at least two children of $a$ can be obtained by counting those where label $i$ occurs in at least one child and subtracting those where this label occurs in exactly one child.    

    We now proceed to a formalization of the above.    
    Let $S\subseteq \lab_a(V_a)$ and $\alpha \colon S\to \{1_{=}, 2_{\geq}\}$ be fixed.
    Let $s_1,\ldots, s_{\abs{\alpha^{-1}(2_{\geq})}}$ be an arbitrary linear ordering of $\alpha^{-1}(2_{\geq})$. 
    %For each label $i$ in $\alpha^{-1}(2_{\geq})$, we group the independent sets counted in $q(a,S,p)$ into those that intersect the label class $V_a(i)$ in exactly one child and those in possibly many children. 
    To compute the number of independent sets that match our choice of $\alpha$, we proceed by iterating over $c\in \{0,\ldots, \abs{\alpha^{-1}(2_{\geq})}\}$, and we count independent sets where the labels in $\{s_1, \dots, s_c\}$ occur exactly once, and the number of such sets where the labels occur at least once. Later, we will obtain the desired number of independent sets via carefully subtracting these two values. In particular, let $\gamma \colon \{s_1,\ldots,s_c\}\to \{1_{=}, 1_{\geq}\}$, and we denote by $q(a,S,\alpha,c,\gamma,p)$ the number of independent sets $I$ of size $p$ of $G_a$ such that 
     \begin{itemize}
        \item for every label $i \notin S$, we have $I_a(i)=\emptyset$;
        \item for every label $i \in \{s_1, \dots, s_c\}$ with $\gamma(i) = 1_=$, %we have $I_a(i)\ne \emptyset$ and 
        there exists a unique child $b_j$ of $a$ such that $I_a(i)\cap V_{b_j} \neq \emptyset$;
        % \todo[inline]{V: I've rewritten these conditions to $I_a(i)\cap V_{b_j} \neq \emptyset$ instead of $I_a(i) \subseteq V_{b_j}$ since the latter one allows $I_a(i)$ to be empty. And excluding the empty set four times looked very redundant.}
        \item for every label $i \in \{s_1, \dots, s_c\}$ with $\gamma(i) = 1_\geq$, %we have $I_a(i)\ne \emptyset$ and 
        there exists at least one child $b_j$ of $a$ such that $I_a(i)\cap V_{b_j} \neq \emptyset$;
        \item for every label $i \in S \setminus \{s_1, \dots, s_c\}$ with $\alpha(i) = 1_=$, %we have $I_a(i)\ne \emptyset$ and 
        there exists a unique child $b_j$ of $a$ such that $I_a(i)\cap V_{b_j} \neq \emptyset$;
        \item and for every label $i \in S \setminus \{s_1, \dots, s_c\}$ with $\alpha(i) = 2_\geq$, there exist at least two children~$b_{j_1}$ and $b_{j_2}$ of $a$ such that $I_a(i)\cap V_{b_{j_1}}\ne \emptyset$ and $I_a(i)\cap V_{b_{j_2}}\ne \emptyset$.
    \end{itemize}
    \iflong
    Then for $c \in [\alpha^{-1}(2_{\geq})]_0$ we define the polynomial $T(a, S, \alpha,c,\gamma) \in \ZZ[x]$ as 
    \[
        T(a, S, \alpha,c,\gamma) = \sum\limits_{p \in \NN_0} q(a,S,\alpha,c,\gamma,p) x^p.
    \]\fi
    We now proceed with some observations that directly follow from the definitions. 

\ifshort
     \begin{observation} \label{obs:qasp} 
     We have $q(a,S,p)$ = $\sum_{\alpha\in \{1_=,2_\geq\}^S, \gamma\in \{1_=,1_\geq\}^\emptyset} q(a,S,\alpha,0,\gamma,p)$ for every $S\subseteq \lab_a(V_a)$ and integer $p$.
     %\begin{equation} \label{eq:incl-excl-base} 
     %\[ 
    %\end{equation}
    %\]
    %\begin{equation} \label{eq:incl-excl-base}
     %   \IS(a, S) = \sum\limits_{\substack{\alpha\in \{1_=,2_\geq\}^S \\ \gamma\in \{1_=,1_\geq\}^\emptyset}}T(a, S, \alpha, 0,\gamma)
    %\end{equation}
         Also, for every $\alpha\in \{1_=,2_\geq\}^S$, every $c \in \{0,\ldots, \abs{\alpha^{-1}(2_\geq)} - 1\}$ and every $\gamma \colon \{s_1, \dots, s_c\} \to \{1_=, 1_\geq\}$, 
         we have
         %the following holds.
         %If $\alpha(s_{c+1} = 2_\geq)$, then we have
    %\[
        $q(a,S,\alpha,c,\gamma,p)  = q(a,S,\alpha, c+1, \gamma[s_{c+1} \mapsto 1_\geq],p) - q(a,S,\alpha,c+1, \gamma[s_{c+1} \mapsto 1_=], p)$.
    %\]   
    %and hence
    %\begin{equation} \label{eq:incl-excl-step-1}
        %T(a, S, \alpha,c,\gamma) = T(a, S, \alpha,c+1,\gamma[s_{c+1} \mapsto 1_\geq]) - T(a, S, \alpha,c+1,\gamma[s_{c+1} \mapsto 1_=]).
    %\end{equation}
\end{observation}
\fi
\iflong
\begin{observation} \label{obs:qasp} 
    For every $S\subseteq \lab_a(V_a)$ and integer $p$, we have
    \[
        q(a,S,p) = \sum\limits_{\substack{\alpha\in \{1_=,2_\geq\}^S, \\ \gamma\in \{1_=,1_\geq\}^\emptyset}} q(a,S,\alpha,0,\gamma,p)
    \]
     %We have $q(a,S,p)$ = $\sum_{\alpha\in \{1_=,2_\geq\}^S, \gamma\in \{1_=,1_\geq\}^\emptyset} q(a,S,\alpha,0,\gamma,p)$ 
     %\begin{equation} \label{eq:incl-excl-base} 
     %\[ 
    %\end{equation}
    %\]
    and hence,
    \begin{equation} \label{eq:incl-excl-base}
        \IS(a, S) = \sum\limits_{\substack{\alpha\in \{1_=,2_\geq\}^S \\ \gamma\in \{1_=,1_\geq\}^\emptyset}}T(a, S, \alpha, 0,\gamma)
    \end{equation}
    Moreover, for every $\alpha\in \{1_=,2_\geq\}^S$, every $c \in \{0,\ldots, \abs{\alpha^{-1}(2_\geq)} - 1\}$ and every $\gamma \colon \{s_1, \dots, s_c\} \to \{1_=, 1_\geq\}$, 
    we have
    \[
        q(a,S,\alpha,c,\gamma,p)  = q(a,S,\alpha, c+1, \gamma[s_{c+1} \mapsto 1_\geq],p) - q(a,S,\alpha,c+1, \gamma[s_{c+1} \mapsto 1_=], p).
    \]   
    and hence
    \begin{equation} \label{eq:incl-excl-step-1}
        T(a, S, \alpha,c,\gamma) = T(a, S, \alpha,c+1,\gamma[s_{c+1} \mapsto 1_\geq]) - T(a, S, \alpha,c+1,\gamma[s_{c+1} \mapsto 1_=]).
    \end{equation}%\todo{V: added to preliminaries R: is $\gamma[s_{c+1} \mapsto 1_\geq]$ standard notation, or should we define/explain it somewhere? Meaning is clear to me but I don't recall seeing it. Maybe we can use $\gamma\cup \{...\}$ instead?}
\end{observation}
\fi

\iflong
It remains then to show how to compute, for every $\alpha\in \{1_=,2_\geq\}^S$ and every $\gamma \in \{1_=,1_\geq\}^{\alpha^{-1}(2_\geq)}$, the polynomial 
$T(a,S,\alpha,\abs{\alpha^{-1}(2_\geq)},\gamma)$.
\fi 
\ifshort
It remains then to show how to compute, for every $\alpha\in \{1_=,2_\geq\}^S$, every $\gamma \in \{1_=,1_\geq\}^{\alpha^{-1}(2_\geq)}$, and every integer $p$ the value $q(a,S,\alpha,\abs{\alpha^{-1}(2_\geq)},\gamma,p)$.
\fi
%$q(a,S,\alpha, \abs{\alpha^{-1}(1_\geq)}, \gamma,p)$. 
It is worth mentioning that if $\beta\colon S\to \{1_=,1_\geq\}$ is such that $\beta^{-1}(1_=)=\alpha^{-1}(1_{=})\cup \gamma^{-1}(1_=)$ and $\beta^{-1}(1_\geq)=\alpha^{-1}(2_{\geq})\setminus \gamma^{-1}(1_=)$, then $q(a,S,\alpha, \abs{\alpha^{-1}(1_\geq)}, \gamma,p)$ is exactly the number of
independent sets $I$ of size $p$ of $G_a$ satisfying the following:
\begin{enumerate}
\item For every $i \in [k] \setminus S$, we have $I_a(i) = \emptyset$.
\item For every $i \in \beta^{-1}(1_{=})$, %$I_a(i)\ne \emptyset$ and 
there exists a unique index $j \in [t]$ such that $I_a(i) \cap V_{b_j} \neq \emptyset$.
\item For every $i \in \beta^{-1}(1_\geq)$, %$I_a(i)\ne \emptyset$ and 
there exists a (not necessarily unique) index $j \in [t]$ such that $I_a(i)
  \cap V_{b_j} \neq \emptyset$. 
\end{enumerate}
We will therefore write $q(a,S,\beta,p)$ instead of $q(a,S,\alpha, \abs{\alpha^{-1}(1_\geq)}, \gamma,p)$ and we define the polynomial $\TIS(a, S, \beta) \in \ZZ[x]$ (where ``T'' stands for ``transformed'') as
%\begin{align*}
  $\TIS(a,S,\beta) = \sum_{p\in \NN} q(a,S,\beta, p)\cdot x^p. $
%\end{align*}
%The \cref{obs:qasp} implies that the polynomial $\IS(a,S)$ can be computed once we are given the polynomials $\TIS(x, S, \beta)$, i.e.,
Recall that because we are computing 
%$q(a,S,p)$ 
$\IS(a, S)$
and
%q(a,S,\alpha, p)$ 
$\TIS(a, S, \beta)$
in a top-down manner, some queries for %$q(b_i,S_i,p_i)$'s 
$\IS(b_i, S_i)$
will be made during the computation. %What we will prove is that the coefficients of the following two polynomials, the coefficients corresponding to the values $q(a,S,p)$ and $q(a,S,\beta, p)$, can be requested
%in time $f(k,d)\cdot n$ and $poly(k,d)\cdot \log n$:
Before continuing in the computation of $\TIS(a,S,\beta)$, let us first explain how to request the polynomials $\IS(b_j,S_j)$ from each child $b_j$ of $a$.  If $a$ is not the root, let $a^*$ be its parent in $T$, and we use $\PIS(a, S)$ 
(where ``P'' stands for ``parent'') to denote the polynomial
 \iflong
        %\[
        \begin{align*} \PIS(a, S) & = \sum\limits_{p \in \NN_0} q^{\rename}(a, S, p) \cdot x^p\\ 
    %\]
    \textrm{where}\\
          q^{\rename}(a, S, p) &= \sum\limits_{\substack{D \subseteq \lab_a(V_a)\colon \\ \rename_{a^*a}(D) = S}} q(a, D, p)
          %
          % \]
          \end{align*}
          \fi
\ifshort
$\PIS(a, S) = \sum_{p \in \NN_0} q^{\rename}(a, S, p) x^p$ where $q^{\rename}(a, S, p) = \sum_{D \subseteq \lab_a(V_a)\colon \rename_{a^*a}(D) = S} q(a, D, p)$
\fi          
    is the number of independent sets of $G_a$ of size $p$ that contain a vertex with label $i \in [k]$ (i.e., $I_{a*}(i) \neq \emptyset$) if and only if $i \in S$ holds, \textbf{where the labels are treated with respect to $\lab_{a^*}$}.
    Then it holds that 
    \begin{equation}\label{eq:recursion-pis}
        \PIS(a, S) = \sum\limits_{\substack{D \subseteq \lab_a(V_a)\colon \\ \rename_{a^*a}(D) = S}} \IS(a, D) \enspace .
    \end{equation}
    %and therefore, the value $q_{a, S, p}$ can be encoded with at most $n$ bits. 
    %So the coefficients of this polynomial require $\mathcal{O}(n^2)$ space.
   
%We have seen in Observation \ref{obs:qasp} that the coefficients of $\IS(a,S)$ can be computed efficiently once we can query the coefficients of the polynomials
%$\TIS$. 
As our next step, we make some observations that will not only allow to restrict the $\beta$'s we will need in computing the polynomial $\IS(a, S)$ from the polynomials $\TIS(a, S, \beta)$,
but will also motivate the forthcoming definitions. Recall that we have fixed $S\subseteq \lab_a(V_a)$ and $\beta\colon S\to \{1_=,1_\geq\}$, and in $\IS(a,S)$ and $\TIS(a,S,\alpha)$
we are only counting independent sets $I$ such that $I_a(i)\ne \emptyset$ if and only if $i\in S$.  

\begin{observation} 
\label{obs:alpha-2-geq} If there exist $i_1, i_2 \in S$ such that $M_a[i_1, i_2] = 1$, then for any independent set $I$ counted in
  $\IS(a,S)$, there exists a unique $j\in [t]$ such that $I_a(i_1)\cup I_a(i_2)\subseteq V_{b_j}$.
\end{observation}

\iflong
\begin{proof} Both $I_a(i_1)$ and $I_a(i_2)$ are non-empty. 
    So if there are at least two distinct $j_1$ and $j_2$ in $[t]$ such that $I_1 \coloneqq I_a(i_1)\cap V_{b_{j_1}}$ and $I_2 \coloneqq I_a(i_2)\cap V_{b_{j_2}}$ are non-empty, 
    then $M_a[i_1, i_2] = 1$ implies that there is a complete bipartite graph between $I_1$ and $I_2$.
    Hence the graph induced on $I$ would contain an edge, which is a contradiction. 
  %then the child $b_{j_3}$ such that $I_a(i_2)\cap V_{b_{j_3}}\ne \emptyset$ is at least distinct from one of
  %$b_{j_1}$ or $b_{j_2}$, implying that $I$ contains an edge, a contradiction. 
\end{proof}
\fi

Recall that for every label $i\in \alpha^{-1}(2_\geq)$, each independent set $I$ contributing to the value $q(a,S,\alpha,0,\gamma,p)$ has the property that there are distinct children $b_{j_1}$ and $b_{j_2}$ such that $I_a(i)\cap V_{b_{j_1}}$ and $I_a(i)\cap V_{b_{j_2}}$ are both non-empty. 
Then by \cref{obs:alpha-2-geq} for every $i_1\in S$ it holds that if $\alpha(i_1)=2_\geq$, then $M_a[i_1,i_2]=0$ for all $i_2\in S$. 
So if $\alpha$ does not satisfy this, the request $T(a,S,\alpha,0,\gamma)$ can be directly answered with $0$.
%Now we restrict ourselves to only conflict-free $\alpha$'s. 
Otherwise, since we use \cref{obs:qasp} %and \eqref{eq:incl-excl-step-2} 
for recursive requests, the requests $\TIS(a, S, \beta)$ made all have the property that for each $i_1\in S$ the following holds: if $\beta(i_1)=1_\geq$, then $M_a[i_1,i_2]=0$ for all $i_2\in S$.
We call such $\beta$'s \emph{conflict-free} and we restrict ourselves to only conflict-free $\beta$'s. 
In other words, we may assume that if $i_1, i_2 \in S$ and $M_a[i_1, i_2] = 1$, then we have $\beta(i_1) = \beta(i_2) = 1_=$.
Observation \ref{obs:alpha-2-geq} implies that for such $i_1$ and $i_2$, each independent set $I$ counted in $\TIS(a,S,\beta)$ is such that $I_a(i_1) \cup I_a(i_2) \subseteq V_{b_j}$ for some child $b_j$ of $a$. 
Now, to
capture this observation, we define an auxiliary graph $F^{a, \beta}$ as follows.  The vertex set of $F^{a, \beta}$ is $\beta^{-1}(1_=)$ and there is an edge
between vertices $i_1 \neq i_2$ if and only if $M_a[i_1, i_2] = 1$.  Thus, by the above observation, if we consider a connected component $C$ of $F^{a, \beta}$,
then in each independent set $I$ counted in $\TIS(a, S, \beta)$, all the vertices of $I$ with labels from $C$ come from a single child of $a$.
%\todo{R: Do I understand correctly that different independent sets can still have their C-labeled vertices come from different children of $a$? (It's just that each independent set has all of its C-labeled vertices come from a single child, right?) If yes, I'd suggest replacing ``same child'' with ``a single''.}
    
\begin{observation} \label{obs:connected-components} Let $C$ be a connected component of $F^{a, \beta}$. %and let $i_1, i_2$ be in $C$.  
For every independent set
  $I$ counted in $\TIS(a, S, \beta)$, there exists a unique $j\in [t]$ such that $\bigcup_{i \in C} I_a(i) \subseteq V_{b_{j}}$. %for all $i \in C$.
\end{observation}
    
    We proceed with some intuition on how we compute $\TIS(a, S, \beta)$ by requesting some $\PIS(b_j, S_j)$.
    %Before proceeding to the formal computation of $\TIS(a, S, \beta)$, let us provide an intuition for how this is carried out.
    Let $I$ be some independent set counted in $\TIS(a, S, \beta)$. 
    This set contains vertices with labels from the set $S$, and the assignment $\beta$ determines whether there is exactly one or at least one child from which the vertices of a certain label come from.
    Moreover, by \cref{obs:connected-components}, for two labels $i_1, i_2$ from the same connected component of $F^{a, \beta}$, the vertices with labels $i_1$ and $i_2$ in $I$ come from the same child of $a$. 
    Hence, to count such independent sets, we have to consider all ways to assign labels from $S$ to subsets of children of $a$ such that the above properties are satisfied---namely, each connected component of $F^{a, \beta}$ is assigned to exactly one child while every label from $\beta^{-1}(1_{\geq})$ is assigned to at least one child.
    Since the number of such assignments can be exponential in $n$, we employ the fast computation of a certain subset cover.
    
    We now formalize this step.
    Let $\cc(F^{a, \beta})$ we denote the set of connected components of~$F^{a, \beta}$.
    The universe $U^{a, \beta}$ (i.e., the set of objects we assign to the children of $a$) is defined as
    %\[
        $U^{a, \beta} = \beta^{-1}(1_{\geq}) \cup \cc(F^{a, \beta})$.
    %\]
    For every $j \in [t]$, we define a mapping $f_j^{a, \beta} \colon 2^{U^{a, \beta}} \to \ZZ[x, z]$ (i.e., to polynomials over $x$ and $z$) as follows:
    %\[
        $f_j^{a, \beta}(X) = \PIS(b_j, \flatten^{a, \beta}(X)) z^{|X \cap \cc(F^{a, \beta})|}$
    %\]
    where $\flatten^{a, \beta} \colon 2^{U^{a, \beta}} \to 2^S$ 
    intuitively performs a union over all the present labels---formally:
\iflong
    \[
        \flatten^{a, \beta}(W) = (W \cap \beta^{-1}(1_{\geq})) \cup \bigcup\limits_{w \in W \cap \cc(F^{a, \beta})} w.
    \]
    \fi
\ifshort
$\flatten^{a, \beta}(W) = (W \cap \beta^{-1}(1_{\geq})) \cup \bigcup_{w \in W \cap \cc(F^{a, \beta})} w.$
\fi    
    So if we fix the set $X$ of labels coming from the child $b_j$, then the (unique) coefficient in $f_j^{a, \beta}(X)$ reflects the number of independent sets of $G_{b_j}$ using exactly these labels (with respect to $\lab_a$).
    The exponent of the formal variable $z$ is intended to store the number of connected components of $F^{a, \beta}$ assigned to $b_j$. 
    This will later allow us to exclude from the computation those assignments of labels from $S$ to children of $a$ where the elements of some connected component of $F^{a, \beta}$ are assigned to multiple children of $a$.
    For every $j \in [t]$, we define a similar function $g^{a, \beta}_j \colon 2^{S} \to \ZZ[x, z]$ as follows:
    \[
        g^{a, \beta}_j(Y) = 
        \begin{cases}
            f^{a, \beta}_j(X) & \text{if }  \flatten^{a, \beta}(X) = Y \text{ for some } X \in 2^{U^{a, \beta}},  \\
            0 & \text{otherwise.}
        \end{cases}
    \]
    Observe that the function $\flatten^{a, \beta}$ is injective and hence $g^{a, \beta}_j$ is well-defined.
    The mapping~$g^{a, \beta}_j$ filters out those assignments where some connected component of $F^{a, \beta}$ is ``split''.
    \iflong
    For simplicity of notation, when $a$ and $\beta$ are clear from the context, we omit the superscript $a, \beta$.%\todo{R: do we ever do this in the short version?}
    
    Crucially for our algorithm, we claim that the following holds:
    \[
        \TIS(a, S, \beta) = \left(\sum\limits_{\substack{X_1, \dots, X_t \subseteq [k] \colon \\ X_1 \cup \dots \cup X_t = S}} g_1(X_1) g_2(X_2) \dots g_t(X_t) \right) \langle z^{|\cc(F)|} \rangle
    \]
    where for a polynomial $P = \sum_{u_1, u_2 \in \NN_0} q_{u_1, u_2} x^{u_1} z^{u_2} \in \ZZ[x, z]$ the polynomial $P \langle z^{|\cc(F)|} \rangle \in \ZZ[x]$ is defined as $P \langle z^{|\cc(F)|} \rangle = \sum_{u_1 \in \NN_0} q_{u_1, |\cc(F)|} x^{u_1}$.
    In simple words, the $\langle z^{|\cc(F)|} \rangle$ operator first removes all terms where the degree of $z$ is not equal to $|\cc(F)|$ and then ``forgets'' about $z$.
    Before we provide a formal proof, let us sketch the idea behind it. 
    On the left side of the equality, we have the polynomial keeping track of the independent sets of $G_a$ that ``respect'' $\beta$. 
    First, for every label $i \in S$, some vertex of this label must occur in at least one child of $a$: this is handled by considering all covers $X_1 \cup \dots \cup X_t = S$ where for every $j \in [t]$, the set $X_j$ represents the labels assigned to the child $b_j$. 
    Next, if some $X_j$ ``splits'' a connected component, i.e., takes only a proper non-empty subset of this component, then such an assignment would not yield an independent set by \cref{obs:connected-components} and the function $g_j$ ensures that the corresponding cover contributes zero to the result.
    Hence, for every cover $X_1 \cup \dots \cup X_t = S$ with a non-zero contribution to the sum, every connected component of $F$ is completely contained in at least one $X_j$.
    In particular, this implies that for every non-zero term on the right side, the degree of the formal variable $z$ in this term is at least $z^{|\cc(F)|}$.
    On the other hand, if some connected component of $F$ is contained in several sets $X_j$, then the degree of the corresponding monomial is strictly larger than the total number of connected components and such covers $X_1, \dots, X_t$ are excluded from the consideration by applying $\langle z^{|\cc(F)|} \rangle$.
    We formalize this intuition below:
    \fi
    
    \ifshort
    \todo{R: if space permits, can put 10-line intuition here from long version.}
%\textcolor{red}{The following lemma then forms the last key remaining piece of the proof:}
    \fi
    \iflong
    \begin{lemma}
    \label{lem:tis-correctness}
        Let $(T, \mathcal{M}, \cR, \lab)$ be a $(d,k)$-tree model of an $n$-vertex graph $G$. Let $a$ be a non-leaf node of $T$ and let $b_1, \dots, b_t$ be the children of $a$. For every $S \subseteq \lab_a(V_a)$, and every conflict-free $\beta \colon S \to \{1_=, 1_\geq\}$, it holds that
        \[
            \TIS(a, S, \beta) = \left(\sum\limits_{\substack{X_1, \dots, X_t \subseteq [k] \colon \\ X_1 \cup \dots \cup X_t = S}} \left (\prod\limits_{j = 1}^t g^{a, \beta}_j(X_j)\right)\right) \langle z^{|\cc(F^{a, \beta})|} \rangle.
        \]
    \end{lemma}
    \fi
    \ifshort
\begin{lemma}
    
    \label{lem:tis-correctness}
        Let $(T, \mathcal{M}, \cR, \lab)$ be a $(d,k)$-tree model of an $n$-vertex graph $G$. Let $a$ be a non-leaf node of $T$ and let $b_1, \dots, b_t$ be the children of $a$. For every $S \subseteq \lab_a(V_a)$, and every conflict-free $\beta \colon S \to \{1_=, 1_\geq\}$, it holds that
        \[
            \TIS(a, S, \beta) = \left(\sum\limits_{\substack{X_1, \dots, X_t \subseteq [k] \colon \\ X_1 \cup \dots \cup X_t = S}} \left (\prod\limits_{j = 1}^t g^{a, \beta}_j(X_j)\right)\right) \langle z^{|\cc(F^{a, \beta})|} \rangle,
        \]
        where for a polynomial $P = \sum_{u_1, u_2 \in \NN_0} q_{u_1, u_2} x^{u_1} z^{u_2} \in \ZZ[x, z]$ the polynomial $P \langle z^{|\cc(F^{a, \beta})|} \rangle \in \ZZ[x]$ is defined as $P \langle z^{|\cc(F^{a, \beta})|} \rangle = \sum_{u_1 \in \NN_0} q_{u_1, |\cc(F^{a, \beta})|} x^{u_1}$.
    \end{lemma}
    \fi
    %\todo[inline]{B: Should we remove the def of $\langle z^{|\cc(F)|} \rangle$ in the long version since it is defined above the lemma?}
    %\todo{V: In the short first, the definition of $\langle z^{|\cc(F)|} \rangle$ should appear since it occurs in the above lemma.}

\iflong
    \begin{proof}
          First, we bring the right-hand side of the equality into a more suitable form.
          \begin{align*}
          	& \sum\limits_{\substack{X_1, \dots, X_t \subseteq [k] \colon \\ X_1 \cup \dots \cup X_t = S}} \prod\limits_{j = 1}^t g_j(X_j) = \\
          	& \sum\limits_{\substack{X_1, \dots, X_t \subseteq [k] \colon \\ X_1 \cup \dots \cup X_t = S, \\ \forall j \in [t] \exists W_j \subseteq U \colon \flatten(W_j) = X_j}} \prod\limits_{j = 1}^t \PIS(b_j, X_j) z^{|W_j \cap \cc(F)|} = \\
    	    & \sum\limits_{\substack{X_1, \dots, X_t \subseteq [k] \colon \\ X_1 \cup \dots \cup X_t = S, \\ \forall j \in [t] \exists W_j \subseteq U \colon \flatten(W_j) = X_j}} \left(\prod\limits_{j = 1}^t \Big(\sum\limits_{p_j \in \NN_0} q^{\rename}(b_j, X_j, p_j) x^{p_j}\Big) z^{|W_j \cap \cc(F)|} \right) = \\
    	    & \sum\limits_{\substack{X_1, \dots, X_t \subseteq [k] \colon \\ X_1 \cup \dots \cup X_t = S, \\ \forall j \in [t] \exists W_j \subseteq U \colon \flatten(W_j) = X_j \\ p_1, \dots, p_t \in \NN_0}} \prod\limits_{j = 1}^t q^{\rename}(b_j, X_j, p_j) x^{p_j} z^{|W_j \cap \cc(F)|} = \\
    	    & \sum\limits_{\substack{X_1, \dots, X_t \subseteq [k] \colon \\ X_1 \cup \dots \cup X_t = S, \\ \forall j \in [t] \exists W_j \subseteq U \colon \flatten(W_j) = X_j \\ p_1, \dots, p_t \in \NN_0}} \left(\prod\limits_{j = 1}^t q^{\rename}(b_j, X_j, p_j) \right) x^{\sum\limits_{j = 1}^t p_j} z^{\sum\limits_{j = 1}^t |W_j \cap \cc(F)|} \\
          \end{align*}
          We recall that $\flatten$ is injective so the sum above is well-defined.
          So we have to prove that
          \begin{align*}
            & \TIS(a, S, \beta) = \\
            & \left(\sum\limits_{\substack{X_1, \dots, X_t \subseteq [k] \colon \\ X_1 \cup \dots \cup X_t = S, \\ \forall j \in [t] \exists W_j \subseteq U \colon \flatten(W_j) = X_j, \\ p_1, \dots, p_t \in \NN_0}} \left(\prod\limits_{j = 1}^t q^{\rename}(b_j, X_j, p_j) \right) x^{\sum\limits_{j = 1}^t p_j} z^{\sum\limits_{j = 1}^t |W_j \cap \cc(F)|}\right) \langle z^{|\cc(F)|} \rangle,
          \end{align*}
          i.e.,
          \begin{align}
            & \TIS(a, S, \beta) = \nonumber \\
            & \sum\limits_{\substack{X_1, \dots, X_t \subseteq 2^{[k]} \colon \\ X_1 \cup \dots \cup X_t = S, \\ \forall j \in [t] \exists W_j \colon \flatten(W_j) = X_j, \\ \sum\limits_{j \in [t]} |W_j \cap \cc(F)| = |\cc(F)|, \\ p_1, \dots, p_t \in \NN_0}} \left(\prod\limits_{j = 1}^t q^{\rename}(b_j, X_j, p_j) \right) x^{\sum\limits_{j = 1}^t p_j} \label{eq:desired-equality}
          \end{align}
         
         To prove that these two polynomials are equal, we show that for every power $p \in \NN_0$ of $x$, the coefficients at $x^p$ in both polynomials are equal.
         So let us fix an arbitrary integer $p$.
                  
         For one direction, let $I$ be an independent set counted in the coefficient $q(a, S, \beta, p)$ at the term $x^p$ on the left-hand side $\TIS(a, S, \beta)$; in particular, we then have $|I| = p$.
         For every $j \in [t]$, let 
         %\[
            $I^j = I \cap V_{b_j}$, 
            $p_j = |I_j|$,
         %and let
            and $X_j = \{ i \in [k] \mid I_a(i) \cap V_{b_j} \neq \emptyset\}$.
         %\]
         Clearly, we have $p_1 + \dots + p_t = p$ and $X_1 \cup \dots \cup X_t = S$. 
         Now consider some $j \in [t]$. 
         The set $I^j$ is an independent set of $G_{b_j}$ that contains vertices with labels from exactly $X_j$ (with respect to $\lab_a$). 
         So $I^j$ is counted in $q^{\rename}(b_j, X_j, p_j)$.
         Let 
         %\[
            $A_j = X_j \cap \beta^{-1}(1_=)$
         %\]
         and 
         %\[
            $B_j = X_j \cap \beta^{-1}(1_\geq)$.
         %\]
         Note that $A_j \cup B_j = X_j$, $A_1 \cup \dots \cup A_t = \beta^{-1}(1_=)$, and $B_1 \cup \dots \cup B_t = \beta^{-1}(1_\geq)$.
         Then by \cref{obs:connected-components}, for every connected component $C$ of $F$ and every $j \in [t]$, we either have $C \subseteq X_j$ or $C \cap X_j = \emptyset$.
         Therefore, for every $j \in [t]$, we have
         %\[
            $X_j = B_j \cup \bigcup_{{C \in \cc(F) \colon C \cap X_j \neq \emptyset}} C$.
         %\]
         and hence,
         %\[
            $X_j = \flatten(W_j)$ where $W_j = B_j \cup \{ C \in \cc(F) \mid C \cap X_j \neq \emptyset \}$.
         %\]
         Finally, by the definition of objects counted in $\TIS(x, S, \beta)$, 
         since the labels from $\beta^{-1}(1_=)$ occur in exactly one child of $a$,
         it holds that $A_{j_1} \cap A_{j_2} = \emptyset$ for any $j_1 \neq j_2 \in [t]$.
         Together with $A_1 \cup \dots \cup A_t = \beta^{-1}(1_=)$ this implies that for every connected component $C$ of $F$, there exists exactly one index $j_C \in [t]$ with $C \subseteq A_{j_C}$, i.e., $C \in W_{j_C}$.
         So we obtain
         %\[
            $\sum\limits_{j \in [t]} |W_i \cap \cc(F)| = |\cc(F)|.$
         %\]
         Altogether, the tuple $(I^1, \dots, I^t)$ is counted in the product $\prod_{j = 1}^t q^\rename(b_j, X_j, p_j)$ and the properties shown above imply that this product contributes to the coefficient at the monomial $x^p$.
         Also note that the mapping of $I$ to $(I^1, \dots, I^t)$ is injective so we indeed obtain that the coefficient at $x^p$ on the left-hand side of \eqref{eq:desired-equality} is at most as large as one the right-hand side.
         
         Now we show that the other inequality holds as well. 
         Let $X_1, \dots, X_t \subseteq [k]$, $I^1, \dots, I^t \subseteq V$, $W_1, \dots, W_t \subseteq U$, and $p_1, \dots, p_t \in \NN_0$ be such that the following properties hold:
         \begin{itemize}
             \item $p_1 + \dots + p_t = p$,
             \item $X_1 \cup \dots \cup X_t = S$,
             \item for every $j \in [t]$, it holds that $\flatten(W_j) = X_j$,
             \item $\sum_{j \in [t]} |W_j \cap \cc(F)| = |\cc(F)|$,
             \item and for every $j \in [t]$, the set $I^j$ is an independent set of $G_{b_j}$ of size $p_j$ such that for every $i \in [k]$, $I^j_a(i) \neq \emptyset$ holds iff $i \in X_j$, i.e., $I^j$ is counted in $q^\rename(b_j, X_j, p_j)$.
         \end{itemize}
         Let $I = I^1 \cup \dots \cup I^t$.
         Since for every $j \in [t]$, we have $I^j \subseteq V_{b_j}$, the sets $I^1, \dots, I^t$ are pairwise disjoint and we have $|I| = p$.
         We also have $I \subseteq V_a$ and for every $i \in [k]$, we have $I_a(i) \neq \emptyset$ iff $i \in S$, i.e., $I$ contains vertices with labels from exactly $S$ with respect to $\lab_a$.
         We claim that $I$ is an independent set of $G_a$.
         Since $I_1, \dots, I_t$ are independent sets of $G_{b_1}, \dots, G_{b_t}$, respectively, and $G_{b_1}, \dots, G_{b_t}$ are induced subgraphs of $G_a$, it suffices to show that there are no edges between $I_{j_1}$ and $I_{j_2}$ for any $j_1 \neq j_2 \in [t]$.
         For this, suppose there is an edge $v_1 v_2$ of $G_a$ with $v_1 \in I^{j_1}$ and $v_2 \in I^{j_2}$ for some $j_1 \neq j_2 \in [t]$. 
         Also let $i_1 = \lab_a(v_1)$ and $i_2 = \lab_a(v_2)$. 
         Since $a$ is the lowest common ancestor of $v_1$ and $v_2$, it holds that $M_a[i_1, i_2] = 1$.
         By the assumption of the lemma, the mapping $\beta$ is conflict-free so we have $\beta(i_1) = \beta(i_2) = 1_=$. 
         Then, the property $M_a[i_1, i_2] = 1$ implies that $i_1$ and $i_2$ belong to the same connected component, say $C$, of $F$.
         Recall that we have $i_1 \in X_{j_1}$, $i_2 \in X_{j_2}$, $\flatten(W_{j_1}) = X_{j_1}$, and $\flatten(W_{j_2}) = X_{j_2}$. 
         Hence, it holds that 
         %\begin{equation}
            $C \in W_{j_1}$ and $C \in W_{j_2}$.
            %\label{eq:repeated-component}
         %\end{equation} 
         On the other hand, let $C'$ be an arbitrary connected component of $F$ and let $i \in C'$ be some label.
         The property $X_1 \cup \dots \cup X_t = S$ implies that there exists an index $j_{C'}$ with $i \in X_{j_{C'}}$.
         Due to $\flatten(W_{j_{C'}}) = X_{j_{C'}}$ we then have $C' \in W_{j_{C'}}$.
         I.e., every connected component of $F$ is contained in at least one of the sets $W_1, \dots, W_t$ while $C$ is contained in at least two such sets so we get
         %This together with \eqref{eq:repeated-component} implies that
         %\[
            $\sum_{j \in [t]} |W_j \cap \cc(F)| > |\cc(F)|$ 
         %\]
         %and we arrive at 
         -- a contradiction.
         
         Hence, the set $I$ is indeed an independent set of $G_a$ of size $p$ such that it contains vertices with labels from exactly $S$ with respect to $\lab_a$. 
         So it is counted in the coefficient $q^{\rename}(a, S, \beta, p)$ of the term $x^p$ in $\TIS(a, S, \beta)$.
         Finally, first note that $(I^1, \dots, I^t)$ is uniquely mapped to the tuple $(X_1, \dots, X_t, p_1, \dots, p_t)$ so it is counted only once on the right-hand side.
         And second, the mapping of $(I^1, \dots, I^t)$ to $I$ is injective (since $V_{b_1}, \dots, V_{b_t}$ are pairwise disjoint).
         %Also note that first, such a tuple $(I^1, \dots, I^t)$ is uniquely mapped to the tuple $(X_1, \dots, X_t, p_1, \dots, p_t)$ and second, the mapping of such tuples $(I^1, \dots, I^t)$ to $I$ is injective.
         %So $(I^1, \dots, I^p)$ is counted only once on the right side.
         Therefore, the coefficient at the term $x^p$ on the right-hand side of \eqref{eq:desired-equality} is at most as large as on the left-hand side.
         Altogether, we conclude that the two polynomials in \eqref{eq:desired-equality} are equal, as desired.
    \end{proof}
    \fi
    %R: continue from here

    \iflong
    The above lemma implies that:
    \begin{align}
        &\TIS(a, S, \beta) = \nonumber \\
        & \sum\limits_{\substack{X_1, \dots, X_t \subseteq [k] \colon \\ X_1 \cup \dots \cup X_t = S}} \left (\prod\limits_{j = 1}^t g_j(X_j)\right) \langle z^{|\cc(F)|} \rangle = \nonumber \\
        & (g_1 \ast_c g_2 \ast_c \dots \ast_c g_t)(S) \langle z^{|\cc(F)|} \rangle \stackrel{\cref{lem:cover-product}}{=} \nonumber \\
        & 
        \bigl((\mu(\xi(g_1 \ast_c g_2 \ast_c \dotsc \ast_c g_t)))(S)\bigr) \langle z^{|\cc(F)|} \rangle = \nonumber \\
        & \left(\sum\limits_{Y \subseteq S} (-1)^{|S \setminus Y|} (\xi(g_1 \ast_c g_2 \ast_c \dots \ast_c g_t))(Y)\right) \langle z^{|\cc(F)|} \rangle \stackrel{\cref{lem:cover-product}}{=} \nonumber \\
        & \left(\sum\limits_{Y \subseteq S} (-1)^{|S \setminus Y|} (\xi g_1)(Y) (\xi g_2)(Y) \dots (\xi g_t)(Y) \right) \langle z^{|\cc(F)|} \rangle = \nonumber \\
        & \left(\sum\limits_{Y \subseteq S} (-1)^{|S \setminus Y|} \prod\limits_{j = 1}^t (\xi g_j)(Y) \right) \langle z^{|\cc(F)|} \rangle = \nonumber \\
        & \left(\sum\limits_{Y \subseteq S} (-1)^{|S \setminus Y|} \prod\limits_{j = 1}^t \sum\limits_{Z \subseteq Y} g_j(Z) \right) \langle z^{|\cc(F)|} \rangle \label{eq:recursion-tis}
    \end{align}
    \fi
    \ifshort
    Now we can apply a result by Bj{\"{o}}rklund et al.\ \cite{BjorklundHKK07} %\todo{EJ: undefined in short version} 
    to accelerate the computation of $\TIS(a, S, \beta)$:
    %The above lemma together with \cref{lem:cover-product} implies that 
    It holds $\TIS(a, S, \beta) = \left(\sum\limits_{Y \subseteq S} (-1)^{|S \setminus Y|} \prod\limits_{j = 1}^t \sum\limits_{Z \subseteq Y} g_j(Z) \right) \langle z^{|\cc(F)|} \rangle$.
    \fi
    We now have the equalities required for our algorithm to solve \textsc{Independent Set} parameterized by shrubdepth.
    By using these equalities directly, we would obtain an algorithm running in time $2^{\mathcal{O}(k d)}\cdot n^{\mathcal{O}(1)}$ and space $\mathcal{O}(d k^2 n^2)$.
    However, the latter can be substantially improved by using a result of Pilipczuk and Wrochna~\cite{PilipczukW18} based on the Chinese remainder theorem:
    \begin{theorem}[\cite{PilipczukW18}]\label{thm:chinese-remainder}
        Let $P(x) = \sum\limits_{i=0}^{n'} q_i x^i$ be a polynomial in one variable $x$ of degree at most $n'$ with integer coefficients satisfying $0 \leq q_i \leq 2^{n'}$ for $i = 0, \dots, n'$. 
        Suppose that given a prime number $p \leq 2n' + 2$ and $s \in \FF_p$, the value $P(s) \pmod p$ can be computed in time~$T$ and space $S$. 
        Then given $k \in \{0, ... , n'\}$, the value $q_k$ can be computed in time  $\mathcal{O}(T \cdot \operatorname{poly}(n'))$ and space $\mathcal{O}(S + \log n')$.
    \end{theorem}

    With this, we can finally prove\iflong:\fi \ifshort  ~\Cref{thm:IS}.\fi

    \iflong
    \spaceIS*
    
        \begin{proof}
        %It suffices to prove the first part of the theorem since by the result of Gajarský and Kreutzer \cite{GajarskyK20}, a $(d, k)$-model of $G$ can be computed in FPT time and polynomial space. 
        % TODO: elaborate a bit more on the poly space part
        %For readability, we first provide an algorithm solving the problem in time $\mathcal{O}^*(2^{\mathcal{O}(k d)})$ and space $\mathcal{O}(k^2 d n^2)$ and in the end, apply the above theorem to reduce the space complexity.
        %As already mentioned, the algorithm is recursive. 
        The independent set polynomial of the graph $G = G_r$ is exactly 
        %\[
            $\sum_{S \subseteq [k]} \IS(r, S)$
        %\]
        where $r$ is the root of $T$.
        Let us denote this polynomial $P$.
        %The size of the maximum independent set of $G$ is given by
        %\[
        %    \max \bigg\{i \in [n]_0 \bigg \vert \sum\limits_{S \subseteq [k]} \IS(r, S) \langle k^i \rangle \neq 0\bigg\}.
        %\]
        %So the goal is to compute this sum.
        To apply \cref{thm:chinese-remainder}, we use $n' := n$.
        Let $p \leq 2n + 2$ be a prime number.
        The bound on $p$ implies that any number from $\FF_p$ can be encoded using $\mathcal{O}(\log n)$ bits, this will bound the space complexity.
        There are at most $2^n$ independent sets of $G$ so every coefficient of $P$ lies between $0$ and $2^n$, and therefore the prerequisites stated in the first sentence of \cref{thm:chinese-remainder} are satisfied.
        Let $s \in \FF_p$.
        We will now show that the value $P(s) \mod p$ can be evaluated in time $2^{\mathcal{O}(k d)} \cdot \polyn$ and space $\mathcal{O}(d k^2 \log n)$.
        At that point, the result will follow by \cref{thm:chinese-remainder}.

        Since we are interested in the evaluation of $P$ at $s$ modulo $p$, instead of querying and storing all coefficients, as a result of the recursion, we return the evaluation of a certain polynomial (e.g., $\TIS(a, S, \beta)$) at $s$ modulo $p$.
        For this, the formal variable $x$ is always substituted by $s$ and then arithmetic operations in $\FF_p$ are carried out.
        In the following, when computing a sum (resp.\ product) of certain values, these values are computed recursively one after another and we store the counter (\eg, current subset $S \subseteq [k]$) as well as the current value of the sum (resp.\ product). 
        Our algorithm relies on the equalities provided above and we now provide more details to achieve the desired time and space complexity.
        Let us denote $\AIS(a,S,\alpha) \coloneqq T_0(a, S, \alpha, 0, \emptyset)$ for simplicity.
                
        First, if $a$ is a leaf of $T$, then $\IS(a, S)$ can be computed directly via
    \begin{equation} \label{eq:leaf}
        \IS(a, S) =
        \begin{cases}
            1 & \text{if } S = \emptyset \\
            x & \text{if } S = \{\lab_a(a)\} \\
            0 & \text{otherwise} 
        \end{cases}
      \end{equation}
      and this is our base case.
    Otherwise, the queries are answered recursively and five types of queries occur, namely $\IS(a, S)$, $\AIS(a,S,\beta)$, %:=\sum_{p\in \NN} q(a,S,\beta,0,\emptyset,p)x^p$, 
    $T(a,S,\alpha,c,\gamma)%:=\sum_{p\in \NN} q(a,S,\beta, c, \gamma,p)x^p
    $, $\TIS(a, S, \beta)$, and $\PIS(a, S)$.
        %If $a$ is a leaf, then the query $\IS(a, S)$ can be answered directly using \cref{eq:leaf} and this is our base case.
        Let $a$ be an inner node with children $b_1, \dots, b_t$.
        To answer a query $T(a,S,\alpha,c,\gamma)$ for $c < |\alpha^{-1}(2_\geq)|$, we recurse via \eqref{eq:incl-excl-base}.
        If $c = |\alpha^{-1}(2_\geq)|$, then we first construct $\beta \colon S \to \{1_=, 1_\geq\}$ given by $\beta^{-1}(1_=)=\alpha^{-1}(1_{=})\cup \gamma^{-1}(1_=)$ and $\beta^{-1}(1_\geq)=\alpha^{-1}(2_{\geq})\setminus \gamma^{-1}(1_=)$ and then query $\TIS(a, S, \beta)$.
        %To answer a query $\AIS(a,S,\beta)$, we recurse via the inclusion-exclusion given by \eqref{eq:incl-excl-termination} starting from the base
        %case $\TIS(x,S,\alpha)$, for some $\alpha$ defined from $\beta$. 
        Then, to answer a query $\TIS(a, S, \beta)$, we recurse via \eqref{eq:recursion-tis}.
        And finally, to answer a query $\PIS(a, S)$, we recurse using~\eqref{eq:recursion-pis}.
        
        Each of the above recurrences is given by a combination of sums and products of the results of recursive calls and these values are from $\FF_p$.
        To keep the space complexity of the algorithm bounded, for such recursion, the result is computed ``from inside to outside'' by keeping track of the current sums (resp.\ products) as well as the next value to be queried.
        For example, for \eqref{eq:recursion-tis}, we iterate through all $Y \subseteq S$, store the current value of the outer sum (modulo $p$), then for fixed $Y$, we iterate over $j \in [t]$ and store $j$ and the current value of the product (modulo $p$), and then for fixed $Y$ and $j$, iterate through $Z \subseteq Y$ and store the current value of $Z$ and the current inner sum. 
        After the complete iteration over $Z$ (resp.\ $j$) we update the current value of the product (resp.\ outer sum) and move on to the next $j$ (resp.\ $Y$).
        
        Now we analyze the time and space complexity of the algorithm.
        We start with the running time.
        For this, we analyze how often every query is answered.
        Namely, for all relevant values of $a$, $S$, $\alpha$, $\beta$, $c$, and $\gamma$, for each query $\IS(a, S)$, $\AIS(a, S, \alpha)$, 
        $T(a, S,\alpha, c,\gamma)$, $\TIS(a, S, \beta)$, resp.\ $\PIS(a, S)$,
        we use $Q(\IS(a, S))$, $Q(\AIS(a, S, \alpha))$, 
        $Q(T(a, S,\alpha, c,\gamma))$, $Q(\TIS(a, S, \beta))$, $Q(\PIS(a, S))$, respectively, to denote the number of times the query is answered by the algorithm and we call this value the \emph{multiplicity} of the query.
        Then, for $h \in [d]_0$, we define the value $q_{\IS}(h)$ to be the maximum multiplicity of a query $\IS(a, S)$ over all nodes $a$ at height $h$ in $T$ and all reasonable $S$.
        Similarly, we define the values $Q_{\AIS}(h)$, 
        $Q_{T}(h)$, $Q_{\TIS}(h)$, and $Q_{\PIS}(h)$ where we maximize over all nodes $a$ at height $h$ and all reasonable values of $S$, $\alpha$, $\beta$, $\gamma$ and $c$.
        We now upper bound these values. 
        
        Let $b$ be a node at height $h$ for some $h \in [d]_0$. 
        If $b = r$, then a query $\PIS(b, S)$ is not asked at all.
        Otherwise, let $a$ be the parent of $b$, and let $j$ be such that $b = b_j$ is the $j$-th child of $a$.
        Then $\PIS(b, S)$ can be asked when answering some query of the form $\TIS(a, D, \beta)$ to compute some value $\xi g_j^{a, \beta}(Y)$ such that $S \subseteq Y \subseteq D$.
        Therefore, for fixed $D$ and $\beta$, the value $\IS(b, S)$ is queried at most $2^k$ times, so we obtain 
        %\[
            $Q(\PIS(b, S)) \leq \sum_{{S \subseteq D \subseteq [k], \beta \colon D \to \{1_=, 1_\geq\}}} Q({\TIS}(a, D, \beta))$
        %\]
        and hence, 
        \[
            Q_{\PIS}(h)
            \begin{cases}
                = 0 & \text{if } h = 0 \\
                \leq 2^{3 k} Q_{\TIS}(h - 1) & \text{otherwise} 
            \end{cases}
            .
        \]
        
        %------------%
        Next, we consider a query of form $\TIS(b, S, \beta)$.
        Observe that for every $\alpha \colon S \to \{1_=, 2_\geq\}$ when recursing via \eqref{eq:incl-excl-base} to answer $T(a, S, \alpha, 0, \emptyset)$, we \emph{branch} on the values $1_=$ and $1_\geq$ for $s_1, \dots, s_{|\alpha^{-1}(2_\geq)|}$ one after another.
        Thus, after $|\alpha^{-1}(2_\geq)|$ steps every branch results in its own $\gamma \colon s_1, \dots, s_{|\alpha^{-1}(2_\geq)|} \to \{1_=, 1_\geq\}$, and hence, in its own $\beta \colon S \to \{1_=, 1_\geq\}$.
        Therefore, if we fix $\alpha$, then every $\TIS(a, S, \beta)$ is queried at most once when answering $T(a, S, \alpha, 0, \emptyset)$.
        %\sum_{p\in \NN} q(a,S,\beta,0,\emptyset,p)x^p$
        %Earlier, we observed that it is asked at most once for every query of form $\AIS(b, S, \alpha)$. 
        Hence, we have 
        \[
            Q(\TIS(b, S, \beta)) \leq \sum_{\alpha \colon S \to \{1_=, 1_\geq\}} Q(\AIS(b, S, \alpha))
        \]
        and therefore,
        %\[
            $Q_{\TIS}(h) \leq 2^k Q_{\AIS}(h)$.
        %\]
        %The above argument also implies that e
        Every query $T(b, S, \alpha,c,\gamma)$ is also asked at most once while answering a query of $T(b, S, \alpha,0,\emptyset)$, i.e.,
        %\[
            $Q(T(b,S,\alpha,c,\gamma)) \leq Q(\AIS(b, S, \alpha))$
        %\]
        and
        %\[
            $Q_T(h) \leq Q_{\AIS}(h)$.
        %\]
        
        Further, for each fixed $\alpha$, a query $\AIS(b, S, \alpha)$ is asked exactly once for every query of $\IS(b, S)$, i.e.,
        %\[
            $Q(\AIS(b, S, \alpha)) \leq Q(\IS(b, S))$
        %\]
        and 
        %\[
            $Q_{\AIS}(h) \leq Q_{\IS}(h)$.  
        %\]
        Finally, a query of the  form $\IS(b, S)$ is queried at most once for every query of the form $\PIS(b, D)$, so we have
        %\[
            $Q(\IS(b, S)) \leq \sum_{D \subseteq [k]} Q(\PIS(b, D))$
        %\]
        and
        %\[
            $Q_{\IS}(h) \leq 2^k Q_{\PIS}(p)$.
        %\]
        
        By induction over $h$, we obtain that
        \[
            Q_{\AIS}(h), Q_T(h), Q_{\TIS}(h), Q_{\IS}(h), Q_{\PIS}(h) \leq 2^{5 h k}
        \]
        and
        \[
            Q_{\AIS}(h), Q_T(h), Q_{\TIS}(h), Q_{\IS}(h), Q_{\PIS}(h) \leq 2^{5 d k} \in 2^{\mathcal{O}(k d)}
        \]
        for every $h \in [d]_0$, i.e., any fixed query is asked $2^{\mathcal{O}(k d)}$ times.
        
        There are $\mathcal{O}(nd)$ nodes in $T$ and there are at most $2^k$ reasonable values of $S$; for any $S$, there are at most $2^k$ choices for $\alpha$, $\beta$, and $\gamma$; and there are at most $k$ reasonable values of $b$.
        Hence, there are at most 
        \[
            \mathcal{O}(nd) (2^{2 k} + 2^{3 k} k + 2^{2 k} + 2^k + 2^k) \in 2^{\mathcal{O}(k)} \cdot \polyn 
        \]
        different forms of queries and so there are at most
        %\[
            $2^{\mathcal{O}(k d)} 2^{\mathcal{O}(k)} \cdot \polyn = 2^{\mathcal{O}(k d)} \cdot \polyn$
        %\]
        recursive calls.
        
        Next, we bound the time spent on each query additionally to the recursive calls. %as well as the space required for each recursive call. 
        For each query, this additional time is mostly determined by $\mathcal{O}(2^{2 k} n)$ arithmetic operations. 
        For a query of the form $\TIS(\cdot)$, arithmetic operations are carried out over polynomials in a formal variable $z$ %whose degree is always bounded by $k$ and 
        where the coefficients are from $\FF_p$.
        It is crucial to observe that since in the end of the computation we apply the $\langle z^{|\cc(F)|} \rangle$ operation and the auxiliary graph $F$ has at most $k$ connected components, we can safely discard coefficients at terms $z^r$ for any $r > k$.
        Therefore, it suffices to keep track of at most $k$ coefficients from $\FF_p$.
        For the remaining queries, the arithmetic operations are carried out over $\FF_p$. 
        %polynomials in at most two formal variables $k$ and $z$ such that the degree of this variable in any non-zero monomial is bounded by $n$ resp.\ $k$.
        So in any case, there are at most $k$ relevant values from $\FF_p$ to store as a partial sum resp.\ product and a single arithmetic operation can be therefore carried out in $\polyn$ time.
        %So there are at most $n k$ relevant coefficients.
        %Observe that the values of these coefficients are bounded by $2^n$ since in each case, they reflect the number of independent sets of $G$ with certain properties. 
        %Hence, each coefficient can be encoded with $\mathcal{O}(n)$ bits.
        Further, when answering a query of the form $\TIS(a, S, \beta)$ and computing a value of the form $g_j^{a, \beta}(X_j)$ for this, we can check whether for $X_j$ there is $W_j$ with $\flatten^{a, \beta}(W_j) = X_j$ as follows. First, we compute the connected components of $F^{a, \beta}$: we start with a partition of $\beta^{-1}(1_=)$ into singletons and then iterate over all pairs of vertices $i_1$ and $i_2$ and if $M_a[i_1, i_2] = 1$, then we merge the sets containing $i_1$ and $i_2$.
        As a result of this process, we obtain the set of connected components of $F^{a, \beta}$.
        Then for each connected component $C$, we check if $C \cap X_j \in \{\emptyset, C\}$ holds.
        If this does not hold for at least one connected component, then we conclude that $g_j^{a, \beta}(X_j) = 0$.
        Otherwise, the desired set $W_j$ exists and we have $|W_j \cap \cc(F^{a, \beta})| = r$ where $r$ is the number of connected components $C$ with $C \cap X_j = C$.
        This process then runs in time $\mathcal{O}(k^3)$ and space $\mathcal{O}(k)$.
        Although this can be accelerated, this step is not a bottleneck so this time and space complexity suffices for our purposes.
        Also, when answering a query of the form $\AIS(a, S, \alpha)$, we need to check whether there exist labels $i_1, i_2 \in S$ with $\alpha(i_1) = 1_\geq$ and $M_a[i_1, i_2] = 1$: this can be done in time $\mathcal{O}(k^2)$ and space $\mathcal{O}(\log k)$ by considering all pairs $i_1, i_2 \in S$ and looking up these properties.
        So for any query, the time spent on this query apart from the recursive calls is bounded by
        %\[
            $\mathcal{O}(2^{2 k} n \cdot k^3 \cdot \log n) = 2^{\mathcal{O}(k)} \cdot \polyn$.
        %\]
        And the total running time of the algorithm is bounded by 
        %\[
            $2^{\mathcal{O}(k d)} \cdot \poly(n) \cdot 2^{\mathcal{O}(k)} \cdot \poly(n) = 2^{\mathcal{O}(k d)} \cdot \poly(n)$, 
        %\]
        i.e., the number of queries times the complexity of a single query.

        Finally, we bound the space complexity.
        The space used by a single query is to store the partial sums and/or products modulo $p$ as well as the counters that store the information about the next recursive call (e.g., current $S$).
        For any query other than $\TIS(\cdot)$, the partial result is in $\FF_p$. 
        For a query of the form $\TIS(\cdot)$, we are working with a polynomial in the formal variable $z$.
        Above we have argued why the coefficients at $z^p$ for $p > r$ can be discarded.
        Therefore, it suffices to keep track of at most $k$ coefficients from $\FF_p$.
        Recall that $p \leq 2n+2$ so any value from $\FF_p$ can be encoded with $\log n$ bits.
        When answering a query of the form $\TIS(a, S, \beta)$, we also need to consider the connected components of $F^{a, \beta}$: as argued above, this can be accomplished in $\mathcal{O}(k)$ space.
        So the space complexity of a single query can be bounded by 
        %\[
            $\mathcal{O}(k \log n) + \mathcal{O}(k) + \log n = \mathcal{O}(k \log n)$.
        %\]
        The depth of the recursion is bounded by $\mathcal{O}(k d)$: the depth of $T$ is $d$ and for each node, there are at most $k + 4$ recursive calls queried at this node (namely, $\PIS(\cdot)$, $\IS(\cdot)$, $\AIS = T(\cdot, c = 0, \cdot)$, $\dots$, $T(\cdot, c = |\alpha^{-1}(2_\geq)| \leq k, \cdot)$, $\TIS(\cdot)$). %and $k$ functions $T(\cdot,\cdot,\cdot,i,\cdot)$, for $i\in [k]$).
        Finally, during the algorithm we need to keep track of the node we are currently at.
        Therefore, the space complexity of the algorithm is 
        %\[
             $\mathcal{O}(k d) \mathcal{O}(k \log n) + \mathcal{O}(\log n) = \mathcal{O}(k^2 d \log n)$.
        %\]
    \end{proof}
    \fi

\ifshort
\medskip
\noindent \textbf{Counting List-Homomorphisms.} \quad
\fi
\iflong
\subsection{Counting List-Homomorphisms}\label{subsec:list-homomorphisms}
\fi
    We now explain how to apply the techniques from \iflong\cref{subsec:independent-set} \fi\ifshort the above \fi to a broader class of problems, namely all problems expressible as instantiations of the \textsc{\#-List-$H$-Homomorphism} problem for a fixed pattern graph $H$ (which we will introduce in a moment). In this way, we cover problems such as \textsc{Odd Cycle Transversal} and \textsc{$q$-Coloring}, for a fixed $q$. Furthermore, the techniques will be useful for solving {\sc{Dominating Set}} later. 
    
    Let $H$ be a fixed undirected graph (possibly with loops) and let $R \subseteq V(H)$ be a designated set of vertices.
    An instance of the \textsc{$R$-Weighted \#-List-$H$-Homomorphism} problem consists of
        a graph $G$,
        a weight function $\omega \colon V(G) \to \NN$,
        a list function $L \colon V(G) \to 2^{V(H)}$,
        a \emph{cardinality} $C \in \NN$ and a \emph{total weight} $W \in \NN$.
    The goal is to count the number of list $H$-homomorphisms of $G$ such that exactly $\cardinality$ vertices of $G$ are mapped to $R$ and their total weight in $\omega$ is $\totalweight$. 
    More formally, we seek the value
    \begin{align*}
        \bigl|\bigl\{
            \phi \colon V(G) \to V(H) \ \bigm\vert\  & \forall v \in V(G) \colon \phi(v) \in L(v),
            \forall uv \in E(G) \colon \phi(u) \phi(v) \in E(H), \\
            &|\phi^{-1}(R)| = \cardinality,\textrm{ and }
            \omega(\phi^{-1}(R)) = \totalweight
        \bigr\}\bigr| \enspace .
    \end{align*}
    We say that such $\phi$ has cardinality $C$ and weight $W$.
    For the ``standard'' \textsc{\#-List $H$-Homomorphism} problem we would use $R = V(H)$, $C = W = |V(G)|$, and unit weights. We also have the following special cases of the \textsc{$R$-Weighted \#-List-$H$-Homomorphism} problem. In all cases, we consider unit weights.
    \begin{itemize}
        \item To model \textsc{Independent Set}, the pattern graph $H$ consists of two vertices $\mathbf u$ and $\mathbf v$ and the edge set contains a loop at $\mathbf v$ and the edge $\mathbf{uv}$. 
    The set $R$ consists of $\mathbf u$ only. 
    \iflong
    Then \textsc{Independent Set} is equivalent to finding the largest $\cardinality$ for which we have a positive number of solutions in the constructed instance of \textsc{$R$-Weighted \#-List-$H$-Homomorphism}. \fi
        \item Similarly, to model \textsc{Odd Cycle Transversal}, the pattern graph $H$ is a triangle on vertex set $\{\mathbf u,\mathbf v,\mathbf w\}$ with a loop added on $\mathbf u$.
        Again, we take $R=\{\mathbf u\}$.
        \item To model \textsc{$q$-Coloring}, we take $H$ to be the loopless clique on $q$ vertices, and $R=V(H)$.
    \end{itemize}
 While in all the cases described above we only use unit weights, we need to work with any weight function in our application to {\sc{Dominating Set}}.
 \ifshort
\iflong In this section, we build on the techniques of Subsection~\ref{subsec:independent-set} to establish the following result:\fi
 \fi
    \ifshort
    We prove the following result.
    \fi

   \begin{theorem}
   \label{thm:counting-weighted-list-homomorphisms}
        Fix a graph $H$ (possibly with loops) and $R \subseteq V(H)$. There is an algorithm which takes as input an $n$-vertex graph $G$ together with a weight function $\omega$ and a $(d, k)$-tree-model, runs in time $2^{\mathcal{O}(d k)} \cdot \polyn \cdot (W^*)^{\mathcal{O}(1)}$ and uses space $\mathcal{O}(k^2 d(\log n + \log W^*))$, and solves the \textsc{$R$-Weighted \#-List-$H$-Homomorphism} in $G$, where $W^*$ denotes the maximum weight in $\omega$.
    \end{theorem}
    
\iflong
    Using the argumentation above, from \cref{thm:counting-weighted-list-homomorphisms} we can derive the following corollaries.
        
        \begin{longcorollary}\label{thm:counting-list-homomorphisms}
        Fix a graph $H$ (possibly with loops).
        Then given an $n$-vertex graph $G$ together with a $(d, k)$-tree-model, \textsc{\#-List-$H$-Homomorphism} in $G$ can be solved in time $2^{\mathcal{O}(dk)} \cdot \polyn$ and space $\mathcal{O}(dk^2 \log n)$.
    \end{longcorollary}
    
    \begin{longcorollary}\label{thm:odd-cycle-transversal}
        Fix $q \in \NN$. Then given an $n$-vertex graph $G$ together with a $(d, k)$-tree-model, \textsc{$q$-Coloring} and \textsc{Odd Cycle Transversal} in $G$ can be solved in time $2^{\mathcal{O}(dk)} \cdot \polyn$ and space $\mathcal{O}(dk^2 \log n)$.
    \end{longcorollary}
    The remainder of this section is devoted to the proof of \cref{thm:counting-weighted-list-homomorphisms}. We assume that the reader is familiar with the approach presented in \cref{subsec:independent-set}, as we will build upon it.
    
    Let now $H$ and $R$ be fixed and let $W^* = \max_{v \in V} \omega(v)$ be the maximum weight in $\omega$.
    Now we show how to adapt our techniques from \cref{subsec:independent-set} to the \textsc{$R$-Weighted \#-List-$H$-Homomorphism} problem.
    We assume that the graph $G$ is provided with a $(d, k)$-model $(T, \mathcal{M}, \cR, \lab)$. 
    There are two main changes: first, we adapt the dynamic programming formulas and second, we show how to apply \cref{thm:chinese-remainder} to polynomials in two variables that will appear in the proof.
    
    We start with dynamic programming.
    Let $a$ be a node of $T$.
    For \textsc{Maximum Independent Set}, our \emph{guess} $S$ was the set of labels occurring in an independent set of the current subgraph $G_a$.
    Now, instead, we guess a subset $S$ of $\mathbf{States} \coloneqq \{ (\mathbf{h},i) \mid \mathbf{h} \in V(H), i \in [k]\}$.
    For each label $i \in [k]$, the set $S$ is intended to reflect to which vertices of $H$ the set $V_a^i$ is mapped by a homomorphism.
    The set $\mathbf{States}$ has size $|V(H)| \cdot k$, i.e., $\mathcal{O}(k)$ for fixed~$H$.
    So as in \cref{subsec:independent-set}, there are still $2^{\mathcal{O}(k)}$ possibilities for $S$ and this will be the reason for the running time of $2^{\mathcal{O}(d k)} \cdot \polyn$ as in that section.
    As before, we then employ guesses of the form $\alpha \colon S \to \{1_=, 2_\geq\}$ and $\beta \colon S \to \{1_\geq, 1_=\}$ to compute the polynomials reflecting the number of $H$-homomorphisms of certain cardinality via inclusion-exclusion.
    Further, we need to forbid that edges of $G$ are mapped to non-edges of $H$.
    For this, the auxiliary graph $F^{a, \beta}$ again has vertex set $\beta^{-1}(1_=)$ but now there is an edge between two vertices $(\mathbf{h},i)$ and $(\mathbf{h'},j)$ whenever $M_a[i, j] = 1$ and $\mathbf{h h'}$ is not an edge of $H$.
    Then, if a homomoprhism maps a vertex $v_1$ with label $i$ to $\mathbf{h}$ and a vertex $v_2$ with label $j$ to $\mathbf{h'}$, our approach from \cref{subsec:independent-set} ensures that $v_1$ and $v_2$ come from the same child of $a$ so that no edge between $v_1$ and $v_2$ is created at $a$. 
    
    In \cref{subsec:independent-set}, all polynomials had only one variable $x$ whose degree reflected the size of an independent set.
    Here, additionally to cardinality we are interested in the weight of vertices mapped to $H$.
    So instead of univariate polynomials from $\ZZ[x]$, we use polynomials in two variables $x$ and $y$ where the degree of $y$ keeps track of the weight.
    The weights of partial solutions are initialized in the leaves of the tree-model, there we also take care of lists~$L$: the polynomial for a guess $S$ and a leaf $v$ is given by $x \cdot y^{\omega(v)}$ provided $S = \{(\mathbf{h},i)\}$ for some $\mathbf{h} \in L(v)$ and $i = \lab(v)$, and otherwise this polynomial is the zero polynomial.
    
    With this adaptations in hand, by a straightforward implementation of the recursion we can already obtain a $2^{\mathcal{O}(d k)} \cdot \polyn$-time algorithm that uses only polynomial space and computes the polynomial 
        $Q(x, y) = \sum_{p \in [n]_0, w \in [n W^*]_0} q_{p, w} x^p y^w$
    where $q_{p, w}$ is the number of list $H$-homomorphisms of $G$ 
    of cardinality $p$ and weight $w$.
    The answer to the problem is then the value $q_{C, W}$.
    To obtain logarithmic dependency on the graph size in space complexity, in \cref{subsec:independent-set} we relied on \cref{thm:chinese-remainder}.
    However, \cref{thm:chinese-remainder} concerns univariate polynomials, while $Q$ has two variables.
    We now explain how to model $Q$ as a univariate polynomial $P \in \ZZ[t]$ in order to apply the theorem.
    
    Let 
    \[
        P(t) = \sum\limits_{j_1 \in [n]_0, j_2 \in [nW^*]_0} q_{j_1, j_2} t^{j_1(nW^* + 1) + j_2}.
    \]
    First, observe that $j_1$ and $j_2$ form a base $nW^* + 1$ representation of the degree of the corresponding monomial. 
    So the coefficient standing by $t^{C(nW^* + 1) + W}$ in $P$ is exactly $q_{C, W}$, i.e., the value we seek.
    Further, it holds
    \[
        P(t) = \sum_{j_1 \in [n]_0, j_2 \in [n W^*]_0} q_{j_1, j_2} (t^{n W^* + 1})^{j_1} t^{j_2},
    \]
    so evaluating $P$ at some value $s \in \ZZ$ modulo a prime number $p$ is equivalent to computing the value $
        Q(s^{n W^* + 1}, s) \bmod p$.
    
    It remains to choose suitable values to apply \cref{thm:chinese-remainder}.
    The degree of $P$ is bounded by $\mathcal{O}(n^2 W^*)$.
    The number of $H$-homomorphisms of $G$, and hence each coefficient of $P$ as well, is bounded by $|V(H)|^n$.
    Since $|V(H)|$ is a problem-specific constant, there is a value $n'$ of magnitude $\mathcal{O}(n^2 W^*)$ satisfying the prerequisities of \cref{thm:chinese-remainder}.
    Then for a prime number $p \leq 2n'+2$, any value from $\FF_p$ is $\mathcal{O}(\log n + \log W^*)$ bits long.
    Now to compute the value $Q(s^{n W^* + 1}, s) \bmod p$ for some $s \in \FF_p$, we proceed similarly to \cref{subsec:independent-set}: during the recursion, instead of storing all coefficients of the polynomials, as a partial result we only store the current result of the evaluation at $x = s^{n W^* + 1}$ and $y = s$ modulo $p$.

    Let us now summarize the time and space complexity of this evaluation similarly to \cref{subsec:independent-set}.
    The depth of $T$ is $d$ and per node of $T$, there are at most $4 + |\textbf{States}|$ recursive calls where $|\textbf{States}|$ reflects that the transformation from $1_\geq$ to $2_\geq$ is carried out for every element of a guess $S \subseteq \textbf{States}$ (recall the tables $T(\cdot, c=0, \cdot), \dots, T(\cdot, c=|\alpha^{-1}(2_\geq)|, \cdot)$ in \cref{subsec:independent-set}).
    Due to $|\textbf{States}| = k \cdot |V(H)|$, the recursion depth is then $\mathcal{O}(kd)$.
    The number of possible guesses $S$ as  well as reasonable $\alpha$ and $\beta$ is bounded by $2^{\mathcal{O}(\textbf{States})} = 2^{\mathcal{O}(k)}$.
    Also, for a node $a$ and a reasonable $\beta$, the auxiliary graph $F^{a, \beta}$ has at most $|\textbf{States}| = \mathcal{O}(k)$ vertices.
    Recall that in \cref{subsec:independent-set}, at some point of the computation we work with a polynomial using a variable $z$.
    For this variable, only coefficients at monomials $z^i$ for $i \leq |V(F^{a, \beta})|$ are relevant.
    Hence, for each query we need to keep only $\mathcal{O}(k)$ coefficients from $\FF_p$ and such a coefficient uses $\Oh(\log n + \log W^*)$ bits.
    The addition and multiplication of two such coefficients can be done in time $\mathcal{O}(\log n + \log W^*)$.
    These properties imply that following the argument from \cref{subsec:independent-set} we obtain the running time of $2^{\mathcal{O}(dk)} \cdot \polyn \cdot \log W^*$ and space complexity of $\mathcal{O}(kd) \cdot k \cdot \mathcal{O}(\log n + \log W^*) = \mathcal{O}(k^2 d (\log n + \log W^*))$.
    
   With that, \cref{thm:chinese-remainder} implies that the coefficients of $P$, and in particular the sought value $q_{C, W}$, can be reconstructed in time $2^{\mathcal{O}(dk)} \cdot \polyn \cdot \log W^*$ and using $\mathcal{O}(k^2 d (\log n + \log W^*))$ space. This concludes the proof of \cref{thm:counting-weighted-list-homomorphisms}.

\medskip

We remark that the result of \cref{thm:counting-weighted-list-homomorphisms} can be combined with the Cut\&Count technique of Cygan et al.~\cite{CyganNPPRW22} in order to incorporate also connectivity constraints to \textsc{List $H$-Homomorphism} and solve problems like \textsc{Connected Vertex Cover} and \textsc{Connected Odd Cycle Transversal}. In essence Cut\&Count provides a randomized reduction from \textsc{List $H$-Homomorphism} with connectivity constraints to \textsc{\#-List $H'$-Homomorphism} for a new pattern graph $H'$ with at most twice as many vertices as $H$. Since in the reduction only the parity of the number of solutions is preserved, in Cut\&Count one typically uses the Isolation Lemma~\cite{ImpagliazzoPZ01} to sample a weight function so that with high probability, there is exactly one (and thus, an odd number) solution of minimum possible weight; then counting the number of solutions mod $2$ for all possible weights reveals the existence of a solution. Note here that the algorithm of \cref{thm:counting-weighted-list-homomorphisms} is already prepared to count weighted solutions. In our setting, the usage of Isolation Lemma necessitates allowing randomization and adds an $\Oh(n\log n)$ factor to the space complexity for storing the sampled weights. We leave the details to the~reader.
\fi

\iflong
\subsection{Max-Cut}
\fi
\ifshort
\smallskip
\noindent \textbf{Max Cut.} \quad
\fi
In the classical \textsc{Max Cut} problem, we are given a graph $G$ and the task is to output $\max_{X\subseteq V(G)} \abs{E(X, V(G) \setminus X)}$.
Towards solving the problem, let us fix a graph $G$ and a $(d,k)$-tree model  $(T, \mathcal{M}, \cR, \lab)$ of $G$.
Recall that for every node $a$ of $T$, $i\in [k]$ and $X\subseteq V_a$, we denote by $X_a(i)$ the set of vertices in $X$ labeled $i$ at $a$, i.e., $X\cap \lambda_a^{-1}(i)$.
Given a child $b$ of $a$, we let $V_{ab}=V_b$ and we denote by $V_{ab}(i)$ the set of vertices in  $V_b$ labeled $i$ at $a$, \ie, $V_b \cap V_a(i)$.
By $X_{ab}(i)$ we denote the set $X\cap V_{ab}(i)$.
Given $c\in\{a,ab\}$, we define the $c$-\textit{signature} of $X\subseteq V_c$ --- denoted by $\signature_c(X)$ --- as the vector $(\abs{X_c(1)},\abs{X_c(2)},\dots,\abs{X_c(k)})$.
We let $\cS(c)$ be the set of $c$-signatures of all the subsets of $V_c$, \ie, $\cS(c)\defeq\{\signature_c(X)\mid X\subseteq V_c\}$.
Observe that $|\cS(c)| \in n^{\Oh(k)}$ holds.
Also, for the children $b_1,\dots,b_t$ of $a$, we define $\cS(ab_1,\dots,ab_t)$ as the set of all tuples $(s^1,\dots,s^t)$ with $s^i\in \cS(ab_i)$ for each $i\in [t]$.
Given $s\in \cS(c)$, we define~$f_c(s)$ as the maximum of $\abs{E(X,V_c\setminus {X})}$ over all the subsets $X\subseteq V_c$ with $c$-signature $s$.
To solve \textsc{Max Cut} on $G$, it suffices to compute $\max_{s\in \cS(r)} f_r(s)$ where $r$ is the root of~$T$.

Let $b$ be a child of $a$. We start explaining how to compute $f_{ab}(s)$ by making at most $n^{\Oh(k)}$ calls to the function $f_{b}$.
Given $s'\in \cS(b)$, we define $\rho_{ab}(s')$ as the vector $s = (s_1, \dots, s_k)\in \cS(ab)$ such that, for each $i\in [k]$, we have $s_i=\sum_{j\in \rho_{ab}^{-1}(i)} s_j'$.
Observe that for every $X\subseteq V_b$, we have $\signature_{ab}(X)= \rho_{ab}(\signature_b(X))$. Consequently, for every $s\in \cS(ab)$, $f_{ab}(s)$ is the maximum of $f_b(s')$ over the $b$-signatures $s'\in \cS(b)$ such that $\rho_{ab}(s')=s$.
It follows that we can compute~$f_{ab}(s)$ with at most $n^{\Oh(k)}$ calls to the function $f_b$.

\begin{observation}\label{obs:maxcut:fab}
    Given a node $a$ of $T$ with a child $b$ and $s\in \cS(ab)$, we can compute $f_{ab}$ in space $\Oh(k\log(n))$ and time $n^{\Oh(k)}$ with $n^{\Oh(k)}$ oracle access to the function $f_b$.
\end{observation}

In order to simplify forthcoming statements, we fix a node $a$ of $T$ with children $b_1,\ldots, b_t$.
Now, we explain how to compute $f_a(s)$ by making at most $n^{\Oh(k)}$ calls to the functions $f_{ab_1}, \ldots, f_{ab_t}$. 
The first step is to express $f_a(s)$ in terms of $f_{ab_1}, \ldots, f_{ab_t}$.
We first describe $\abs{E(X,V_a\setminus X)}$ in terms of $\abs{E(X\cap V_{b_i},V_{b_i}\setminus X)}$.
We denote by 
$E(V_{b_1},\dots,V_{b_t})$ the set of edges of $G[V_a]$ whose endpoints lie in
different $V_{b_i}$'s, i.e. $E(G[V_{b_1},\dots,V_{b_t}])\setminus (E(G[V_{b_1}] \cup \dots \cup E(G[V_{b_t}])))$.
Given $X\subseteq V_a$, we denote by $E_a(X)$ the intersection of $E(X,V_a\setminus X)$ and $E(V_{b_1},\dots,V_{b_t})$. 
In simple words, $E_a(X)$ is the set of all cut-edges (i.e., between $X$ and $V_a \setminus X$) running between distinct children of $a$.
For $i,j\in [k]$, we denote by $E_a(X,i,j)$ the subset of $E_a(X)$ consisting of the edges whose endpoints are labeled $i$ and $j$.
We capture the size of $E_a(X,i,j)$ with the following notion. For every $c\in \{a,ab_1,\dots,ab_t\}$, $s\in \cS(c)$ and $i,j\in [k]$, we define 
\[ \edgelabel_c(s,i,j) \defeq \begin{cases}
    s_i\cdot (\abs{V_c(j)}-s_j) + s_j\cdot (\abs{V_c(i)}-s_i) \text{ if }i\neq j,\\
    s_i \cdot (\abs{V_c(i)}-s_i) \text{ otherwise.}
\end{cases} \]
It is not hard to check that, for every subset $X\subseteq V_a$ with $a$-signature $s$, $\edgelabel_a(s,i,j)$ is the size of $\pairs_a(X,i,j)$ 
being the set of pairs of distinct vertices in $V_a$ labeled $i$ and $j$ at $a$ such that exactly one of them is in $X$.
Observe that when $M_a[i,j]=1$, then $\abs{E_a(X,i,j)}$ is the number of pairs in $\pairs_a(X,i,j)$ whose endpoints belong to different sets among $V_{b_1},\dots,V_{b_t}$.
Moreover, given a child $b$ of $a$, the number of pairs in $\pairs_a(X,i,j)$ whose both endpoints belong to $V_b$ is exactly $\edgelabel_{ab}(\signature_{ab}(X),i,j)$.
Thus when $M_a[i,j]=1$, we have
\begin{equation}\label{eq:maxcut:edgelabel}
        \abs{E_a(X,i,j)} =  \edgelabel_a(\signature_a(X),i,j) - \sum_{i\in[t]} \edgelabel_{ab_i}(\signature_{ab_i}(X),i,j) \enspace .
\end{equation}
We capture the size of $E_a(X)$ with the following notion. For every $c\in \{a,ab_1,\dots,ab_t\}$, $s\in \cS(c)$ and $(k\times k)$-matrix $M$, we define 
\iflong
 \[ m_c(s,M)\defeq \sum_{\substack{i,j\in [k], i\leq j \\ M[i,j]=1}}   \edgelabel_c(s,i,j). \]
\fi
\ifshort
$m_c(s,M)\defeq \sum_{\substack{i,j\in [k], i\leq j \\ M[i,j]=1}}   \edgelabel_c(s,i,j).$
  \fi
Note that $\abs{E_a(X)} = \sum_{i,j\in [k] \colon i\leq j,  M_a[i,j]=1} \abs{E_a(X,i,j)}$.
Hence, by Equation~\ref{eq:maxcut:edgelabel}, we deduce that $\abs{E_a(X)}=m_a(\signature_a(X),M_a) - \sum_{i\in [t]} m_{ab_i}(\signature_{ab_i}(X),M_a)$.
Since $E(X,V_a\setminus X)$ is the disjoint union of $E_a(X)$ and the sets $E(X\cap V_{b_1},V_{b_1}\setminus X),\dots,E(X\cap V_{b_t},V_{b_t}\setminus X)$ , we deduce:

\begin{observation}\label{lem:maxcut:recursion:anyset}
    For every $X\subseteq V_a$ we have
    \begin{equation*}
        \abs{E(X, V_a \setminus X)} = m_a(\signature_a(X),M_a) + \sum_{i=1}^{t} \left(\abs{E(X_i\cap V_{b_i}, V_{b_i}\setminus X_i)} - m_{ab_i}(\signature_{ab_i}(X_i),M_a)\right) \enspace . 
    \end{equation*}
\end{observation}
 	
We are ready to express $f_a(s)$ in terms of $f_{ab_1},\dots,f_{ab_t}$ and $m_a,m_{ab_1},\dots,m_{ab_t}$.

 	\begin{lemma} 
 \label{lem:maxcut:recursion:fa(s)}
		For every $s\in \cS(a)$, we have 
		\[ f_a(s) = m_a(s,M_a) + \max_{\substack{(s^1,\dots,s^t)\in \cS(ab_1,\dots,ab_t) \\ s=s^1+\dots+s^t}} \left(\sum_{i=1}^{t} \left(f_{ab_i}(s^i) - m_{ab_i}(s^i,M_a)\right)\right) \enspace .\]
	\end{lemma}

\iflong
 	\begin{proof}
		Let $s\in \cS(a)$.
		By \Cref{lem:maxcut:recursion:anyset} we know that
		\[ f_a(s) = \max_{\substack{X\subseteq V_a \\ \signature_a(X)=s}} \abs{E(X, V_a \setminus X)} = m_a(s,M_a) +
		\max_{\substack{X\subseteq V_a \\ \signature_a(X)=s}}  \left( \sum_{i=1}^{t} \abs{E(X_i, V_{b_i}\setminus X_i)} - m_{ab_i}(\signature_{ab_i}(X_i),M_a)\right)  \]
		where $X_i$ is a shorthand for $X\cap V_{b_i}$.
		Observe that for every $X\subseteq V_a$, we have $\signature_a(X)=s$ iff $s=\sum_{i=1}^t \signature_{ab_i}(X\cap V_{b_{i}})$.
		Since $f_{ab_i}(s^i)$ is the maximum $\abs{E(X_i, V_{b_i}\setminus X_i)}$ over all $X_i\subseteq V_{b_i}$ with $ab_i$-signature $s^i$ while $m_{ab_i}(\signature_{ab_i}(X_i),M_a)$ only depends on $s^i$ and not on the concrete choice of $X_i$, we conclude that $f_a(s)$ equals $m_a(s,M_a)$ plus
		\[ \max_{\substack{(s^1,\dots,s^t)\in \cS(ab_1,\dots,ab_t) \\ s=s^1+\dots+s^t}} \left(\sum_{i=1}^{t} f_{ab_i}(s^i) - m_{ab_i}(s^i,M_a)\right).\]
	\end{proof}
 \fi
 	
	To compute $f_a(s)$ we use a twist of Kane's algorithm~\cite{Kane10} for solving the \textsc{$k$-dimensional Unary Subset Sum} in Logspace. 
 \iflong
	The twist relies on using a polynomial, slightly different from  the original work of Kane~\cite{Kane10}, defined in the following lemma.
	
	Given a vector $s=(s_1,\dots,s_k)\in \bZ^k$ and $B\in \bZ$, we denote by $s | B$ the vector $(s_1,\dots,s_k,B)$.
	We denote by $C$ the number $2n^2+1$ and, given a vector $s'\in \bZ^{k+1}$, we denote by $C(s')$ the sum $\sum_{i\in [k+1]} C^{i-1} s'_i$.
	
	\begin{longlemma}
 \label{lem:maxcut:poly}
		Let $s\in \cS(a)$ and $B\in [\abs{E(G[V_a])}]$.
		Let $A(s,B)$ be the number of tuples $(s^1,\dots,s^t)\in \cS(ab_1,\dots,ab_t)$ such that $s=s^1+\dots + s^t$ and 
		\[ B - m_a(s,M_a) = \sum_{j=1}^{t}  f_{ab_j}(s^j) - m_{ab_j}(s^j,M_a). \]
		For every prime number $p > C^{k+1} + 1$,  we have $-A(s,B) \equiv P_{a,s}(B,p)$ (mod $p$) where
		\[ P_{a,s}(B,p) \defeq \sum_{x=1}^{p-1} x^{C(s | B - m_a(s,M_a))}
		\left( \prod_{j=1}^{t} \left(\sum_{s^j\in \cS(ab_j)}  x^{-C(s^j | f_{ab_j}(s^j) - m_{ab_j}(s^j,M_a))}\right)\right).\]
	\end{longlemma}

 \iflong
	\begin{proof}
		First, note that 
		\begin{equation}\label{eq:poly}
			x^{C(s | B - m_a(s,M_a))}
			\left( \prod_{j=1}^{t} \left(\sum_{s^j\in \cS(ab_j)}  x^{-C(s^j | f_{ab_j}(s^j) - m_{ab_j}(s^j,M_a))}\right)\right)
			= \sum_{s^1,\dots,s^t\in \cS(b_1,\dots,b_t)} x^{\alpha(s^1,\dots,s^t)} 
		\end{equation}
		where 
		\[ \alpha(s^1,\dots,s^t)= C(s | B - m_a(s,M_a)) - \sum\limits_{j = 1}^{t} \left(C(s^j | f_{ab_j}(s^j) - m_{ab_j}(s^j,M_a))\right). \]
		As in \cite{Kane10}, the idea of this proof is to change the order of summation, show that the terms where $\alpha(s^1,\dots,s^t)\neq 0$ cancel out, and prove that the sum of the terms where $\alpha(s^1,\dots,s^t)= 0$ is $-A(s,B)$.
		The latter is implied by the following claim.
		
		\begin{longclaim}\label{claim:cancel}
			For every $(s^1,\dots,s^t)\in \cS(ab_1,\dots,ab_t)$, the absolute value of $\alpha(s^1,\dots,s^t)$ is at most $C^{k+1}$.
			Moreover, $\alpha(s^1,\dots,s^t)=0$ iff $s=s^1+\dots+s^t$ and $B - m_a(s,M_a) = \sum_{i=1}^{t} f_{ab_i}(s^i) - m_{ab_i}(s^i,M_a)$.
		\end{longclaim}
		\begin{claimproof}
			By definition of $C(\cdot|\cdot)$, we have 
			\[ \alpha(s^1,\dots,s^t) = \left(\sum_{i=1}^{k} C^{i-1} \left( s_i - \sum_{j=1}^{t} s_i^j \right)\right) + C^{k+1} \left(B - m_a(s,M_a) - \sum_{j=1}^{t} \left( f_{ab_j}(s^j) - m_{ab_j}(s^j,M_a)\right) \right). \]
			I.e.,
			\[
			     \alpha(s^1,\dots,s^t)  = \sum_{i=1}^{k+1} C^{i-1} e_i
			\]
			with 
			\[
			    e_i = 
			    \begin{cases}
			        s_i - \sum_{j=1}^{t} s_i^j & \text{if } 1 \leq i \leq k \\
			        B - m_a(s,M_a) - \sum_{j=1}^{t} \left(f_{ab_j}(s^j) - m_{ab_j}(s^j,M_a)\right) & \text{if } i = k + 1
			    \end{cases}
			    .
			\]
			We claim that the absolute value of each $e_i$ is at most $C-1$.
			For every $i\in [k]$,  by definition, $s_i$ and $\sum_{j=1}^{t} s_i^j$ are at least $0$ and at most $\abs{V_a(i)}\leq n$.
			Hence, for each $i\in [k]$ the absolute value of $e_i$  is at most $n < C$.
			Both $B$ and $\sum_{j=1}^{t} f_{ab_j}(s^j)$ are upper bounded by $\abs{E(G[V_a])}\leq n^2$.
			Moreover, from the definition of the functions $m_a,m_{ab_1},\dots,m_{ab_t}$, we deduce that both $m_a(s,M_a)$ and $\sum_{j=1}^{t} m_{ab_j}(s^j,M_a)$ are upper bounded by $\abs{V_a}^2 \leq n^2$.
			It follows that the absolute value of $e_{k+1}$ is at most $2n^2 < C$.
			Thus, the absolute value of $\alpha(s^1,\dots,s^t) $ is at most $\sum_{i=1}^{k+1} C^{i-1} e_i  \leq \sum_{i=1}^{k+1} C^{i-1} (C - 1) = C^{k+1} - 1$.
   
			It remains to prove that that $\alpha(s^1,\dots,s^t)=0$ iff $e_j=0$ for every $j\in [k+1]$.
            One direction is trivial.
            For the other direction, observe that if $e_{k+1}\neq 0$, then the absolute value of $C^k e_{k+1}$ is at least $C^k$.
            But the absolute value of $\alpha(s^1,\dots,s^t) -C^k e_{k+1} = \sum_{i=1}^{k} C^{i-1} e_i $ is at most $\sum_{i=1}^k C^{i-1} (C-1) = C^k-1$.
            Hence, if $e_{k+1}\neq 0$, then $\alpha(s^1,\dots,s^t)\neq 0$.
            By induction, it follows that $\alpha(s^1,\dots,s^t)=0$ is equivalent to $e_i=0$ for every $i\in [k+1]$.
		\end{claimproof}
		
		By using \Cref{eq:poly} on $P_{a,s}(B,p)$ and interchanging the sums, we deduce that
		\[ 	P_{a,s}(B,p)= \sum_{s^1,\dots,s^t\in \cS(b_1,\dots,b_t)} \left(\sum_{x=1}^{p-1} x^{\alpha(s^1,\dots,s^t)}\right). \]
		It was proven in the proof of Lemma 1 in \cite{Kane10} that
		\[ \sum_{x=1}^{p-1} x^\ell\, \pmod p = 
		\begin{cases}
			-1 \text{ if } \ell \equiv 0\, \pmod{p-1}\\
			0 \text{ otherwise.} 
		\end{cases} \]
		We infer from the above formula that 
		\[ P_{a,s}(B,p) \pmod p = \sum_{\substack{s^1,\dots,s^t\in \cS(ab_1,\dots,ab_t) \\ \alpha(s^1,\dots,s^t)\equiv 0 \pmod{p - 1}}} -1 \pmod p.  \]
		Observe that, for every $(s^1,\dots,s^t)\in \cS(ab_1,\dots,ab_t)$, we have $\alpha(s^1,\dots,s^t)\equiv 0 \pmod{p - 1}$ iff $\alpha(s^1,\dots,s^t)= 0$ because $C^{k+1} < p - 1$ and the absolute value of $\alpha(s^1,\dots,s^t)$ is at most $C^{k+1}$ by \Cref{claim:cancel}.
		From the equivalence given by \Cref{claim:cancel}, we deduce that there are $A(s,B)$ tuples $(s^1,\dots,s^t)\in \cS(ab_1,\dots,ab_t)$ such that $\alpha(s^1,\dots,s^t)= 0$, i.e.,
		\[
		    P_{a,s}(B,p) \pmod p = -A(s,B) \pmod p.
		\]
	\end{proof}
 \fi
	\fi
	With this, we can prove \ifshort Theorem~\ref{thm:maxcut}.\fi \iflong Theorem~\ref{thm:maxcut} via \Cref{algo:maxcut}. As a subroutine, we use the function \textsf{NextPrime}($p$), which computes the smallest prime larger than $p$.\fi

 \iflong
	\begin{algorithm}[bth]
		\SetAlgoLined
		\DontPrintSemicolon
		\KwIn{A internal node $a$ of $T$ and $s\in \cS(a)$.}
		\KwOut{$f_a(s)$}
		\For{$B=\abs{E(G[V_a])}$ to $0$}
		{
			$c\defeq 0$\;
			$p\defeq\textsf{NextPrime}(C^{k+1})$\;
			\While{$c\leq nk\log( n)$}
			{
				\lIf{$P_{a,s}(B,p)\not\equiv 0 $ (mod $p$)}
				{\Return $B$}
				$c\defeq c + \lfloor\log(p)\rfloor$\;
				$p\defeq \textsf{NextPrime}(p)$\;	
			}
		}
		\caption{Algorithm for computing $f_x(s)$.}
		\label{algo:maxcut}
	\end{algorithm}
	
	\begin{longlemma}\label{lem:maxcut:algo}
		Let $s\in \cS(a)$.
		\Cref{algo:maxcut} computes $f_a(s)$ in space $\Oh(k\log(n))$ and time $n^{\Oh(k)}$ with $n^{\Oh(k)}$ oracle access to the functions $f_{ab_1},\dots,f_{ab_t}$.
	\end{longlemma}
	\begin{proof}
		The correctness of \Cref{algo:maxcut} follows from the following claims.
		Let $B$ be an integer between 0 and $\abs{E(G[V_a])}$, and let $A(s,B)$ be the integer defined in \Cref{lem:maxcut:poly}.
		
		\begin{longclaim}\label{claim:A(s)notzero}
			If the algorithm returns $B$, then $A(s,B)\neq 0$.
		\end{longclaim}
		\begin{claimproof} 
			Suppose there exists a prime number $p>C^{k+1}$ such that $P_{a,s}(B,p)\not\equiv 0$ (mod $p$).
			As $P_{a,s}(B,p)\equiv A(s,B)$ (mod $p$)	by \Cref{lem:maxcut:poly}, we have $A(s,B)\neq 0$ and thus there exists $(s^1,\dots,s^t)\in \cS(ab_1,\dots,ab_t)$ such that $s=s^1+\dots + s^t$ and 
			$B - m_a(s,M_a) = \sum_{i=1}^{t}  f_{ab_i}(s^i) - m_{ab_i}(s^i,M_a).$
			From \Cref{lem:maxcut:recursion:anyset}, we deduce that there exists $X\subseteq V_x$ such that $\signature(X)=s^1+\dots+s^t=s$ and $\abs{E(X,V_x\setminus X)}=B$.			
		\end{claimproof}
		
		\begin{longclaim}\label{claim:A(s)zero}
			If $P_{a,s}(B,p)\equiv 0$ (mod $p$) for every value taken by the variable $p$, then $A(s,B)=0$.
		\end{longclaim}
		\begin{claimproof}
			Let $d$ be the product of the values taken by $p$.
			Then $d$ is a product of distinct primes $p$ such that $P_{a,s}(B,p)\equiv 0$ (mod $p$).
			By \Cref{lem:maxcut:poly}, we have $P_{a,s}(B,p)\equiv A(s,B)$ (mod $p$) for every prime $p>C^{k+1}$.
			Therefore, $A(s,B)$ is a multiple of $d$. 
			Observe that $d>2^{c}$ and $c> nk\log(n)$.
			Hence, we have $d> n^{nk}$.
			Since $A(s,B)$ corresponds to the number of tuples $(s^1,\dots,s^t)\in \cS(ab_1,\dots,ab_t)$ that satisfy some properties, we have $A(s,B) \leq \prod_{i=1}^{t}\abs{\cS(ab_i)}\leq n^{nk}$.
			As $d$ divides $A(s,B)$ and $d>A(s,B)$, we conclude that $A(s,B)=0$.
		\end{claimproof}
		
		From Claims \ref{claim:A(s)notzero} and \ref{claim:A(s)zero}, we infer that Algorithm \ref{algo:maxcut} returns $B$ where $B$ is the maximum between $0$ and $\abs{E(G[V_a])}$ such that $A(s,B)\neq 0$.
		By definition of $A(s,B)$ and Lemma \ref{lem:maxcut:recursion:fa(s)}, we conclude that $f_a(s)=B$.
		
		\subparagraph{Complexity.}
		We adapt the arguments used in \cite{Kane10} to prove the complexity of our algorithm. 
		\begin{itemize}
			\item First, the variable $p$ is never more than $n^{\Oh(k)}$.
			Indeed, standard facts about prime numbers imply that there are $nk\log(n)$ prime numbers between $C^{k+1}$ and $(C^{k+1}+nk\log(n))^{\Oh(1)}=n^{\Oh(k)}$.
			Each of these primes causes $c$ to increase by at least 1.
			Thus, each value of $p$ can be encoded with $\Oh(k\log(n))$ bits.
			
			\item Secondly, observe that we can compute $P_{a,s}(p,B) \pmod p$ in space $\Oh(k\log(n))$.
			Recall that
			\[ P_{a,s}(B,p) \defeq \sum_{x=1}^{p-1} x^{C(s | B - m_a(s,M_a))}
			\left( \prod_{i=1}^{t} \left(\sum_{s^i\in \cS(b_i)}  x^{-C(s^i | f_{ab_i}(s^i) - m_{ab_i}(s^i,M_a))}\right)\right).\]
			To compute $P_{a,s}(B,p)$, it is sufficient to keep track of the current value of $x$, the current running total (modulo $p$) and enough information to compute the next term, i.e. $x^{C(s | B - m_a(s,M_a))}$ or $x^{-C(s^i | f_{ab_i}(s^i) - m_{ab_i}(s^i,M_a))}$.
			For that, we need only the current values of $i$ (at most $\log n$ bits) and $s^i$ (at most $k \log n$ bits) and the current running total to compute $C(s | B - m_a(s,M_a))$ (or $C(s^i | f_{b_i}(s^i) - m_{b_i}(s^i,M_a)$) modulo $p$. 
			
			\item Finally, primality testing of numbers between $C^{k+1}$ and $n^{\Oh(k)}$ can be done in space $\Oh(k\log(n))$ via $n^{\Oh(k)}$ divisions, and thus each call to $\textsf{NextPrime}(\cdot)$ can be computed in $n^{\Oh(k)}$ time and $\Oh(k\log(n))$ space. \qedhere
		\end{itemize}
	\end{proof}

We are now ready to prove that one can solve \textsc{Max-Cut}  in time $n^{\Oh(dk)}$ using $\Oh(dk\log(n))$ space. 

    \spaceMC*
	\begin{proof}
		Given $r$ the root of $T$, we solve \textsc{Max-Cut} by computing $\max_{s\in \cS(r)}f_r(s)$.
		For every internal node of $a$ of $T$ with children $b_1,\dots,b_t$, we use \Cref{algo:maxcut} to compute each call of $f_a$ from calls to $f_{ab_1},\dots,f_{ab_t}$. For every internal node $a$ with child $b$, we use \Cref{obs:maxcut:fab} to compute each call of $f_{ab}$ from calls to $f_{b}$.
        Finally, for every leaf $\ell$ of $T$, we simply have $f_\ell(s)=0$ for every $s\in \cS(\ell)$ because $V_\ell$ is a singleton.
		
		First, we prove the running time.
		By \Cref{lem:maxcut:algo}, for each node $a$ with children $b_1,\dots,b_t$ and $s\in \cS(a)$, we compute $f_a(s)$ by calling at most $n^{\Oh(k)}$ times the functions $f_{ab_1},\dots,f_{ab_t}$.
        By \Cref{obs:maxcut:fab}, for each node $b$ with parent $a$ and $s\in\cS(ab)$, we compute $f_{ab}(s)$ by calling at most $n^{\Oh(k)}$ times the function $f_b$.
		Consequently, we call each of these functions at most $n^{\Oh(dk)}$ times in total.
		Since $T$ has $\Oh(n)$ nodes, we conclude that computing $\max_{s\in \cS(r)}f_r(s)$ this way takes $n^{\Oh(dk)}$ time.
		
		Finally, observe that the stack storing the calls to these functions is of size at most $\Oh(d)$.
		Our algorithm solves \textsc{Max Cut} in space $\Oh(dk\log(n))$.
	\end{proof}
	\fi

\iflong
\subsection{Dominating Set}\label{subsec:dominating-set}
\fi
\ifshort
\medskip
\noindent
\textbf{Dominating Set.}\quad
\fi
\iflong
In this section we prove \cref{thm:domset}, which we recall for convenience.

\spaceDS*

The remainder of this section is devoted to the proof of \cref{thm:domset}.
\fi
\ifshort
We now prove Theorem~\ref{thm:domset}.
\fi
Note that \textsc{Dominating Set} cannot be directly stated in terms of $H$-homomorphisms for roughly the following reason.
For $H$-homomorphisms, the constraints are \emph{universal}: every neighbor of a vertex with a certain state must have one of allowed states.
For \textsc{Dominating Set}, there is an \emph{existential} constraint: a vertex in state ``dominated'' must have at least one neighbor in the dominating set.
Also, the state of a vertex might change from ``undominated'' to ``dominated'' during the algorithm.
The techniques we used for $H$-homomorphisms cannot capture such properties.

The problem occurs for other parameters as well.
One approach that circumvents the issue is informally called {\em{inclusion-exclusion branching}}, and was used by Pilipczuk and Wrochna~\cite{PilipczukW18} in the context of {\sc{Dominating Set}} on graphs of low treedepth.
Their dynamic programming uses the states \emph{Taken} (i.e., in a dominating set), \emph{Allowed} (i.e., possibly dominated), and \emph{Forbidden} (i.e., not dominated).
These states reflect that we are interested in vertex partitions into three groups such that there are no edges between \emph{Taken} vertices and \emph{Forbidden} vertices; these are constraints that can be modelled using $H$-homomorphisms for a three-vertex pattern graph $H$.
Crucially, for a single vertex $v$, if we fix the states of the remaining vertices, the number of partitions in which $v$ is dominated is given by the number of partitions where $v$ is possibly dominated minus the number of partitions where it is not dominated, i.e., informally ``\emph{Dominated} = \emph{Allowed} - \emph{Forbidden}''.
\iflong
We will come back to this state transformation later to provide more details. We also remark that the transformed formulation of dynamic programming is exactly what one gets by applying the zeta-transform to the standard dynamic programming for {\sc{Dominating Set}}.
\fi

For technical reasons explained later, our algorithm uses the classic Isolation Lemma:

\begin{theorem}[Isolation lemma, \cite{MulmuleyVV87}]\label{thm:isolation-lemma}
    Let $\mathcal{F} \subseteq 2^{[n]}$ be a non-empty set family over the universe $[n]$.  For each $i \in [n]$, choose a weight $\omega(i) \in [2n]$ uniformly and independently at random. Then with probability at least $1/2$ there exists a unique set of minimum weight in $\mathcal{F}$.
\end{theorem}

Consequently,
we pick a weight function $\omega$ that assigns every vertex a weight from $1, \dots, 2n$ uniformly and independently at random. Storing $\omega$ takes $\mathcal{O}(n \log n)$ space. 
\iflong
By \cref{thm:isolation-lemma}, with probability at least $1/2$ among dominating sets with the smallest possible cardinality there will be a unique one of minimum possible weight.
\fi
\ifshort
The remainder of the algorithm uses only $\mathcal{O}(dk^2 \log n)$ space.
\fi

To implement the above idea, we let the graph $H$ have vertex set $\{\mathbf{T}, \mathbf{A}, \mathbf{F} \}$
standing for \emph{Taken}, \emph{Allowed}, and \emph{Forbidden}.
This graph~$H$ has a loop at each vertex as well as the edges $\mathbf{TA}$ and $\mathbf{AF}$. Further, let $R \coloneqq  \{\mathbf{T}\}$.
Following our approach for $H$-homomorphisms,
for every set $S \subseteq \mathbf{States}$ with $\textbf{States} \coloneqq \{(\mathbf T, 1), (\mathbf F, 1), \dots, (\mathbf T, k), (\mathbf F, k)\}$, every cardinality $c \in [n]_0$, and every weight $w \in [2 n^2]_0$, in time $2^{\mathcal{O}(d k)} \cdot \polyn$ and space $\mathcal{O}(dk^2 \log n)$ (recall that here for the maximum weight $W^*$ we have $W^* \leq 2n$) we can compute the value $a_{S, c, w}$ being the number of ordered partitions $(\widehat T, \widehat F, \widehat A)$ of $V(G)$ satisfying the following properties:
    \begin{enumerate}
        \item there are no edges between $\widehat T$ and $\widehat F$;
        \item $|\widehat T| = c$ and $\omega(\widehat T) = w$; and
        \item for every $i \in [k]$ and $ I \in \{T, F\}$, we have $(\mathbf I, i) \in S$ iff $\widehat I \cap V(i) \neq \emptyset$.
    \end{enumerate}
    Note that we do not care whether vertices of some label $i$ are mapped to $A$ or not. 
        
    After that, we aim to obtain the number of dominating sets of cardinality $c$ and weight~$w$ from values $a_{S, c, w}$. 
    For this we need to transform the ``states'' \emph{Allowed} and \emph{Forbidden} into \emph{Dominated}.
    Above we have explained how this transformation works if we know the state of a single vertex.
    However, now the set $S$ only captures for every label $i$, which states occur on the vertices of label $i$.
    First, the vertices of this label might be mapped to different vertices of $H$.
    And even if we take the partitions where all vertices of label $i$ are possibly dominated and subtract the partitions where all these vertices are not dominated, then we obtain the partitions where \emph{at least one vertex} with label $i$ is dominated.
    However, our goal is that \emph{all vertices} of label $i$ are dominated. 
    So
    the \emph{Dominated = Allowed - Forbidden} equality is not directly applicable here.
    
    Recently, Hegerfeld and Kratsch~\cite{HegerfeldK23} showed that when working with label sets, this equality is in some sense still true modulo $2$.
    On a high level, they show that if we fix a part~$\widehat T$ of a partition satisfying the above properties,  
    then any undominated vertex might be put to any of the sides $\widehat A$ and $\widehat F$.
    Thus, if $\widehat T$ is not a dominating set of $G$, then there is an even number of such partitions and they cancel out modulo $2$.

\ifshort
    We can apply the same transformation to obtain from $a_{S,c,w}$'s the number of dominating sets of size $c$ and weight $w$ modulo 2.
    Isolation lemma implies that with probability at least~$1/2$ for some $w$ this number if non-zero if a dominating set of size $c$ exists.
\fi
\iflong
    Now we follow their ideas to formalize this approach and conclude the construction of the algorithm.
    For $i \in [k]$ and $S \subseteq \{(\mathbf T, 1), (\mathbf F, 1), \dots, (\mathbf T, i), (\mathbf F, i)\}$ we define the value $D^w_i(S)$ as the number of ordered partitions $(\widehat T, \widehat F, \widehat X)$ of $V(G)$ with the following properties:
    \begin{enumerate}
        \item there are no edges between $\widehat T$ and $\widehat F$;
        \item $|\widehat T| = c$ and $\omega(\widehat T) = w$;
        \item for every $j \in [i]$ and $I \in \{T, F\}$, we have $(\mathbf I, j) \in S$ iff $\widehat I \cap V(j) \neq \emptyset$; and
        \item $(V({i+1}) \cup \dots \cup V(k)) \setminus \widehat T$ is dominated by $\widehat T$.
    \end{enumerate}
    
   The following observation is obvious.

   \begin{longclaim}\label{obs:D_ell_table}
        For every $S \subseteq \textbf{States}$, we have $D^{c, w}_k(S) = a_{S, c, w}$.
   \end{longclaim}

   Next, we observe that it suffices to compute values $D^{c,w}_i(S)$ for $i=0$ and $S=\emptyset$.

    \begin{longclaim}
        $D^{c,w}_0(\emptyset)$ is the number of dominating sets of size $c$ and total weight $w$.
    \end{longclaim}
    \begin{proof}
        Consider a partition $(\widehat T, \widehat F, \widehat X)$ counted in $D^{c,w}_0(\emptyset)$.
        Recall that $V(1) \cup \dots \cup V(k) = V(G)$.
        So the fourth property implies that $V(G) \setminus \widehat T$ is dominated by $\widehat T$, i.e., $\widehat T$ is a dominating set of $G$.
        The first property then implies that $\widehat F$ is empty and $\widehat X = V(G) \setminus \widehat T$.
        Finally, by definition of $D^{c,w}_0(\emptyset)$, we know that the size of $\widehat T$ is $c$ and its weight is $w$.
        On the other hand, every dominating set $\widehat T$ of cardinality $c$ and weight $w$ defines a partition $(\widehat T, \emptyset, V(G) \setminus \widehat T)$ counted in $D^{c,w}_0(\emptyset)$.
    \end{proof}

    Finally, we prove that modulo $2$, $D_i^{c,w}(S)$ can be computed from $D_{i+1}^{c,w}(S)$.

    \begin{longclaim}
    \label{claim:ds-inclusion-exclusion}
        For every $i \in [k - 1]_0$ and every $S \subseteq \{(\mathbf T, 1), (\mathbf F, 1), \dots, (\mathbf T, i), (\mathbf F, i)\}$, it holds that 
        \[
            D^{c,w}_i(S) \equiv \sum_{B \subseteq \{(\mathbf T, i+1), (\mathbf F, i+1)\}} D^{c,w}_{i+1}(S \cup B) \mod 2.
        \]
    \end{longclaim}
    \begin{proof}
        We follow the proof idea of Hegerfeld and Kratsch.
        For $B \subseteq \{(\mathbf T, i+1), (\mathbf F, i+1)\}$, let $\AAA_{i+1}(S \cup B)$ be the set of partitions counted in $D^{c,w}_{i+1}(S \cup B)$ (see the definition above).
        Note that we have $\AAA_{i+1}(S \cup B_1) \cap \AAA_{i+1}(S \cup B_2) = \emptyset$ for any $B_1 \neq B_2 \subseteq \{(\mathbf T, i+1), (\mathbf F, i+1)\}$.
        So 
        \[
            \sum\limits_{B \subseteq \{(\mathbf T, i+1), (\mathbf F , i+1)\}} D^w_{i+1}(S \cup B) = 
            \left|\bigcup\limits_{B \subseteq \{(\mathbf T, i+1), (\mathbf F, i+1)\}} \AAA_{i+1}(S \cup B)\right|.
        \]
        Let $\mathcal{L}$ be the set of partitions counted in $D^{c,w}_i(S)$ and let
        $\mathcal{R} = \cup_{B \subseteq \{(\mathbf T, i+1), (\mathbf F, i+1)\}} \AAA_{i+1}(S \cup B)$.
        The goal is to prove $|\mathcal{L}| \equiv |\mathcal{R}| \bmod{2}$.

        By definition of these values we have $\mathcal{L} \subseteq \mathcal{R}$. 
        We claim that the size of $\mathcal{R} \setminus \mathcal{L}$ is even.
        To see this, consider some fixed partition $(\widehat T, \widehat F, \widehat X) \in \mathcal{R} \setminus \mathcal{L}$.
        This is exactly the case if the following properties hold:
        \begin{enumerate}
            \item there are no edges between $\widehat T$ and $\widehat F$;
            \item $|\widehat T| = c$ and $\omega(\widehat T) = w$;
            \item for every $j \in [i]$ and $I \in \{T, F\}$, we have $(I, j) \in S$ iff $\widehat I \cap V(j) \neq \emptyset$; and
            \item the set $(V({i+2}) \cup \dots \cup V(k)) \setminus \widehat T$ is dominated by $\widehat T$ while the set $(V({i+1}) \cup \dots \cup V(k)) \setminus \widehat T$ is not dominated by $\widehat T$,
        \end{enumerate}
        Let $U = V(i+1) \setminus N[\widehat T]$.
        The last property implies that $U$ is non-empty.
        Also let $X' = \widehat X \setminus V(i+1)$ and $F' = \widehat F \setminus V({i+1})$.
        Observe that $N(\widehat T) \cap V({i+1}) \subseteq \widehat X$ due to the first property.
        We claim that if we fix the first set $\widehat T$ of the partition as well as the partition of $V \setminus V({i+1})$ (by fixing $X'$ and $F'$), then the extensions of $(\widehat T, F', X')$ to a partition in $\mathcal{R} \setminus \mathcal{L}$ are exactly the partitions of form
        \begin{equation}\label{eq:ds-extensions}
            \bigl(\widehat T, F' \cup (U \setminus U'), X' \cup (N(\widehat T) \cap V({i+1})) \cup U'\bigr)
        \end{equation}
        for $U' \subseteq U$.
        So informally speaking, if we fix $\widehat T, X', F'$, every vertex of $U$ can be put to either $\widehat X$ or $\widehat F$ thus giving rise to an even number $2^{|U|}$ of such extensions.

        Now we prove this claim following the idea of Hegerfeld and Kratsch.
        First, consider a partition of form \eqref{eq:ds-extensions} for an arbitrary $U' \subseteq U$.
        Since $\widehat T$ is fixed and the partition on $V \setminus V({i+1})$ is fixed as well, the last three properties defining $\mathcal{R} \setminus \mathcal{L}$ trivially hold.
        Next, due to $F' \subseteq \widehat F$, 
        there are no edges between $\widehat T$ and $F'$. 
        And since $U \setminus U' \subseteq U$ is not dominated by $\widehat T$, there are no edges between $\widehat T$ and $U \setminus U'$ as well, so the first property holds too.

        For the other direction, if we consider an extension $(\widehat T, \widetilde F, \widetilde X) \in \mathcal{R} \setminus \mathcal{L}$ of $(\widehat T, F', X')$, then by the first property we know that $\widetilde F \cap V(i+1)$ has no edges to $\widehat T$ and hence, it is a subset of $U$.

        So, for any fixed $(\widehat T, F', X')$, either there is no extension to a partition from $\mathcal{R} \setminus \mathcal{L}$ at all or there are $2^{|U|}$ of them where $U$ is a non-empty set.
        So the size of $\mathcal{R} \setminus \mathcal{L}$ is even and this concludes the proof.
     \end{proof}

     The application of \cref{claim:ds-inclusion-exclusion} for $i = 0, \dots, k-1$ implies \begin{align*}
        D^{c,w}_0(\emptyset) \equiv &\sum_{B_1 \subseteq \{(\mathbf T, 1), (\mathbf F, 1)\}} D^{c,w}_1(B_1) \equiv \\
        & \sum_{B_1 \subseteq \{(\mathbf T, 1), (\mathbf F, 1)\}} \sum_{B_2 \subseteq \{(\mathbf T, 2), (\mathbf F, 2)\}} D^{c,w}_2(B_1 \cup B_2) \equiv \\
        & \dots \\
        & \sum_{B_1 \subseteq \{(\mathbf T, 1), (\mathbf F, 1)\}} \sum_{B_2 \subseteq \{(\mathbf T, 2), (\mathbf F, 2)\}} \dots \sum_{B_k \subseteq \{(\mathbf T, k), (\mathbf F, k)\}} D^{c,w}_k(B_1 \cup B_2 \dots \cup B_k) \equiv \\
        & \sum\limits_{S \subseteq \{(\mathbf T, 1), (\mathbf F, 1), \dots, (\mathbf T, k), (\mathbf F, k)\}} D^w_k(S) \mod{2}.
     \end{align*}
     By \cref{obs:D_ell_table}, the parity of the number of dominating sets of size $c$ and weight $w$ can be expressed as
     \[
        D^{c,w}_0(\emptyset) \equiv \sum_{S \subseteq \{(\mathbf T, 1), (\mathbf F, 1), \dots, (\mathbf T, k), (\mathbf F, k)\}} a_{S, c, w} \mod{2}.
     \]
     Recall that every $a_{S, c, w}$ can be computed in time $2^{\mathcal{O}(d k)} \cdot \polyn$ and space $\mathcal{O}(dk^2 \log n)$, hence this is also the case for their sum modulo 2.
     We compute the value $D^{c, w}_0(\emptyset)$ for all cardinalities $c \in [n]_0$ and all weights $w \in [2n^2]_0$ and output the smallest value $c$ such that for some $w$ the value $D^{c, w}_0(\emptyset)$ is non-zero (or it outputs $n$ if no such value exists).
     
     Now we argue the correctness of our algorithm.
     Let $C$ denote the size of the smallest dominating set of $G$.
     First, this implies that for any $c < C$ and any $w \in [2n^2]_0$, the value $D^{c, w}_0(\emptyset)$ is zero.
     And second, Isolation Lemma (\cref{thm:isolation-lemma}) implies that with probability at least $1/2$, the weight function $\omega$ isolates the family of dominating sets of $G$ of size $C$, i.e., there exists a weight $W$ such that there is exactly one dominating set of size $C$ and weight~$W$, and therefore $D^{c,w}_0(\emptyset) = 1$.
     In this case, the algorithm outputs $C$.
     So with probability at least $1/2$ our algorithm outputs the minimum size of a dominating set of $G$.
     
     The iteration over all $c$ and $w$ increases the space complexity by an additive $\mathcal{O}(\log n)$ and it increases the running time by a factor of $\mathcal{O}(n^2)$.
     Recall that in the beginning, to sample the weight function we have used space $\mathcal{O}(n \log n)$.
     So all in all, the running time of the algorithm is $2^{\Oh(dk)}\cdot n^{\Oh(1)}$ and the space complexity is $\Oh(dk^2\log n + n\log n)$. This concludes the proof of \cref{thm:domset}.

     \medskip
     Note that in our algorithm, the only reason for super-logarithmic dependency on $n$ in the space complexity is the need to sample and store a weight function in order to isolate a minimum-weight dominating set. We conjecture that this can be avoided and ask:

\fi
    \begin{question}
    Is there an algorithm for \textsc{Dominating Set} of $n$-vertex graphs provided with a $(d, k)$-tree-model that runs in time $2^{\mathcal{O}(k d)} \cdot \polyn$ and uses  $(d+k)^{\Oh(1)} \log n$ space?
    \end{question}

\section{The Lower Bound}\label{sec:lb}
    \def\Inf{\mathsf{Inf}}
    \def\Match{\mathsf{Match}}
    \def\sfM{\mathsf{M}}
    \def\sfS{\mathsf{S}}
    \def\IJ{\mathcal{M}}
    \def\obj{\mathsf{goal}}
    
    In this section, we prove \Cref{thm:lcs-lb}.
    This lower bound is based on a reasonable conjecture on the complexity of the problem \textsc{Longest Common Subsequence (LCS)}.
    
    An instance of \textsc{LCS} is a tuple $(N,t,\Sigma,s_1,\dots,s_r)$ where $N$ and $t$ are positive integers, $\Sigma$ is an alphabet and  $s_1,\dots,s_r$ are $r$ strings over $\Sigma$ of length $N$. 
    The goal is to decide whether there exists a string $s \in \Sigma^t$ of length $t$ appearing as a subsequence in each $s_i$.
    There is a standard dynamic programming algorithm for \textsc{LCS} that has time and space complexity~$\Oh(N^{r})$.
    \iflong From the point of view of parameterized	complexity, \textsc{LCS} is 
    $\mathsf{W}[p]$-hard for every level $p$ when parameterized by $r$~\cite{BodlaenderDFW95}. It remains $W[1]$-hard when the size of the alphabet is constant~\cite{Pietrzak03}, and it is $\mathsf{W}[1]$-complete when parameterized by
    $r+t$~\cite{Guillemot11}.\fi
    Abboud et al.~\cite{AbboudBW15} proved that the existence of an algorithm with running time $\Oh(N^{r-\varepsilon})$ for any $\varepsilon > 0$ would contradict the Strong Exponential-Time Hypothesis. As observed by Elberfeld et al.~\cite{ElberfeldST15}, \textsc{LCS} parameterized by $r$ is complete for the class $\mathsf{XNLP}$: parameterized problems solvable by a nondeterministic Turing machine using $f(k)\cdot n^{\Oh(1)}$ time and $f(k)\log n$ space, for a computable function $f$. \iflong See also~\cite{BodlaenderDFW95,BodlaenderGJJL22,BodlaenderGJPP22,BodlaenderGNS21,BodlaenderCP22} for further research on $\mathsf{XNLP}$ and related complexity classes.\fi
    The only known progress on the space complexity is due to Barsky et al.\ with an algorithm running in $\Oh(N^{r-1})$ space \cite{BarskySTU07}.
    This motivated Pilipczuk and Wrochna to formulate the following conjecture~\cite{PilipczukW18}.
    
    \begin{conjecture}[\cite{PilipczukW18}]\label{conj:LCS}
        There is no algorithm  that solves the \textsc{LCS} problem in time $M^{f(r)}$ and using $f(r)M^{\Oh(1)}$ space for any computable function $f$, where $M$ is the total bitsize of the instance and $r$ is the number of input strings.
    \end{conjecture}

    Note that in particular, the existence of an algorithm with time and space complexity as in \cref{conj:LCS} implies the existence of such algorithms for all problems in the class $\mathsf{XNLP}$.
    
    Our lower bound is based on the following stronger variant of \Cref{conj:LCS}, in which we additionally assume that the sought substring is short.
    
    \begin{conjecture}\label{conj:main}
        For any unbounded and computable function $\delta$, \Cref{conj:LCS} holds even when $t \leq \delta(N)$.
    \end{conjecture}
    
    \iflong
    Thus, we may rephrase \Cref{thm:lcs-lb} as follows.
    
    \begin{theorem}\label{thm:lowerbound}
        Unless \Cref{conj:main} fails, for any unbounded and computable function~$\delta$, 
        there is no algorithm that solves the {\sc{Independent Set}} problem in graphs supplied with $(d,k)$-tree-models satisfying $d\leq \delta(k)$ that would run in time $2^{\Oh(k)}\cdot n^{\Oh(1)}$ and use $n^{\Oh(1)}$ space.
    \end{theorem}

    The remainder of this section is devoted to the proof of \cref{thm:lowerbound}. Not surprisingly, we provide a reduction from {\sc{LCS}} to {\sc{Independent Set}} on graphs provided with suitable tree-models.
    \fi

    \ifshort
    Let $(N,t,\Sigma,s_1,\dots,s_r)$ be an instance of \textsc{LCS}. We assume, without loss of generality, that~$N$ is a power of $2$.  
    We provide a reduction from $(N,t,\Sigma,s_1,\dots,s_r)$ to an equivalent instance of \textsc{Independent Set} consisting of a graph $G$ with $(r+t+N)^{\Oh(1)}$ vertices which admits a $(d,k)$-tree-model where $d=\Oh(\log t)$ and $k=\Oh(r \log N)$.
    This implies \Cref{thm:lcs-lb} since for every unbounded and computable function $\delta$ there exists an unbounded and computable function $\delta'$  such that if $t \leq \delta'(N)$, then $d \leq  \delta(k)$ for all sufficiently large $N,r\in \bN$.    
    \fi

    \iflong
    Let $(N,t,\Sigma,s_1,\dots,s_r)$ be an instance of \textsc{LCS}.
    For the sake of clarity, we assume without loss of generality that $N$ is a power of $2$.
    Indeed, we can always obtain an equivalent instance $(2^{\lceil \log N \rceil}, t + t', \Sigma', s_1',\dots,s_r')$ where $t'=2^{\lceil \log N \rceil} - N$, $\Sigma'$ is obtained from $\Sigma$ by adding a new letter $\spadesuit$ and each $s_i'$ is obtained by adding $t'$ times $\spadesuit$ at the end of $s_i$.
    
    For every $I\in [N]$, we denote the $I$-th letter of $s_p$ by $s_p[I]$.
    In the following, we present our reduction from $(N,t,\Sigma,s_1,\dots,s_r)$ to an equivalent instance of \textsc{Independent Set} consisting of a graph $G$ with $(r+t+N)^{\Oh(1)}$ vertices and a $(d,k)$-tree-model where $d=\Oh(\log t)$ and $k=\Oh(r \log N)$.
    This implies \Cref{thm:lowerbound} since for every unbounded and computable function $\delta$ there exists an unbounded and computable function $\delta'$  such that if $t \leq \delta'(N)$, then $d \leq  \delta(k)$ (we explain this in more details at the end of this section).
    \fi

    \ifshort
    To outline the main idea of the reduction, let $s^\star$ be a potential common substring of $s_1, \dots, s_r$ of length $t$.
    We use matchings to represent the binary encoding of the positions of the letters of $s^\star$
    in each string.
    \fi

    \iflong
    In the intuitions along the construction, we denote by $s^\star$ a potential common substring of $s_1, \dots, s_r$ of length $t$.
    The main idea is to use matchings to represent the binary encoding of the positions of the letters of $s^\star$
    in each string.
    \fi
      
    For every string $s_p$ and $q\in [t]$, we define the \textbf{selection gadget} $\sfS_{p}^{q}$ which contains, for every $i\in [\log N]$, an edge called the \textit{$i$-edge} of $\sfS_{p}^{q}$. One endpoint of this edge is called the {\em{0-endpoint}} and the other is called the {\em{1-endpoint}}; i.e., a selection gadget induces a matching on $\log N$ edges.
    This results in the following natural bijection between $[N]$ and the maximal independent sets of $\sfS_{p}^{q}$.
    For every $I\in [N]$, we denote by $\sfS_{p}^{q}|I$ the independent set that contains, for each $i\in [\log N]$, the $x$-endpoint of the $i$-edge of $\sfS_{p}^{q}$ where $x$ is the value of the $i$-th bit of the binary representation of $I-1$ (we consider the first bit to be the most significant one and the $\log N$-th one the least significant).
    Then the vertices selected in $\sfS_p^q$ encode the position of the $q$-th letter of $s^\star$ in $s_p$.
      
    We need to guarantee that the selected positions in the gadgets $\sfS_{p}^{1},\dots,\sfS_{p}^{t}$ are coherent, namely, for every $q\in [t]$, the position selected in $\sfS_{p}^{q}$ is strictly smaller than the one selected in~$\sfS_{p}^{q+1}$.
    For this, we construct an \textbf{inferiority gadget} denoted by $\Inf(p,q)$ for every string~$s_p$ and every $q\in [t-1]$.
    The idea behind it is to ensure that the only possibility for an independent set to contain at least $3\log N$ vertices from $\sfS_{p}^{q},\sfS_{p}^{q+1}$, and their inferiority gadget, is the following: there exist $I < J\in [N]$ such that the independent set contains $\sfS_{p}^{q}|I \cup \sfS_{p}^{q+1}|J$.
    The maximum solution size in the constructed instance of \textsc{Independent Set}---which is the sum of the independence number of each gadget---will guarantee that only such selections are possible.
    \ifshort
    We refer to the full version of this paper for the construction of these inferiority gadgets and the arguments proving the following observation.
    \fi
        
    \iflong
    \begin{figure}[th]
        \centering
        \includegraphics[width=0.85\linewidth]{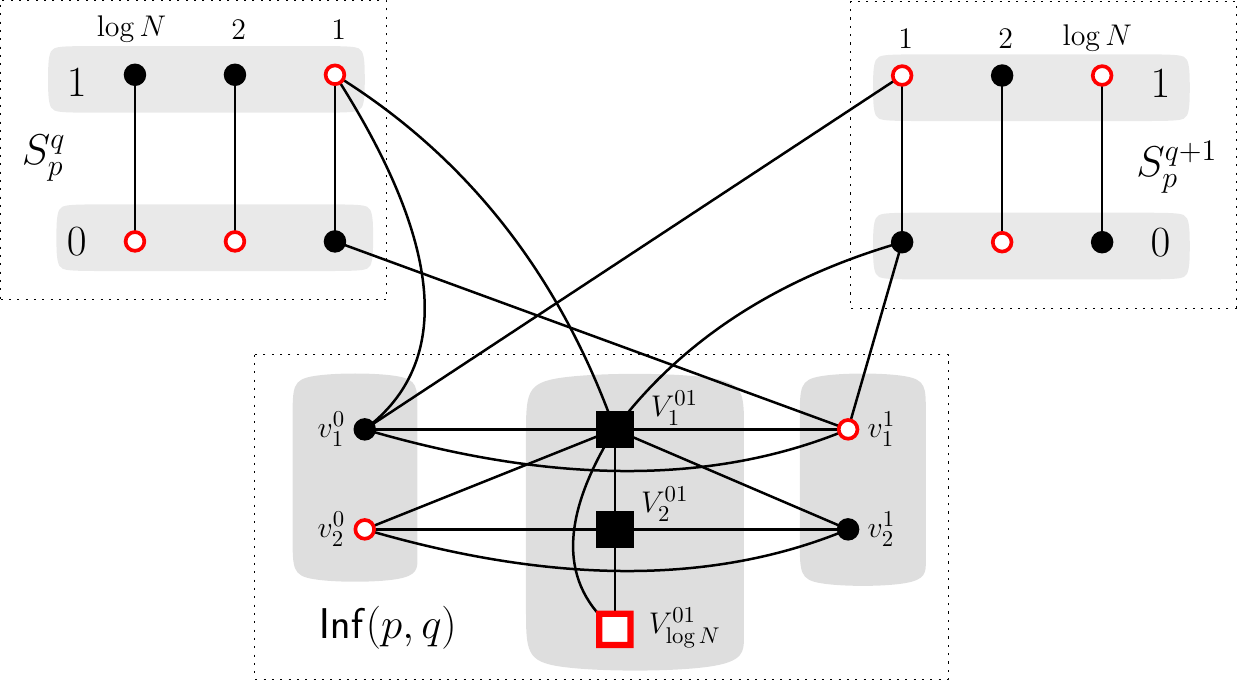}
        \caption{Example of an inferiority gadget with $\log N= 3$. The squares represent the sets of vertices $V^{01}_{i}$. For legibility, among the edges going out of the inferiority gadget, we only represent those incident to $v_{1}^{0},v_{1}^{1}$ and $V^{01}_{1}$. 
            The independent set consisting of the white filled vertices has $3\log N$ vertices, the position selected for $\sfS_{p}^{q}$ is 5 and for $\sfS_{p}^{q+1}$ it is 6.
        }
        \label{fig:infgadget}
    \end{figure}
    
    \Cref{fig:infgadget} provides an example of the following construction.
    The vertex set of $\Inf(p,q)$ consists of the following vertices: for each $i\in [\log N-1]$, there are two vertices $v_{i}^{0,p,q}$ and $v_{i}^{1,p,q}$.
    Moreover, for each $i\in [\log N]$, there is a set $V^{01}_{i,p,q}$ of $\log N -i +1$ vertices (we drop $p,q$ from the notation when they are clear from the context).		
    We now describe the edges incident to the inferiority gadget:
    \begin{itemize}
        \item For every $i\in [\log N-1]$, $v_i^{0}$ and $v_i^{1}$ are adjacent and for each $x\in\{0,1\}$, $v_i^{x}$ is adjacent to the $(1-x)$-endpoints of the $i$-edges from $\sfS_{p}^{q}$ and $\sfS_{p}^{q+1}$. 
        
        \item For every $i\in[\log N]$, all the vertices in $V^{01}_{i}$ are adjacent to (1)~the $1$-endpoint of the $i$-edge from $\sfS_{p}^{q}$, (2)~the $0$-endpoint from the $i$-edge of $\sfS_{p}^{q+1}$, (3)~all the vertices $v_j^{0},v_j^{1}$ for every $j\geq i$ and (4)~all the vertices in $V_{j}^{01}$ for every $j> i$.
    \end{itemize}
    On a high level, an inferiority gadget reflects that for values $I < J \in [N]$, if we go from high-order to low-order bits, then the binary encodings of $I$ and $J$ first contains the same bits and then there is an index, where $I$ has a zero-bit and $J$ has a one-bit.
    If such a difference first occurs at some position $\ell \in [\log N]$, then the corresponding independent set first takes $\ell-1$ vertices of the form $v_{\ell'}^{0}$ or $v_{\ell'}^{1}$ (for $\ell' < \ell$) and then takes $\log N - (\ell - 1)$ vertices from $V^{01}_\ell$ -- this results in $\log N$ vertices taken in the inferiority gadget. The following statement follows from 
    \fi
    
    \begin{observation}\label{obs:infgadget}
        Let $p\in [r]$ and $q\in [t-1]$. The independence number of $\Inf(p,q)$ is $\log N$ and for every $I, J\in [N]$, we have $I < J$ iff there exists a set of $\log N$ vertices $S$ from $\Inf(p,q)$ such that the union of $S$, $\sfS_{p}^{q}|I$ and $\sfS_{p}^{q+1}|J$ induces an independent set.
    \end{observation}

    Next, we need to ensure that the $t$ positions chosen in $s_1, \dots, s_r$ indeed correspond to a common subsequence, i.e., for every $q \in [t]$, the $q$-th chosen letter must be the same in every $s_1, \dots, s_r$. 
    For $p\in [r-1]$, let $\IJ_{p}$ denote the set of all ordered pairs $(I,J)\in [N]^2$ such that the $I$-th letter of $s_p$ and the $J$-th of $s_{p+1}$ are identical.
    For each $p\in [r-1]$ and $q\in [t]$, we create the \textit{matching gadget} \textit{$\Match(p,q)$} as follows:
    \begin{itemize}
        \item For every pair $(I,J)\in \IJ_{p}$ and for each $p^\star\in \{p,p+1\}$, we create a copy $\sfM_{p^\star}^{p,q,I,J}$ of $\sfS_{p^\star}^{q}$ and for every $\ell\in[\log N]$ and $x\in\{0,1\}$, we add an edge between the $x$-endpoint of the $\ell$-edge of $\sfS_{p^\star}^{q}$ and the $(1-x)$-endpoint of the $\ell$-edge of $\sfM_{p^\star}^{p,q,I,J}$.
        
        \item For every pair $(I,J)\in \IJ_{p}$, we add a new vertex $v^{q}_{p,I,J}$ adjacent to (1)~all the vertices from $\sfM_{p}^{p,q,I,J}$ that are not in $\sfM_{p}^{p,q,I,J}|I$ and (2)~all the vertices from $\sfM_{p+1}^{p,q,I,J}$ that are not in $\sfM_{p+1}^{p,q,I,J}|J$.
    \end{itemize}

    Finally, we turn $\{v^{q}_{p,I,J} \mid (I,J)\in \IJ_{p}\}$ into a clique.
    Observe that, for each $p^\star\in \{p,p+1\}$, an independent set $S$ contains $(\abs{\IJ_{p}}+1)\log N$ vertices from $\sfS_{p^\star}^{q}$ and its copies $\sfM_{p^\star}^{p,q,I,J}$ if and only if there exists a value $I\in [N]$ such that $S$ contains $\sfS_{p^\star}^{q}|I$ and $\sfM_{p^\star}^{p,q,I,J}|I$ for each copy.
    This leads to the following observation.
    
    \begin{observation}\label{obs:matching:gadget}
        Let $p\in [r-1]$ and $q\in [t]$. The independence number of $\Match(p,q)$ is $1+2\cdot\abs{\IJ_{p}} \cdot \log N$ and for every $I,J\in [N]$, we have $(I,J)\in \IJ_p$ iff there exists an independent set $S$ of $\Match(p,q)$ with $1+2\abs{\IJ_{p}} \cdot \log N$ vertices such that the union of $S$, $\sfS_{p}^{q}|I$ and $\sfS_{p+1}^{q}|J$ is an independent set.
    \end{observation}

    This concludes the construction of the graph $G$. See
    \Cref{fig:overview:reduction} below for an overview.
    
    \begin{figure}[th]
        \centering
        \includegraphics[width=0.85\linewidth]{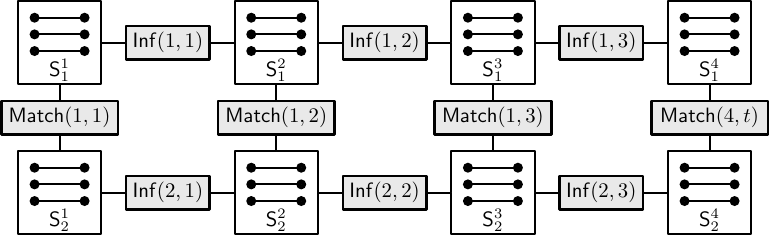}
        \caption{Overview of the graph $G$ with $\log N = 3$, $r=2$ and $t=4$. There are some edges between two gadgets if and only if there are some edges between their vertices in $G$.}
        \label{fig:overview:reduction}
    \end{figure}

    \ifshort
    We prove correctness of the reduction in the following lemma which follows mostly from Observations~\ref{obs:infgadget} and~\ref{obs:matching:gadget}.
    \fi
    \iflong
    We prove the correctness of the reduction in the following lemma.
    \fi
    
    \begin{lemma}\label{lem:reduc:correctness}
        There exists an integer $\obj$ such that $G$ admits an independent set of size at least $\obj$ iff the strings $s_1,\dots,s_r$ admit a common subsequence of length $t$.
    \end{lemma}
    \iflong
    \begin{proof}
        Let $\obj=  (rt + r(t-1) )\log N +  \sum_{p\in [r-1]} t ( 1 + 2\cdot \abs{\IJ_{p}} \cdot \log N )$.
        
        \subparagraph{($\Rightarrow$)} Assume that $s_1,\dots,s_r$ admit a common subsequence $s^\star$ of length $t$.
        Then, for every string $s_p$, there exist $I_{p}^{1},\dots,I_{p}^{t}\in [N]$ such that $I_{p}^{1} < \dots < I_{p}^{t}$ and $s_p[I_p^{q}]= s^\star[q]$ for every $q\in [t]$.
        We construct an independent set $S$ as follows.
        For every selection gadget $\sfS_{p}^{q}$, we add $\sfS_{p}^{q}| I_{p}^{q}$ to $S$.
        Note that, at this point, $S$ is an independent set because there is no edge between the selection gadgets $\sfS_{p}^{q}$ in $G$.
        For every inferiority gadget $\Inf(p,q)$, since $I_{p}^{q} < I_{p}^{q+1}$, we can use \Cref{obs:infgadget} and add a set of $\log N$ vertices from $\Inf(p,q)$ to $S$.
        Note that $S$ remains an independent set because the added vertices are not adjacent to $\sfS_{p}^{q}| I_{p}^{q}$ and $\sfS_{p}^{q}| I_{p}^{q+1}$ by \Cref{obs:infgadget} and the only edges going out of $\Inf(p,q)$ are incident to $\sfS_{p}^{q}$ and $\sfS_{p}^{q+1}$.
        At this point, we have $(rt + r(t-1))\log N$ vertices in $S$.
        
        Observe that for every $p\in [r-1]$ and $q\in [t]$, we have $s_p[I_p^{q}]= s^\star[q] = s_{p+1}[I_{p+1}^{q}]$.
        Thus, we have $(I_{p}^{q},I_{p+1}^{q})\in \IJ_p$ and by \Cref{obs:matching:gadget}, there exists an independent set $S_{p,q}$ of $\Match(p,q)$ with  $1+2\abs{\IJ_{p}} \cdot \log N$ vertices such that the union of $S_{p,q}$, $\sfS_{p}^{q}| I_{p}^{q}$ and $\sfS_{p+1}^{q}| I_{p+1}^{q}$ is an independent set.
        We add $S_{p,q}$ to $S$ and note that $S$ remains an independent set since the only edges going out of $\Match(p,q)$ are incident to $\sfS_{p}^{q}$ and $\sfS_{p+1}^{q}$.
        As we do this for every $p\in [r-1]$ and $q\in [t]$, the union of the $S_{p,q}$'s contains $\sum_{p\in [r-1]} t ( 1 + 2\abs{\IJ_{p}} \cdot \log N )$ vertices.
        We conclude that $G$ admits an independent set of size $\obj$.
        
        \subparagraph{($\Leftarrow$)} Assume $G$ admits an independent set $S$ of size at least $\obj$.
        The independence number of each selection gadget $\sfS_{p}^{q}$ is $\log N$ and, by \Cref{obs:infgadget}, this is also the case for each inferiority gadget $\Inf(p,q)$.
        Hence, $S$ contains at most $(rt + r(t-1) )\log N $ vertices from selection and inferiority gadgets.
        By \Cref{obs:matching:gadget}, the independence number of each matching gadget $\Match(p,q)$ is $ 1 + 2\abs{\IJ_{p}} \cdot \log N$, thus $S$ contains at most $\sum_{p\in [r-1]} t ( 1 + 2\abs{\IJ_{p}} \cdot \log N )$ vertices from the matching gadgets.
        From the definition of $\obj$, we obtain that $S$ contains exactly $\log N$ vertices from each selection and inferiority gadget, and it contains exactly $1 + 2\abs{\IJ_{p}} \cdot \log N$ vertices from each matching gadget.
        We make the following deductions:
        \begin{itemize}
            \item For each $p\in [r]$, there exist $I_{p}^{1},\dots,I_{p}^{t}\in [N]$ such that $S$ contains $\sfS_{p}^{q}|I_{p}^{q}$ for every $q\in [t]$.
            
            \item For each $p\in [r]$ and $q\in[t-1]$, the independent set $S$ contains the vertices in $\sfS_{p}^{q}|I_{p}^{q}$ and $\sfS_{p}^{q}|I_{p}^{q+1}$ as well as $\log N$ vertices from $\Inf(p,q)$. \Cref{obs:infgadget} implies that $I_{p}^{q}< I_{p}^{q+1}$. Thus, $s_{p}[I_{p}^{1}]\dots s_{p}[I_{p}^{t}]$ is a subsequence of $s_p$.
            
            \item For every $p\in [r-1]$ and $q\in [t]$, the independent set $S$ contains $\sfS_{p}^{q}|I_{p}^{q}$ and $\sfS_{p+1}^{q}|I_{p+1}^{q}$ as well as  $ 1 + 2\abs{\IJ_{p}} \cdot \log N $ vertices from $\Match(p,q)$.
            We deduce from \Cref{obs:matching:gadget} that $(I_{p}^{q},I_{p+1}^{q})\in \IJ_p$ and consequently, $s_{p}[I_{p}^{q}]=s_{p+1}[I_{p+1}^{q}]$.
            Hence, for every $q\in [t]$, we have $s_{1}[I_{1}^{q}]=\dots= s_{r}[I_{r}^{q}]$.
        \end{itemize}
        We conclude that $s_{1}[I_{1}^{1}]\dots s_{1}[I_{1}^{t}]$ is a common subsequence of $s_1,\dots,s_r$.
    \end{proof}
    \fi
    
    The next step is to construct a tree-model of $G$.
    
    \begin{lemma}\label{lem:reduc:shrubdepth}
        We can compute in polynomial time a $(d,k)$-tree-model of $G$ where $d=2\log t + 4$ and $k=14r\log N-3$.
    \end{lemma}
    \ifshort
    \begin{proof}[Sketch of proof]
        First, we prove that the union of the gadgets associated with a position $q\in [t]$ admits a simple tree-model.
        For every $q\in [t]$, we denote by $G^{q}$ the union of the selection gadgets $\sfS_{p}^{q}$ with $p\in [r]$ and the matching gadgets $\Match(p,q)$ with $p\in [r-1]$.
        
        For each $q\in [t]$, we prove that $G^q$ admits a $(3,k)$-tree-model $(T^q,\cM^q,\cR^q,\lambda^q)$ where the tree $T^q$ is constructed as follows.        
        We create the root $a^{q}$ of $T^q$ and we attach all the vertices in the selection gadgets $\sfS_{p}^{q}$ with $p\in [r]$ as leaves adjacent to $a^{q}$.
        Then, for every $p\in [r-1]$, we create a node~$a^{q}_{p}$ adjacent to $a^q$ and for every $(I,J)\in \IJ_p$, we create a node~$a^{q}_{p,I,J}$ adjacent to $a^{q}_{p}$. 
        For each $(I,J)\in\IJ_p$, we make $a^{q}_{p,I,J}$ adjacent to the vertex $v^{q}_{p,I,J}$ and all the vertices in $\sfM^{p,q,I,J}_{p}$ and $\sfM^{p,q,I,J}_{p+1}$. 
        Note that all the vertices in $\Match(p,q)$ are the leaves of the subtree rooted at $a^{q}_{p}$, and the leaves of $T^q$ are exactly the vertices in $G^q$. 
        See \Cref{fig:tree-model-reduction} for an illustration of $T^{q}$.

        \begin{figure}[bht]
            \centering
            \includegraphics[width=0.6\linewidth]{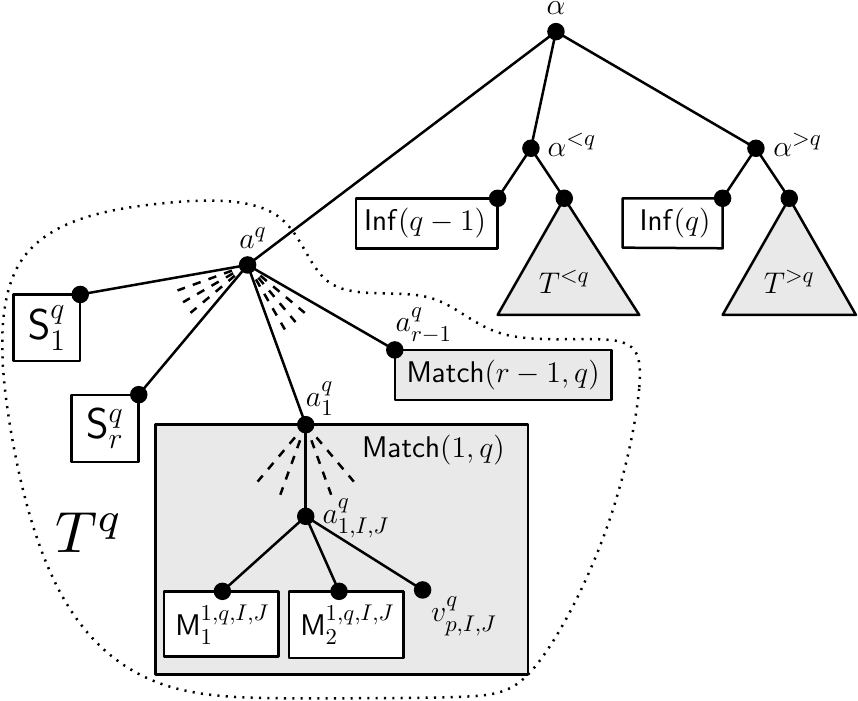}
            \caption{Illustration of the tree $T$ and its subtree $T^{q}$ for the tree-model constructed in \Cref{lem:reduc:shrubdepth}.
                An edge between a white filled rectangle labeled $\mathsf{X}$ and a node $a$ of the tree means that all the vertices in $\mathsf{X}$ are leaves adjacent to $a$. }
            \label{fig:tree-model-reduction}
        \end{figure}

        For every $q\in[t-1]$, we denote by $\Inf(q)$ the union of $\Inf(1,q),\dots,\Inf(r,q)$.
        Moreover, for every interval $[x,y]\subseteq  [t]$, we denote by $G^{x,y}$, the union of the graphs $G^q$ over $q\in [x,y]$, and the inferiority gadgets in $\Inf(q)$ over $q\in [x,y]$ such that $q+1\in [x,y]$.
        
        For every interval $[x,y]$, we prove by induction on $y-x$ that $G^{x,y}$ admits a $(2\log(y-x+1)+4,k)$-tree-model.
        In particular, it implies that $G^{1,t}=G$ admits a $(d,k)$-tree-model. 
        It is also easy to see from our proof that this $(d,k)$-tree-model is computable in polynomial time. For the base case of the induction, when $y=x$, we have $G^{x,y}=G^x$ and we have proved above that it admits a $(3,k)$-tree-model. 
        When $x<y$, let $q = \lfloor (y-x)/2 \rfloor$.
        We use the induction hypothesis to obtain:
        \begin{itemize}
            \item A $(2\log(q - x)+4,k)$-tree model $(T^{<q},\cR^{<q},\cM^{<q},\lambda^{<q})$ for $G^{x,q-1}$.
            \item A $(2\log(y-q)+4,k)$-tree-model $(T^{>q},\cR^{>q},\cM^{>q},\lambda^{>q})$ for $G^{q+1,y}$.
        \end{itemize} 
        Then, we construct a $(4+2\log(y-x+1),k)$-tree-model $(T,\cR,\cM,\lambda)$ of $G^{x,y}$ from the tree-models of $G^{x,q-1}, G^{q+1,y}$, but also the $(3,k)$-tree-model $(T^q,\lambda^q,\cR^q,\cM^q)$ of $G^{q}$.
        To obtain $T$, we create the root $\alpha$ of $T$ and we make it adjacent to $a^{q}$, the root of $T^{q}$, and two new vertices: $\alpha^{<q}$ and $\alpha^{>q}$.
        We make $\alpha^{<q}$ adjacent to the root of $T^{<q}$ and to all the vertices in $\Inf(q-1)$.
        Symmetrically, we make $\alpha^{>q}$ adjacent to the root of $T^{>q}$ and to all the vertices in $\Inf(q)$.
        See \Cref{fig:tree-model-reduction} for an illustration of $T$.
    \end{proof}
    \fi
    \iflong
    \begin{proof}
        Let $L_{\sfS},L_{\sfM},L_{\min},L_{\max},L_{\Inf}$ and $\{\ell_0\}$ be disjoint subsets of $[k]$ such that each set among $L_{\sfS},L_{\sfM},L_{\min},L_{\max},$ has size $2r\log N$ and $L_{\Inf}$ has size $6r\log N -4$. 
        First, we prove that the union of the gadgets associated with a position $q\in [t]$ admits a simple tree-model.
        For every $q\in [t]$, we denote by $G^{q}$ the union of the selection gadgets $\sfS_{p}^{q}$ with $p\in [r]$ and the matching gadgets $\Match(p,q)$ with $p\in [r-1]$.
        
        \begin{claim}\label{claim:reduction}
            For every $q\in [t]$, we can construct in polynomial time a $(3,k)$-tree-model $(T^q,\cM^q,\cR^q,\lambda^q)$ for $G^{q}$.
        \end{claim}
        \begin{claimproof}
            Let $q\in [t]$. 
            We create the root $a^{q}$ of $T^q$ and we attach all the vertices in the selection gadgets $\sfS_{p}^{q}$ with $p\in [r]$ as leaves adjacent to $a^{q}$.
            Then, for every $p\in [r-1]$, we create a node $a^{q}_{p}$ adjacent to $a^q$ and for every $(I,J)\in \IJ_p$, we create a node $a^{q}_{p,I,J}$ adjacent to $a^{q}_{p}$. 
            For each $(I,J)\in\IJ_p$, we make $a^{q}_{p,I,J}$ adjacent to the vertex $v^{q}_{p,I,J}$ and all the vertices in $\sfM^{p,q,I,J}_{p}$, $\sfM^{p,q,I,J}_{p+1}$. 
            Note that all the vertices in $\Match(p,q)$ are the leaves of the subtree rooted at $a^{q}_{p}$, and the leaves of $T^q$ are exactly the vertices in $G^q$. 
            See \Cref{fig:tree-model-reduction} for an illustration of $T^{q}$.
            
            We define $\lambda^q$ as follows:
            \begin{itemize}
                \item $\lambda^q$ maps each vertex in $\sfS_{1}^{q},\dots,\sfS_{r}^{q}$ to a unique label in $L_{\sfS}$.
                
                \item For every $(p,i,x)\in [r]\times [\log N] \times \{0,1\}$, $\lambda^q$ maps all the $x$-endpoints of the $i$-edges from the different copies $\sfM_{p}^{p',q,I,J}$ of $\sfS_{p}^{q}$ to a unique label in $L_{\sfM}$.
                
                \item We have $\lambda^{q}(v_{p,I,J}^{q})=\ell_0$ for every $p\in[r-1]$ and $(I,J)\in \IJ_p$.
            \end{itemize}
            We define $\cR^q=\{\rho_{ab} \mid ab\in E(T^q)\}$ such that $\rho_{ab}$ is the identity function for every $ab\in E(T^q)$.
            It follows that $\lambda^q_a=\lambda^q$ for every node $a$ of $T^q$.
            We finish the construction of the tree-model of $G^q$ by proving that there exists a family of matrices $\cM^q$ such that $(T^q,\cM^q,\cR^q,\lambda^q)$ is a $(3,k)$-tree-model of $G^q$.
            For doing so, we simply prove that the property $\varphi(a,\ell)$ holds for every label $\ell\in [k]$ and internal node $a$ of $T^q$ with children $b_1,\dots,b_c$,  where $\varphi(a,\ell)$ is true if:
            \begin{itemize}
                \item\label{item:property:model} For every $u\in V_{b_i}$ and $v\in V_{b_j}$ with $\lambda^q_a(u)=\lambda^q_a(v)=\ell$, we have $N(u)\cap ( V_a\setminus ( V_{b_i}\cup V_{b_j}))=N(v)\cap ( V_a\setminus ( V_{b_i}\cup V_{b_j}))$.
            \end{itemize}
            Observe that $\varphi(a,\ell)$ trivially holds when there is at most one vertex labeled $\ell$ in $V_a$.
            Consequently, $\varphi(a^{q}_{p,I,J},\ell)$ is true for every node $a^{q}_{p,I,J}$ and $\ell \in [k]$. Moreover, $\varphi(a, \ell)$ is true for every internal node $a$ and every $\ell\in L_{\sfS}\cup L_{\min}\cup L_{\max}\cup L_{\Inf}$.
            Recall that $\lambda^q_a=\lambda^q$ for every node $a$ of $T^q$.
            For every pair $(u,v)$ of distinct vertices in $V(G^q)$, if $\lambda^q(u)=\lambda^q(v)$ then either:
            \begin{itemize}
                \item There exists $(p,i,x)\in [r]\times [\log N]\times \{0,1\}$ such that $u$ and $v$ are the $x$-endpoints of the $i$-edges in respectively $\sfM_{p}^{p^\star,q,I,J}$ and $\sfM_{p}^{p',q,I',J'}$ for some $p^\star,p'\in \{p-1,p\}$, $(I,J)\in \IJ_{p^{\star}}$ and $(I',J')\in \IJ_{p'}$. 
                Observe that the parent of $u$ is $a^{q}_{p^\star,I,J}$ and the neighbors of $u$ in $V_{a^q_{p^\star}}$ are all children of $a^{q}_{p^\star,I,J}$.
                Indeed, the only neighbors of $u$ in $V_{a^q_{p^\star}}$ are the $(1-x)$-endpoint of the $i$-edge of $\sfM_{p}^{p^\star,q,I,J}$ and potentially  $v_{p^\star,I,J}^{q}$.			 
                Symmetrically, $v$ belongs to $V_{a^{q}_{p',I',J'}}$, its parent is $a^{q}_{p',I',J'}$  and the only neighbors of $v$ in $V_{a^q_{p'}}$ are children of $a^{q}_{p',I',J'}$.
                Moreover, observe that $u$ and $v$ have both only one neighbor in $V_{a^q} \setminus V_{a^q_p}$ which is the $(1-x)$-endpoint of the $i$-edge of $\sfS^{q}_{p}$.
                We deduce that $\varphi(a^q,\ell)$ and $\varphi(a^q_{p},\ell)$ and $\varphi(a^q_{p,I,J},\ell)$ are true for every $\ell\in L_{\sfM}$.
                
                \item We have $u=v^{q}_{p,I,J}$ and $v=v^{q}_{p',I',J'}$.
                In that case, $u$ is a child of $a^{q}_{p,I,J}$ and $v$ a child of $a^{q}_{p',I',J'}$.
                The only neighbors of $u$ and $v$ that are not children of $a^{q}_{p,I,J}$ nor $a^{q}_{p',I',J'}$ are all the other vertices of label $\ell_0$.
                Thus, $\varphi(a,\ell_0)$ holds for every internal node $a$ of $T^q$.
            \end{itemize}
            We conclude that $\varphi(a,\ell)$ holds for every internal node $a$ 	of $T^q$ and every $\ell\in [k]$.
            Hence, there exists a family of matrices $\cM^q$ such that $(T^q,\cR^q,\cM^q,\lambda^q)$ is a $(3,k)$-tree model of $G^q$ for every $q\in [t]$.
        \end{claimproof}
        
        \begin{figure}[bht]
            \centering
            \includegraphics[width=0.7\linewidth]{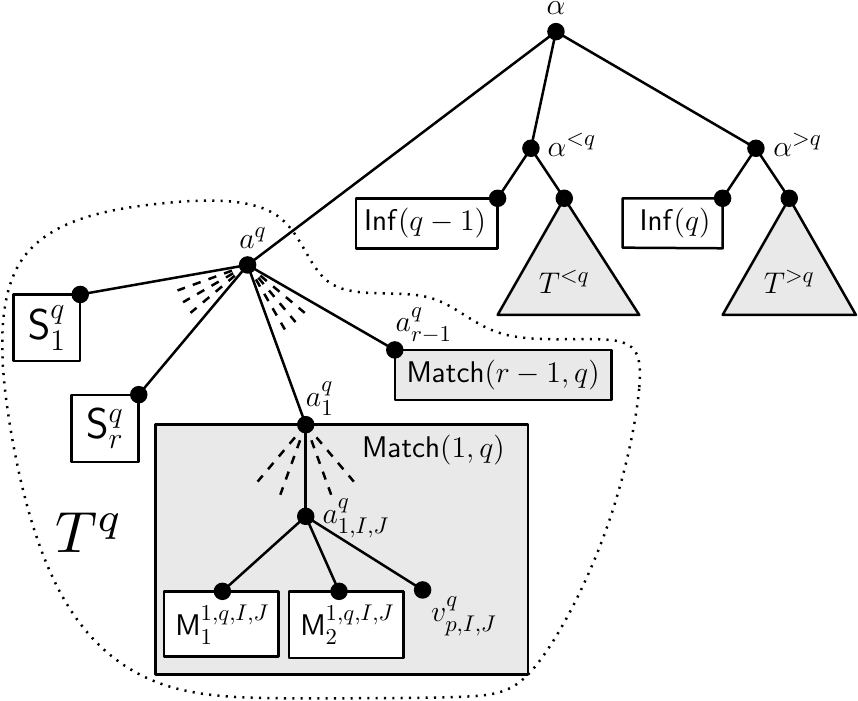}
            \caption{Illustration of the tree $T$ and its subtree $T^{q}$ for the tree-model constructed in \Cref{lem:reduc:shrubdepth}.
                An edge between a white filled rectangle labeled $\mathsf{X}$ and a node $a$ of the tree means that all the vertices in $\mathsf{X}$ are leaves adjacent to $a$. }
            \label{fig:tree-model-reduction}
        \end{figure}
        
        For every $q\in[t-1]$, we denote by $\Inf(q)$ the union of $\Inf(1,q),\dots,\Inf(r,q)$.
        Moreover, for every interval $[x,y]\subseteq  [t]$, we denote by $G^{x,y}$, the union of the graphs $G^q$ over $q\in [x,y]$ and the inferiority gadgets in $\Inf(q)$ over $q\in [x,y]$ such that $q+1\in [x,y]$.
        We prove by induction that for every interval $[x,y]\subseteq  [t]$, there exists a $(2\log(y-x+1)+4,k)$-tree-model $(T,\cR,\cM,\lambda)$ of $G^{x,y}$ such that given the root $\alpha$ of $T$, the following properties are satisfied:
        \begin{enumerate}[(A)]
            \item\label{item:minlabel} $\lambda_{\alpha}$ maps each vertex from $\sfS^{x}_{1},\dots,\sfS^{x}_{r}$ to a unique label in $L_{\min}$.
            \item\label{item:maxlabel} $\lambda_{\alpha}$ maps each vertex from $\sfS^{y}_{1},\dots,\sfS^{y}_{r}$ to a unique label in $L_{\min}$. 
        \end{enumerate}
        When $x=y$, we only require that exactly one property among (\ref{item:minlabel}) and (\ref{item:maxlabel}) is  satisfied (we can choose which one as it is symmetric).
        The induction is on $y-x$. 
        The base case is when $x=y$, in which case $G^{x,y}=G^{x}$ and we simply modify the $(3,k)$-tree-model $(T^x,\lambda^x,\cR^x,\cM^x)$ for $G^{x}$ as follows. 
        We add to $T^{x}$ a new root $\alpha$ adjacent to the former root $a^{x}$.
        We add to $\cR^x$ the function $\rho_{\alpha a^{x}}$ that bijectively maps $L_{\sfS}$ to $L_{\min}$ (or $L_{\max}$ if we want to satisfy (\ref{item:maxlabel}) rather than (\ref{item:minlabel})) and every label not in $L_{\sfS}$ to $\ell_0$. Finally, we add $M_\alpha$ the zero $k\times k$-matrix to $\cM^x$. After these modifications, it is easy to see that $(T^x,\lambda^x,\cR^x,\cM^x)$ is a $(4,k)$-tree model of $G^{x,y}$ that satisfies (\ref{item:minlabel}) or (\ref{item:maxlabel}).        
        
        Now assume that $x< y$ and that $G^{x',y'}$ admits the desired tree-model for every $[x',y']$ strictly included in $[x,y]$.
        Let $q = \lfloor (y-x)/2 \rfloor$.
        By induction hypothesis, there exist:
        \begin{itemize}
            \item A $(2\log(q - x)+4,k)$-tree model $(T^{<q},\cR^{<q},\cM^{<q},\lambda^{<q})$ for $G^{x,q-1}$ with the desired properties (if $x=q -1$, we require (\ref{item:minlabel}) to be satisfied).
            \item A $(2\log(y-q)+4,k)$-tree-model $(T^{>q},\cR^{>q},\cM^{>q},\lambda^{>q})$ for $G^{q+1,y}$ with the desired properties (if $y=q +1$, we require (\ref{item:maxlabel}) to be satisfied).
        \end{itemize} 
        For the sake of legibility, we assume that $x$ is different from $q$, which implies that $G^{x,q-1}$ is not empty graph (note that $G^{q+1,y}$ is not empty as $x<y$ and $q = \lfloor (y-x)/2 \rfloor$). 
        We lose some generality with this assumption, but we can easily deal with the case $x=q$ with some simple modifications on the following construction (i.e. removing some nodes and changing some renaming functions).
        
        In the following, we construct a $(4+2\log(y-x+1),k)$-tree-model $(T,\cR,\cM,\lambda)$ of $G^{x,y}$ from the above tree-models of $G^{x,q-1}, G^{q+1,y}$, but also the $(3,k)$-tree-model $(T^q,\lambda^q,\cR^q,\cM^q)$ of $G^{q}$ given by \Cref{claim:reduction}.
        To obtain $T$, we create the root $\alpha$ of $T$ and we make it adjacent to $a^{q}$ the root of $T^{q}$ and two new vertices: $\alpha^{<q}$ and $\alpha^{>q}$.
        We make $\alpha^{<q}$ adjacent to the root of $T^{<q}$ and to all the vertices in $\Inf(q-1)$.
        Symmetrically, we make $\alpha^{>q}$ adjacent to the root of $T^{>q}$ and to all the vertices in $\Inf(q)$.
        See \Cref{fig:tree-model-reduction} for an illustration of $T$.
        
        We define $\lambda$ as follows:
        \[ \lambda(v)= \begin{cases}
            \lambda^{<q}(v) & \text{ if } v\in V(G^{x,q-1}),\\
            \lambda^{>q}(v) & \text{ if } v\in V(G^{q+1,y}),\\
            \lambda^q(v) & \text{ if } v\in V(G^{q}),\\
            \lambda'(v) &\text{ otherwise (when $v$ belongs to $\Inf(q-1)$ or $\Inf(q))$}
        \end{cases} \]
        where $\lambda'$ maps the vertices in $G^{x,y}$ from $\Inf(q-1)$ and $\Inf(q)$ to $L_{\Inf}$ such that for each label $\ell$ of $L_{\Inf}$, there exists $q'\in \{q-1,q\}$ and $p\in [r]$ such that $\ell$ is associated with either: (1)~$v^{x,p,q'}_{i}$ for some $i\in [\log N-1]$ and $x\in \{0,1\}$ or (2)~all the vertices in $V^{01,p,q'}_{i}$ for some $i\in [\log N]$.
        Since $\abs{L_{\Inf}}= 6r\log( N) -4$, we have enough labels for doing so.
        The family of renaming function $\cR$ is obtained from the union of $\cR^{<q}\cup \cR^{>q} \cup \cR^{q}$ by adding for every edge $e$ in $T$ that is not in $T^{q}, T^{<q}$ or $T^{>q}$ a function $\rho_e$ defined as follows:
        \begin{itemize}
            \item $\rho_e$ is the identify function when $e$ is an edge adjacent to a leaf from $\Inf(q-1)$ or $\Inf(q)$.
            
            \item $\rho_{e}$ maps every label in $L_{\min}\cup L_{\max}$ to itself and every other label to $\ell_0$, when $e$ is the edge between $a^{\circledast q}$ and the root of $T^{\circledast q}$ for $\circledast\in \{<,>\}$.
            
            \item $\rho_e$ maps every label in $L_{\sfS}$ to itself and every other label to $\ell_0$ when $e=\alpha a^{q}$.
            
            \item $\rho_{e}$ maps every label in $L_{\min}\cup L_{\Inf}$ to itself and every other label to $\ell_0$, when $e=\alpha \alpha^{<q}$.
            
            \item $\rho_{e}$ maps every label in $L_{\max}\cup L_{\Inf}$ to itself and every other label to $\ell_0$, when $e=\alpha \alpha^{>q}$.
        \end{itemize}
        Observe that $\lambda_{\alpha}$ satisfies Properties (\ref{item:minlabel}) and (\ref{item:maxlabel}).
        As $\lambda_{\alpha^{<q}}$ satisfies (\ref{item:minlabel}), this function maps every vertex from $\sfS^{x}_{1},\dots,\sfS^{x}_{r}$ to a unique label in $L_{\min}$.
        The above renaming functions guarantee that the only vertices mapped to a label in $L_{\min}$ by $\lambda_{\alpha}$ are from $V_{\alpha^{<q}}$. We deduce that Property~(\ref{item:minlabel}) holds and symmetrically, Property~(\ref{item:maxlabel}) holds too.
        
        Now we prove that a family of matrices $\cM$ exists such that $(T,\cR,\cM,\lambda)$ is a tree model of $G^{x,y}$.
        As before, we prove that $\varphi(a,\ell)$ holds for every internal node $a$ of $T$ and every label $\ell\in [k]$.
        Since our construction is based on tree-models for $G^{q}, G^{x,q-1}$ and $G^{q+1,y}$, we only need to prove that $\varphi(a,\ell)$ holds for every $a\in\{\alpha^{<q},\alpha^{>q},\alpha\}$ and $\ell\in [k]$.
        
        We first deal with $\alpha^{<q}$.
        Let us describe the labeling function $\lambda_{\alpha^{<q}}$.
        Remember that $(T^{<q},\cR^{<q},\cM^{<q},\lambda^{<q})$ satisfies Properties (\ref{item:minlabel}) and (\ref{item:maxlabel}), or just (\ref{item:minlabel})  when $x=q-1$. Moreover, the renaming function associated with the edge between $\alpha^{<q}$ and $T^{<q}$ preserves the labels in $L_{\min}\cup L_{\max}$ and maps the other labels to $\ell_0$.
        Thus, $\lambda_{\alpha^{<q}}$ assign every vertex in $\sfS^{x}_{1},\sfS^{q-1}_{1},\dots,\sfS^{x}_{r},\sfS^{q-1}_{r}$ to a unique label in $L_{\min}\cup L_{\max}$.
        We deduce that for every a pair $(u,v)$ of distinct vertices in $V_{\alpha^{<q}}$, if $\lambda_{\alpha^{<q}}$ assigns $u$ and $v$ to the same label $\ell\in [k]$, then either:
        \begin{itemize}
            \item $u,v\in V^{01,p,q-1}_{i}$ for some $p\in [r]$ and $i\in [\log N]$. In this case, $u$ and $v$ are false twins by construction of $\Inf(p,q-1)$---i.e. $N(u)=N(v)$---and we deduce that $\varphi(\alpha^{<q}, \ell)$ holds.
            
            \item $\ell=\ell_0$ and $u,v$ are in $V(G^{x,q-1})$ and not from $\sfS^{x}_{1},\sfS^{q-1}_{1},\dots,\sfS^{x}_{r},\sfS^{q-1}_{r}$.
            Then, all the neighbors of $u$ and $v$ are in $G^{x,q-1}$ and thus $N(u)\setminus V(G^{x,q-1}) = N(v) \setminus  V(G^{x,q-1})$. We deduce that $\varphi(\alpha^{<q}, \ell)$ holds in this case too.
        \end{itemize}
        We conclude that $\varphi(\alpha^{<q},\ell)$ holds for every $\ell\in [k]$ and with symmetric arguments, we can prove that $\varphi(\alpha^{>q},\ell)$ holds also for every $\ell\in [k]$.
        
        For $\alpha$, notice that for every $a\in\{a^{q},\alpha^{<q},\alpha^{>q}\}$, the vertices in $V_a$ labeled $\ell_0$ by $\lambda_{\alpha}$ have neighbors only in $V_{a}$, hence $\varphi(a,\ell_0)$ holds.
        Furthermore, every  label $\ell$ in $L_{\min}\cup L_{\max} \cup L_{\sfS}$ is mapped by $\lambda_{\alpha}$ to a unique vertex in $V_{\alpha}$, so $\varphi(\alpha,\ell)$ holds. 
        Finally, each label in $L_{\Inf}$ is mapped by $\lambda_{\alpha}$ to a unique vertex or to all the vertices in $V^{01,p,q'}_{i}$ for some $p\in[r], q'\in \{q-1,q\}$ and $i\in [\log N]$. Since the vertices in $V^{01,p,q'}_{i}$ are false twins, we deduce that $\varphi(\alpha,\ell)$ holds for every $\ell\in L_{\Inf}$.
        We conclude that $\varphi(\alpha,\ell)$ holds for every $\ell\in [k]$ and thus there exists a family $\cM$ of matrices such that $(T,\cM,\cR,\lambda)$ is a tree-model of $G^{x,y}$.
        
        It remains to prove that the depth of $T$ is at most $d=2\log(y-x+1)+4$.
        By definition of $q$, both $q-x$ and $y-q$ are smaller than $(y-x+1)/2$.
        Thus, $\log(q-x)$ and $\log(y-q)$ are smaller than $\log(y-x+1)-1$.
        Now observe that the depth of $T$ is the maximum between (i)~the depth of $T^{q}$ plus 1 which is 4, (ii)~the depth of $T^{<q}$ plus 2, and (iii)~the depth of $T^{>q}$ plus 2.
        The depth of $T^{<q}$ is at most $2\log(q-x)+4$.
        Since $\log(q-x)\leq \log(y-x+1)-1$, the depth of $T^{<q}$ plus 2 is at most $2\log(y-x+1)+4$.
        Symmetrically, the depth of $T^{>q}$ plus 2 is also upper bounded by $2\log(y-x+1)+4$.
        It follows that the depth of $T$ is at most $2\log(y-x+1)+4$. 
        
        We conclude that for every interval $[x,y]$, $G^{x,y}$ admits a $(2\log(y-x+1)+4,k)$-tree-model.
        In particular, it implies that $G^{1,t}=G$ admits a $(d,k)$-tree-model.
        It is easy to see from our proof that this $(d,k)$-tree-model is computable in polynomial time.
    \end{proof}

    We are now ready to prove \Cref{thm:lowerbound}.
    
    \begin{proof}[Proof of \Cref{thm:lowerbound}]
        Let $\delta$ be an unbounded and computable function.
        Assume towards a contradiction that there exists an algorithm $\cA$ that solves the {\sc{Independent Set}} problem in graphs supplied with $(d,k)$-tree-models satisfying $d\leq \delta(k)$ that runs in time $2^{\Oh(k)}\cdot n^{\Oh(1)}$ and uses $n^{\Oh(1)}$ space.
        
        Since $\delta$ is unbounded and computable and $\log$ is monotone, there exists an unbounded and computable function $\delta'$ such that for all sufficiently large $N,r\in \bN$, we have \[2\log(\delta'(N))+4 \leq \delta( 14r\log N-3).\]
        Let $(N,t,\Sigma,s_1,\dots,s_r)$ be an instance of \textsc{LCS} such that $t\leq \delta'(N)$.
        Our reduction provides us with a graph $G$ and an integer $\obj$ such that the following holds:
        \begin{itemize}
            \item $G$ has $\Oh(rtN^2\log N)$ vertices and thus it can be constructed in $M^{\Oh(1)}$ time where $M$ is the total bitsize of $(N,t,\Sigma,s_1,\dots,s_r)$.
            Indeed, the selection gadgets are made of $2rt\log N$ vertices, the inferiority gadgets have exactly $r(t-1)(2\log N + \log N(1+\log N)/2)$ vertices and the matching gadgets consist of $\sum_{p\in [r-1]} t \cdot 2\abs{\IJ_{p}} \cdot(1 + 2\log N)$ vertices.
            \item By \Cref{lem:reduc:correctness}, $G$ admits an independent set of size at least $\obj$ iff $s_1,\dots,s_r$ admits a common subsequence of size $t$.
            \item By \Cref{lem:reduc:shrubdepth}, we can construct in polynomial time a $(d,k)$-tree-model of $G$ with $d=2\log t + 4$ and $k=14r\log N-3$.
        \end{itemize}
        Observe that we have 
        \[ d=2\log t+4 \leq 2\log(\delta'(N))+4 \leq \delta( 14r\log N-3 )=\delta(k). \]
        Consequently, we can run $\cA$ to check whether $G$ admits an independent set of size at least $\obj$ in time $2^{\Oh(k)}\cdot n^{\Oh(1)}$ and space $n^{\Oh(1)}$.
        Since $k=14r\log N-3$ and $n=\Oh(rtN^2\log N)\leq M^{\Oh(1)}$, it follows that we can solve $(N,t,\Sigma,s_1,\dots,s_r)$ in time $N^{\Oh(r)}\cdot M^{\Oh(1)} \leq M^{\Oh(r)}$ and space $M^{\Oh(1)}$.
        As this can be done for every instance $(N,t,\Sigma,s_1,\dots,s_r)$ where $t\leq \delta'(N)$, it contradicts \Cref{conj:main}.
    \end{proof}
    \fi

\section{Fixed-Parameter Algorithms for Metric Dimension and Firefighting}

\ifshort
Theorem~\ref{thm:FPTproblems}---and in particular the fixed-parameter tractability of \textsc{Metric Dimension} and \textsc{Firefighter} parameterized by shrub-depth---can be obtained by combining known results about these problems~\cite{BazganCCFFL14,GimaHKKO22} with a bound on the maximum length of induced paths in graph classes of bounded shrubdepth~\cite[Theorem~3.7]{GanianHNOM19}.
These results contrast the \NP-hardness of both problems on graphs of bounded pathwidth~\cite{ChlebikovaC17,LiP22}.
\fi

\iflong
    The {\sc Firefighter} problem on a graph $G$ is the following. At time $0$, a vertex $r \in V(G)$ catches fire. Then at each time step $i \geq 1$, first a firefighter is permanently placed on a vertex that is not currently on fire. This vertex is now permanently \emph{protected}. Then the fire spreads to all unprotected neighbors of all vertices currently on fire. This process ends in the time step when the fire no longer spreads to new vertices. All vertices that do not catch fire during this process (including the protected vertices) are called \emph{saved}; the rest are called \emph{burned}. The goal is to maximize the number of saved vertices.
    
    Bazgan et al.~\cite{BazganCCFFL14} showed that the {\sc Firefighter} problem is fixed-parameter tractable when parameterized by the treewidth of the input graph and the number $k$ of vertices that may be protected during the process. In this result, one first writes an $\mathsf{MSO}_2$ formula $\varphi(X)$ that expresses that a set of vertices $X$ can be saved assuming that $k$ vertices can be protected, and then applies the optimization variant of Courcelle's Theorem, due to Arnborg et al.~\cite{ArnborgLS91}, to find the largest vertex subset $A$ for which $\varphi(A)$ is satisfied. By inspection of the formula, we can see that it does not quantify over edge sets, hence $\varphi(X)$ is in fact an $\mathsf{MSO}_1$ formula. Then, by replacing the usage of the algorithm of Arnborg et al.\ with the algorithm of Courcelle et al.~\cite{CourcelleMR00}, we conclude that the {\sc Firefighter} problem is fixed-parameter tractable when  parameterized by the cliquewidth of the input graph and the number of vertices that may be protected.

    We now recall that in graphs with a $(d,k)$-tree model, any induced path has length at most $\Oh(2^{k^{d+1}})$ (this follows from~\cite[Theorem~3.7]{GanianHNOM19}; the bound accounts for our slightly different definition of a tree model). This implies that the firefighting game has at most $\Oh(2^{k^{d+1}})$ time steps and the same amount of vertices can be protected. 
    Hence, recalling that a graph with a $(d,k)$-tree model has bounded cliquewidth~\cite[Proposition~3.4]{GanianHNOM19}, we immediately obtain the following result.
    
    \begin{theorem}
    The {\sc Firefighter} problem is fixed-parameter tractable when parameterized by $d$ and $k$ on graphs provided with a $(d, k)$-tree-model.
    \end{theorem}

    We observe that this is in contrast to the complexity of the {\sc Firefighter} problem on graphs of bounded treewidth. The {\sc Firefighter} problem is in fact already NP-hard on trees of maximum degree~$3$ (which are graphs of treewidth~$1$) \cite{FinbowKMR07} and trees of pathwidth~$3$~\cite{ChlebikovaC17}.

    A similar situation arises for the {\sc Metric Dimension} problem. In {\sc Metric Dimension}, given a graph $G$, we are asked to find a smallest set $Z \subseteq V(G)$ such that for any pair $u,v \in V(G)$, there is a vertex $z \in Z$ such that the distance between $u$ and $z$ and the distance between $v$ and $z$ are distinct. Gima et al.~\cite{GimaHKKO22} observed that {\sc Metric Dimension} is fixed-parameter tractable when parameterized by the cliquewidth and the diameter of the input.
    Since in graphs with a $(d,k)$-tree model, any induced path has length at most $\Oh(2^{k^{d+1}})$ (the bound accounts for our slightly different definition of a tree model), and any such graph has bounded cliquewidth~\cite{GanianHNOM19}, we immediately obtain the~following.

    \begin{theorem}
    The {\sc Metric Dimension} problem is fixed-parameter tractable when parameterized by $d$ and $k$ on graphs provided with a $(d, k)$-tree-model.
    \end{theorem}

    This is again in contrast to the complexity of the {\sc Metric Dimension} problem on graphs of bounded treewidth. The {\sc Metric Dimension} problem is in fact already $\mathsf{NP}$-hard on graphs of pathwidth~$24$ \cite{LiP22}.
    \fi

\bibliographystyle{plain}
\bibliography{biblio}

\begin{thebibliography}{10}

\bibitem{AbboudBW15}
Amir Abboud, Arturs Backurs, and Virginia~Vassilevska Williams.
\newblock Tight hardness results for {LCS} and other sequence similarity
  measures.
\newblock In {\em Proc. FOCS 2015}, pages 59--78, 2015.

\bibitem{AllenderCLPT14}
Eric Allender, Shiteng Chen, Tiancheng Lou, Periklis~A. Papakonstantinou, and
  Bangsheng Tang.
\newblock Width-parametrized {SAT:} {T}ime--space tradeoffs.
\newblock {\em Theory Comput.}, 10(12):297--339, 2014.

\bibitem{ArnborgLS91}
Stefan Arnborg, Jens Lagergren, and Detlef Seese.
\newblock Easy problems for tree-decomposable graphs.
\newblock {\em J. Algorithms}, 12(2):308--340, 1991.

\bibitem{BarskySTU07}
Marina Barsky, Ulrike Stege, Alex Thomo, and Chris Upton.
\newblock Shortest path approaches for the longest common subsequence of a set
  of strings.
\newblock In {\em Proc. {BIBE} 2007}, pages 327--333, 2007.

\bibitem{BazganCCFFL14}
Cristina Bazgan, Morgan Chopin, Marek Cygan, Michael~R. Fellows, Fedor~V.
  Fomin, and Erik~Jan van Leeuwen.
\newblock Parameterized complexity of firefighting.
\newblock {\em J. Comput. System Sci.}, 80(7):1285--1297, 2014.

\bibitem{BjorklundHKK07}
Andreas Bj\"{o}rklund, Thore Husfeldt, Petteri Kaski, and Mikko Koivisto.
\newblock Fourier meets {M}{\"o}bius: fast subset convolution.
\newblock In {\em Proc. STOC 2007}, pages 67--74, 2007.

\bibitem{BodlaenderDFW95}
Hans~L. Bodlaender, Rodney~G. Downey, Michael~R. Fellows, and Harold~T.
  Wareham.
\newblock The parameterized complexity of sequence alignment and consensus.
\newblock {\em Theoret. Comput. Sci.}, 147(1{\&}2):31--54, 1995.

\bibitem{BodlaenderGJJL22}
Hans~L. Bodlaender, Carla Groenland, Hugo Jacob, Lars Jaffke, and Paloma~T.
  Lima.
\newblock Xnlp-completeness for parameterized problems on graphs with a linear
  structure.
\newblock In {\em 17th International Symposium on Parameterized and Exact
  Computation, {IPEC} 2022}, volume 249 of {\em LIPIcs}, pages 8:1--8:18.
  Schloss Dagstuhl --- Leibniz-Zentrum f{\"{u}}r Informatik, 2022.

\bibitem{BodlaenderGJPP22}
Hans~L. Bodlaender, Carla Groenland, Hugo Jacob, Marcin Pilipczuk, and
  Micha\l{} Pilipczuk.
\newblock On the complexity of problems on tree-structured graphs.
\newblock In {\em 17th International Symposium on Parameterized and Exact
  Computation, {IPEC} 2022}, volume 249 of {\em LIPIcs}, pages 6:1--6:17.
  Schloss Dagstuhl --- Leibniz-Zentrum f{\"{u}}r Informatik, 2022.

\bibitem{BodlaenderGNS21}
Hans~L. Bodlaender, Carla Groenland, Jesper Nederlof, and C{\'{e}}line M.~F.
  Swennenhuis.
\newblock Parameterized problems complete for nondeterministic {FPT} time and
  logarithmic space.
\newblock In {\em 62nd {IEEE} Annual Symposium on Foundations of Computer
  Science, {FOCS} 2021}, pages 193--204. {IEEE}, 2021.

\bibitem{BodlaenderCP22}
Hans~L. Bodlaender, Carla Groenland, and Micha\l{} Pilipczuk.
\newblock Parameterized complexity of binary {CSP}: Vertex cover, treedepth,
  and related parameters.
\newblock {\em CoRR}, abs/2208.12543, 2022.

\bibitem{ChlebikovaC17}
Janka Chleb{\'{\i}}kov{\'{a}} and Morgan Chopin.
\newblock The firefighter problem: further steps in understanding its
  complexity.
\newblock {\em Theoret. Comput. Sci.}, 676:42--51, 2017.

\bibitem{CourcelleMR00}
Bruno Courcelle, Johann~A. Makowsky, and Udi Rotics.
\newblock Linear time solvable optimization problems on graphs of bounded
  clique-width.
\newblock {\em Theory Comput. Syst.}, 33(2):125--150, 2000.

\bibitem{cygan2015parameterized}
Marek Cygan, Fedor~V Fomin, {\L}ukasz Kowalik, Daniel Lokshtanov, D{\'a}niel
  Marx, Marcin Pilipczuk, Micha{\l} Pilipczuk, and Saket Saurabh.
\newblock {\em Parameterized Algorithms}.
\newblock Springer, 2015.

\bibitem{CyganNPPRW22}
Marek Cygan, Jesper Nederlof, Marcin Pilipczuk, Micha{\l} Pilipczuk, Johan
  M.~M. van Rooij, and Jakub~Onufry Wojtaszczyk.
\newblock Solving connectivity problems parameterized by treewidth in single
  exponential time.
\newblock {\em {ACM} Trans. Algorithms}, 18(2):17:1--17:31, 2022.

\bibitem{DeVosKO20}
Matt DeVos, O{-}joung Kwon, and Sang{-}il Oum.
\newblock Branch-depth: {G}eneralizing tree-depth of graphs.
\newblock {\em European J. Combin.}, 90:Article 103186, 2020.

\bibitem{Diestel12}
Reinhard Diestel.
\newblock {\em Graph Theory}, volume 173 of {\em Graduate texts in
  mathematics}.
\newblock Springer, 4th edition, 2012.

\bibitem{DowneyF2013}
Rodney~G. Downey and Michael~R. Fellows.
\newblock {\em Fundamentals of Parameterized Complexity}.
\newblock Texts in Computer Science. Springer, 2013.

\bibitem{Dreier21}
Jan Dreier.
\newblock Lacon- and shrub-decompositions: {A} new characterization of
  first-order transductions of bounded expansion classes.
\newblock In {\em Proc. LICS 2021}, pages 1--13, 2021.

\bibitem{DreierGKPT22}
Jan Dreier, Jakub Gajarsk{\'{y}}, Sandra Kiefer, Micha{\l} Pilipczuk, and
  Szymon Toru{\'n}czyk.
\newblock Treelike decompositions for transductions of sparse graphs.
\newblock In {\em Proc. LICS 2022}, pages 31:1--31:14, 2022.

\bibitem{ElberfeldST15}
Michael Elberfeld, Christoph Stockhusen, and Till Tantau.
\newblock On the space and circuit complexity of parameterized problems:
  Classes and completeness.
\newblock {\em Algorithmica}, 71(3):661--701, 2015.

\bibitem{FinbowKMR07}
Stephen Finbow, Andrew~D. King, Gary MacGillivray, and Romeo Rizzi.
\newblock The firefighter problem for graphs of maximum degree three.
\newblock {\em Discret. Math.}, 307(16):2094--2105, 2007.

\bibitem{FominGLS10}
Fedor~V. Fomin, Petr~A. Golovach, Daniel Lokshtanov, and Saket Saurabh.
\newblock Intractability of clique-width parameterizations.
\newblock {\em {SIAM} J. Comput.}, 39(5):1941--1956, 2010.

\bibitem{FominGLS14}
Fedor~V. Fomin, Petr~A. Golovach, Daniel Lokshtanov, and Saket Saurabh.
\newblock Almost optimal lower bounds for problems parameterized by
  clique-width.
\newblock {\em {SIAM} J. Comput.}, 43(5):1541--1563, 2014.

\bibitem{FominK22}
Fedor~V. Fomin and Tuukka Korhonen.
\newblock Fast {FPT}-approximation of branchwidth.
\newblock In {\em Proc. STOC 2022}, pages 886--899, 2022.

\bibitem{Furer17}
Martin F{\"{u}}rer.
\newblock Multi-clique-width.
\newblock In {\em Proc. ITCS 2017}, volume~67 of {\em Leibniz Int. Proc.
  Inform.}, pages 14:1--14:13, 2017.

\bibitem{FurerY17}
Martin F{\"{u}}rer and Huiwen Yu.
\newblock Space saving by dynamic algebraization based on tree-depth.
\newblock {\em Theory Comput. Syst.}, 61(2):283--304, 2017.

\bibitem{GajarskyK20}
Jakub Gajarsk{\'{y}} and Stephan Kreutzer.
\newblock Computing shrub-depth decompositions.
\newblock In {\em Proc. STACS 2020}, volume 154 of {\em Leibniz Int. Proc.
  Inform.}, pages 56:1--56:17, 2020.

\bibitem{GajarskyKNMPST20}
Jakub Gajarsk{\'{y}}, Stephan Kreutzer, Jaroslav Ne{\v{s}}et{\v{r}}il, Patrice
  Ossona~de Mendez, Micha{\l} Pilipczuk, Sebastian Siebertz, and Szymon
  Toru{\'{n}}czyk.
\newblock First-order interpretations of bounded expansion classes.
\newblock {\em ACM Trans. Comput. Log.}, 21(4):Art. 29, 41, 2020.

\bibitem{GanianHNOM19}
Robert Ganian, Petr Hlin{\v{e}}n{{\'y}}, Jaroslav Ne{\v{s}}et{\v{r}}il, Jan
  Obdr\v{z}\'{a}lek, and Patrice Ossona~de Mendez.
\newblock Shrub-depth: Capturing height of dense graphs.
\newblock {\em Log. Methods Comput. Sci.}, 15(1):7:1--7:25, 2019.

\bibitem{GimaHKKO22}
Tatsuya Gima, Tesshu Hanaka, Masashi Kiyomi, Yasuaki Kobayashi, and Yota
  Otachi.
\newblock Exploring the gap between treedepth and vertex cover through vertex
  integrity.
\newblock {\em Theoret. Comput. Sci.}, 918:60--76, 2022.

\bibitem{Guillemot11}
Sylvain Guillemot.
\newblock Parameterized complexity and approximability of the longest
  compatible sequence problem.
\newblock {\em Discret. Optim.}, 8(1):50--60, 2011.

\bibitem{HegerfeldK20}
Falko Hegerfeld and Stefan Kratsch.
\newblock Solving connectivity problems parameterized by treedepth in
  single-exponential time and polynomial space.
\newblock In {\em Proc. STACS 2020}, volume 154 of {\em Leibniz Int. Proc.
  Inform.}, 2020.

\bibitem{HegerfeldK23}
Falko Hegerfeld and Stefan Kratsch.
\newblock Tight algorithms for connectivity problems parameterized by
  clique-width.
\newblock Technical report, 2023.
\newblock \url{https://arxiv.org/abs/2302.03627}.

\bibitem{ImpagliazzoPZ01}
Russell Impagliazzo, Ramamohan Paturi, and Francis Zane.
\newblock Which problems have strongly exponential complexity?
\newblock {\em J. Comput. Syst. Sci.}, 63(4):512--530, 2001.

\bibitem{Kane10}
Daniel~M. Kane.
\newblock Unary subset-sum is in logspace.
\newblock Technical report, 2010.
\newblock \url{https://arxiv.org/abs/1012.1336}.

\bibitem{LiP22}
Shaohua Li and Marcin Pilipczuk.
\newblock Hardness of metric dimension in graphs of constant treewidth.
\newblock {\em Algorithmica}, 84(11):3110--3155, 2022.

\bibitem{LokshtanovMS11}
Daniel Lokshtanov, Matthias Mnich, and Saket Saurabh.
\newblock Planar $k$-path in subexponential time and polynomial space.
\newblock In {\em Proc. WG 2011}, volume 6986 of {\em Lecture Notes Comput.
  Sci.}, pages 262--270, 2011.

\bibitem{MulmuleyVV87}
Ketan Mulmuley, Umesh~V. Vazirani, and Vijay~V. Vazirani.
\newblock Matching is as easy as matrix inversion.
\newblock {\em Combinatorica}, 7(1):105--113, 1987.

\bibitem{NadaraPS22}
Wojciech Nadara, Micha{\l} Pilipczuk, and Marcin Smulewicz.
\newblock Computing treedepth in polynomial space and linear {FPT} time.
\newblock In {\em Proc. ESA 2022}, volume 244 of {\em Leibniz Int. Proc.
  Inform.}, pages 79:1--79:14, 2022.

\bibitem{NederlofPSW20}
Jesper Nederlof, Micha\l{} Pilipczuk, C{\'{e}}line M.~F. Swennenhuis, and Karol
  W\k{e}grzycki.
\newblock Hamiltonian cycle parameterized by treedepth in single exponential
  time and polynomial space.
\newblock In {\em Proc. {WG} 2020}, volume 12301 of {\em Lecture Notes Comput.
  Sci.}, pages 27--39, 2020.

\bibitem{OhlmannPPT23}
Pierre Ohlmann, Micha{\l} Pilipczuk, Wojciech Przybyszewski, and Szymon
  Toru{\'n}czyk.
\newblock Canonical decompositions in monadically stable and bounded shrubdepth
  graph classes.
\newblock Technical report, 2023.
\newblock \url{https://arxiv.org/abs/2303.01473}.

\bibitem{OSSONADEMENDEZ2022103660}
Patrice {Ossona de Mendez}, Micha\l{} Pilipczuk, and Sebastian Siebertz.
\newblock Transducing paths in graph classes with unbounded shrubdepth.
\newblock {\em European J. Combin.}, page 103660, 2022.

\bibitem{Pietrzak03}
Krzysztof Pietrzak.
\newblock On the parameterized complexity of the fixed alphabet shortest common
  supersequence and longest common subsequence problems.
\newblock {\em J. Comput. Syst. Sci.}, 67(4):757--771, 2003.

\bibitem{PilipczukW18}
Micha{\l} Pilipczuk and Marcin Wrochna.
\newblock On space efficiency of algorithms working on structural
  decompositions of graphs.
\newblock {\em {ACM} Trans. Comput. Theory}, 9(4):18:1--18:36, 2018.

\bibitem{Wanke94}
Egon Wanke.
\newblock $k$-{NLC} graphs and polynomial algorithms.
\newblock {\em Discret. Appl. Math.}, 54(2-3):251--266, 1994.

\end{thebibliography}
 
\end{document}